\documentclass[letterpaper,11pt]{article}
\pdfoutput=1

\input{tcs.sty}
\newif\ifstocproceedings
\stocproceedingsfalse

\title{Learning the closest product state}
\author{%
    Ainesh Bakshi${}^\ast$
    \hspace{15pt} John Bostanci${}^\dagger$
    \hspace{15pt} William Kretschmer${}^\flat$
    \hspace{15pt} Zeph Landau${}^\flat$ \\
    Jerry Li${}^\natural$
    \hspace{15pt} Allen Liu${}^\ast$
    \hspace{15pt} Ryan O'Donnell${}^\sharp$
    \hspace{15pt} Ewin Tang${}^\flat$ \\
\vspace{2em}
\textit{${}^\ast$MIT,
${}^\dagger$Columbia,
${}^\flat$UC Berkeley,
${}^\natural$University of Washington, 
${}^\sharp$Carnegie Mellon University}
}

\date{}

\begin{document}

\maketitle

\begin{abstract}
    We study the problem of finding a (pure) product state with optimal fidelity to an unknown $n$-qubit quantum state $\rho$, given copies of $\rho$.
    This is a basic instance of a fundamental question in quantum learning: is it possible to efficiently learn a simple approximation to an arbitrary state?
    We give an algorithm which finds a product state with fidelity $\varepsilon$-close to optimal, using $N = n^{\text{poly}(1/\varepsilon)}$ copies of $\rho$ and $\text{poly}(N)$ classical overhead.
    We further show that estimating the optimal fidelity is NP-hard for error $\varepsilon = 1/\text{poly}(n)$, showing that the error dependence cannot be significantly improved.
    
    For our algorithm, we build a carefully-defined cover over candidate product states, qubit by qubit, and then demonstrate that extending the cover can be reduced to approximate constrained polynomial optimization.
    For our proof of hardness, we give a formal reduction from polynomial optimization to finding the closest product state.
    Together, these results demonstrate a fundamental connection between these two seemingly unrelated questions.
    %We give a novel iterative algorithm that builds a cover over all possible candidate product states, qubit by qubit, and we then demonstrate that the task of extending the cover can be reduced to approximately solving a constrained polynomial optimization problem.
    %We then complement this with a formal reduction from finding the closest product state to solving generic polynomial optimization problems; together, these results demonstrate a fundamental connection between these two seemingly unrelated questions.
    Building on our general approach, we also develop more efficient algorithms in three simpler settings: when the optimal fidelity exceeds $5/6$; when we restrict ourselves to a discrete class of product states; and when we are allowed to output a matrix product state.

    % scratch
    %Formally, let $\OPT = \max_{\ket{\pi}} \bra{\pi} \rho \ket{\pi}$ be the largest fidelity achievable by a pure, fully product state $\ket{\pi} = \ket{\pi_1} \otimes \ldots \otimes \ket{\pi_n}$.
    %We give an algorithm which outputs a product state with fidelity at least $\OPT - \eps$, using $N = n^{\poly(1/\eps)}$ copies of $\rho$ and $\poly(N)$ classical overhead.
    %By comparison, naive algorithms for this task like randomized Clifford measurements are sample-efficient but require exponential classical overhead.
    %We have two key ideas.
    %First, we build up good product states qubit by qubit, which allows us to work within tame regions of the otherwise notoriously flat and non-convex energy landscape of the fidelity function.
    %we draw a connection between finding good product states and tensor optimization: while it is not even clear how to formulate the product state problem in terms of polynomials, we construct low-degree approximations and search over a careful choice of nets to re-parameterize the problem as a tractable polynomial system.
\end{abstract}

\ewin{TODO: update (Zeph added acknowledgements)}
\ewin{TODO: possibly discuss the "crypto hardness" implications or the "algorithmic" perspective on preparing states with good overlap with product states.}

%TODOs:
%\begin{itemize}
    %\item finalize MPS stuff
    %\item decide what to do about qudits
    %\item Add remarks about the conversion from $\eta$ to $\OPT$ in the text
    %\item Check algorithm statements about ``given''
    %\item Take pass on tech overview?
    %\item Remark to make in intro: given a good product state we can find $\OPT$
    %\item At some point: make a remark about bit precision, and notions of complexity
    %\item figure out the polynomial in the exponent
%\end{itemize}

\thispagestyle{empty}
\clearpage
\newpage

%\microtypesetup{protrusion=false}
\tableofcontents
\thispagestyle{empty}
%\microtypesetup{protrusion=true}
\clearpage
\setcounter{page}{1}

%!TEX root = main.tex
\section{Introduction}
\label{sec:intro}

%\subsection{alternate intro attempt}
When can we obtain a classical description of a complex quantum system?
This problem, at the heart of quantum information theory, is one commonly faced by experimentalists: when we have a large, intricate quantum device, how can we tell what it is doing?
Due to the exponentiality inherent to quantum mechanics, a generic system of $n$ particles is described by a number of parameters scaling exponentially with $n$, so in general, an efficient description simply does not exist.
However, real-world systems are not generic: the physics governing the device will suggest a corresponding model of the system, giving us a hint for how the system can be efficiently described. 
%\ainesh{what is the point of this last sentence? seems a bit vague}
%\will{I think the point is to motivate why concept classes are a thing}

Simultaneously, in real-world applications, the state which one is learning may not---and typically will not---exactly fall within a given model class, due to noise or other forms of imprecision in how our model represents the real world.
In light of this, the natural question is to seek the best approximation to the underlying state within the prescribed model.
Such an approximation can serve as a far more tractable proxy for the true state when it is complex to describe exactly.
In this work, we consider the problem for the class of product states, arguably the most fundamental class of states to consider.
Stated plainly, the question we ask is the following:
\begin{quote}
    \begin{center}
        {\em Can we efficiently learn the best product state approximation to any given state? }
    \end{center}
\end{quote}
\noindent
We formalize this problem as follows:
\begin{problem}[Learning the closest product state]
\label{question:main}
    Consider the set of $n$-qubit product states $\mathcal{P} = \braces{\ket{\pi_1} \otimes \dots \otimes \ket{\pi_n}}$, and let $\epsilon, \delta > 0$ be error parameters.
    Given $N$ copies of an arbitrary $n$-qubit state with density matrix $\rho$, output a classical description of a state $\ket{\pi} \in \mathcal{P}$ such that, with probability $\geq 1-\delta$,
    \[
    \bra{\pi} \rho \ket{\pi} \geq \OPT - \eps\,,
    \qquad \text{where } \OPT = \max_{\ket{\pi} \in \mathcal{P}} \bra{\pi} \rho \ket{\pi}\,.
    \]
%    \ewin{the main ``non-natural'' choice here is to choose fidelity; maybe we should say something about that.}
\end{problem}
Product states are natural to study in this context for a number of reasons.
Because of the locality inherent in physical systems, we commonly model physical systems with states exhibiting low entanglement.
Chief among them are \emph{mean-field theories}, which model systems as states which exhibit \emph{zero entanglement}, e.g., product states~\cite{bh16}.
The mean-field approximation plays a central role in domains relevant to quantum computing: in particular, in quantum chemistry, mean-field theories like Hartree-Fock theory and density functional theory are the standard algorithmic workhorses for understanding chemical processes~\cite{chan24}.
In light of this, we can rephrase \cref{question:main} as asking for the best (pure) mean-field approximation to an arbitrary quantum state, and for the quality of that approximation.
From this perspective, we believe that obtaining an efficient algorithm for \cref{question:main} will have important implications both for validating the effectiveness of these theories and for understanding their properties in real-world settings.
%Further, such an approximation yields interesting information even when $\OPT$ is much smaller than $1$ i.e. when $\rho$ is only weakly correlated with a product state.

As an example application, physicists already run computations to solve \cref{question:main} in the setting where the input is not a quantum state, but a description of a condensed matter system.
Collective entanglement of a multipartite state is often measured by $\OPT$, the best fidelity of the state with a product state, also known as the \emph{geometric measure of entanglement}~\cite{wg03}.
Since its introduction in 2003, this entanglement measure has been used to understand a variety of condensed matter systems (see related work).
An algorithm for \cref{question:main} can be used to compute the geometric measure of entanglement for states which are efficiently preparable on a quantum computer, by preparing copies of the state and running the algorithm to estimate $\OPT$, giving an advantage when such states are classically intractable.

Despite the apparent simplicity of the problem, relatively little was known about the computational complexity of~\cref{question:main}.
From a statistical point of view, one can obtain sample-efficient learners via classical shadow estimation~\cite{hkp20} or shadow tomography~\cite{aaronson20}, but these estimators require exponential runtime.
On the other hand, efficient algorithms were known only for highly restricted versions of the problem~\cite{gikl24}.
This lack of efficient algorithms might be surprising, as when when the unknown state $\rho$ is a product state, i.e.\ $\OPT = 1$, this task is easy: many algorithms work, including learning every register separately.
However, these algorithms are brittle, and fail catastrophically when $\OPT < 0.99$.
Even algorithms for the related problem of product state testing, initiated by the important work of Harrow and Montanaro~\cite{hm13}, do not admit estimates of $\OPT$ when $\OPT$ is bounded away from 1.
In contrast, one would hope to obtain efficient algorithms even when $\OPT$ is a small constant (say, $0.1$): product states with constant fidelity are still great approximations, considering that almost all product states will have fidelity exponentially small in $n$.
%While this setting turns out to be quite challenging algorithmically, it is also arguably the most natural, and it is also crucial to the application to product state testing.

Beyond specific applications, we hope that understanding this algorithmic task can shed light on a broader program in quantum learning theory.
An emerging line of work has been studying ``learning the closest state in a hypothesis class'', also known as \emph{agnostic tomography}: formally, this problem is \cref{question:main}, except the class of product states $\mathcal{P}$ is replaced with a different hypothesis class $\mathcal{C}$.
%\ainesh{how do we feel about leading with this program in the intro? it seems to make our lives easier }
Product states appear as a special case of several well-studied classes of quantum states, including states described by low-depth quantum circuits, matrix product states, and Gibbs states and ground states of local Hamiltonians. 
Understanding the computational complexity of agnostic tomography of product states is therefore an important stepping stone to building up to richer approximations. 
As we demonstrate below, it turns out that learning the closest product state is already a surprisingly deep problem.
%for the simple class of product states, efficiently learning the best approximation to an arbitrary quantum state within the class requires considerable technical insight.
%In fact, it turns out that this gap in our knowledge is not merely coincidental.
%As we will soon see, this problem displays a very rich algorithmic landscape.
%In particular, the aforementioned ``high misspecification'' regime turns out to have deep connections to the literature on tensor optimization and approximating dense CSPs.
%\jerry{not particularly happy with this last paragraph...maybe just cut?}

%Obtaining learning algorithms which are robust to such ``model misspecification'' is known in the literature as \emph{agnostic state tomography} CITES.
%This is also the natural quantum analog of the well-studied \emph{agnostic learning} paradigm in classical learning theory CITES.

%\ewin{move to results?}
%In fact, any efficient algorithm for Question~\ref{question:main} immediately improves upon the performance of the classic swap test protocol of Harrow and Montanaro~\cite{hm13} for product state testing.
\subsection{Results}

We answer the aforementioned question in the affirmative and provide the first efficient algorithm for agnostic tomography of product states: 

\begin{theorem}[Learning the closest product state\ifstocproceedings\else, \cref{thm:main}\fi]
\label{thm:main-informal}
    There is an algorithm which, given as input $\eps > 0$ and $N = n^{\poly (1 / \eps )}$ copies of an unknown $n$-qubit\footnote{
        For simplicity, we only consider when the local systems are qubits.
        We believe that the results should generalize to qudits without too much struggle\ifstocproceedings\else: see \cref{rmk:qudit}\fi.
    } state~$\rho$, runs in time $\poly(N)$ and outputs the classical description of a pure product state $\ket{\phi}$ that, with probability at least $0.99$, satisfies
    \begin{equation}
        \label{eq:guarantee}
        \braket{\phi | \rho | \phi} \geq \OPT  - \eps \,.
    \end{equation}
    The algorithm also produces an estimate of $\OPT$ to $\eps$ error.
\end{theorem}

We pause to make several comments about this result.
First, the regime we are primarily interested in is when $\eps$ is a constant (though possibly small). In this regime, our algorithm runs in polynomial time.
This resolves an open question posed in \cite{gikl24}. 

Secondly, our result holds for all values of $\OPT$, and not just $\OPT$ close to $1$.
The setting where $\OPT$ is a small constant (say, $0.1$) is particularly challenging: in this regime, there may not be a unique closest product state.
In this setting, our algorithm in fact actually outputs a net (albeit in a relatively weak sense) over \emph{all} product states which are close to the unknown state $\rho$\ifstocproceedings\else; see Section~\ref{sec:reduction-to-poly-opt} for more detailed discussion\fi. 
Moreover, our algorithm does not need to know the value of $\OPT$, nor does it need even a lower bound on $\OPT$ (though if $\OPT$ is large the algorithm's complexity improves\ifstocproceedings\else---see \cref{rmk:applications_of_thm}\fi).
Note, however, that the guarantee on the fidelity of $\ket{\phi}$ with $\rho$ is only nontrivial when $\OPT > \eps$.

Finally, prior to this work, the only algorithms for this task were \emph{sample}-efficient, but not time-efficient.
For example, a polynomial number of random Clifford measurements suffices to estimate every fidelity with a product state $\braket{\pi | \rho | \pi}$ to $\eps$ error~\cite{hkp20}.
However, there are an exponential number of these product states, and computing even one fidelity from these randomized measurements requires exponential time~\cite{jvdn14}.

\paragraph{Improved product state testing.} Agnostic tomography of product states is closely related to the well-studied problem of product state testing~\cite{hm13}, where the goal is to determine whether or not a state $\ket{\psi}$ is a product state, or has fidelity at most $1 - \eps$ with any product state.
In the former case, the test should always accept, and in the latter, the test should reject with probability at least $p$, for some $p = \Theta (\eps)$.

Our results shed new light on this problem: the celebrated tester of Harrow and Montanaro~\cite{hm13} exhibits a strange behavior, wherein their rejection probability satisfies $p \leq 1/2 + o(1)$, even when $\eps \to 1$.
That is, for some reason, the tester cannot distinguish the case where $\ket{\psi}$ has overlap roughly $1/2$ with some product state, versus the case where the state has overlap $\ll 1/2$ with any product state.
Since our algorithm also produces an estimate of $\OPT$ to error $\eps$, it improves upon the best-known guarantees for product state testing~\cite{sw22} in this ``tolerant''~\cite{canonne20} regime.
%\ainesh{do we want an informal corollary here?}

\begin{comment}
# This stuff doesn't appear to work as written since our notion of a cover is too weak
\paragraph{Learning mixtures of product states}
Our techniques immediately allow us to learn states with low ``product rank'', i.e. states of the form i.e., states of the form 
\[
\sigma = \sum_{i = 1}^k w_i \ketbra{\phi_i}{\phi_i} \; .
\]
Theorem~\ref{thm:main-informal} almost immediately implies a learning algorithm for such states that runs in time
in time $n^{\poly (1 / w_{\min})}$, where $w_{\min} = \min_i w_i$.\jerry{do we want to have a short blurb in the appendix proving this?} \ewin{what is meant by ``learn'' here? are we assuming that the $\phi_i$'s have low fidelity with each other?}
\jerry{I don't think we need to make that assumption: I think the claim is that if $\sigma$ has this form, then we can recover another $\sigma'$ that is close in fidelity to $\sigma$, but I think this $\sigma'$ may have a lot more components}
\end{comment}

\paragraph{Computational lower bounds.}

It is natural to ask whether or not one can hope for a running time for this problem which is polynomial in both $n$ and $1/\eps$.
We complement our upper bound with a lower bound, demonstrating that our runtime is, in a qualitative sense, close to optimal:

\begin{theorem}[Hardness of product state approximation\ifstocproceedings\else, \cref{thm:hardness}\fi]
\label{thm:lb-informal}
    Suppose there is an efficient quantum algorithm for solving the following problem: given $\poly(n)$ copies of an unknown, $n$-qubit mixed state $\rho$, with probability $\geq 0.01$, output $\ket{\psi}$ satisfying
    \[
    \braket{\psi | \rho | \psi} \geq \max_{\ket{\pi} \in \mathcal{P}} \braket{\pi | \rho | \pi} - \frac{1}{\poly (n)} \,.
    \]
    Then $\mathsf{BQP} \supseteq \mathsf{NP}$.
\end{theorem}
In particular, this rules out algorithms with strongly polynomial dependencies on all parameters.
We prove this hardness via a straightforward, polynomial-time reduction to an $\mathsf{NP}$-complete problem.
Consequently, this also rules out any algorithms that have sub-exponential dependence on $1 / \eps$, assuming the quantum analog of the exponential time hypothesis. 
We interpret this as saying that it is likely challenging to obtain substantial qualitative improvements to the runtime in Theorem~\ref{thm:main-informal}.
% It is unclear precisely how hard product state learning is, or indeed, agnostic learning more broadly: with an appropriate formalization, these questions appear to be in QCMA, but it is an interesting question to pin down its complexity further.
% It is an interesting question to pin down the 

We also remark that this hardness result demonstrates an interesting computational-statistical gap for the problem of finding the closest product state.
Namely, classical shadow estimation~\cite{hkp20} demonstrates that this regime can be solved sample efficiently, but on the other hand, our lower bound demonstrates that this rate cannot be matched by any efficient algorithm.

\paragraph{Approximate tensor optimization.}
The upper and lower bound are based on a new connection to the classical problem of \emph{approximate tensor optimization}.
Here, one is given a $d$-tensor $T \in (\C^{n})^{\otimes d}$, and the goal is to find a unit vector $\vec{x} \in \C^n$ satisfying
\ewin{Add arrows to vectors}
\[
T(\vec{x}, \ldots, \vec{x}) \geq \max_{\twonorm{u} = 1} T(\vec{u}, \ldots, \vec{u}) - \eps \norm{T}_F \, .
\]
Our lower bound proceeds by direct reduction to this problem for $d = 4$, which is known to be $\mathsf{NP}$-hard when $\eps = 1 / \poly (n)$~\cite{fl17}, and our upper bound works by reducing the problem to many different instances of constrained versions of this problem.
This problem itself bears great resemblance to the problem of solving dense CSPs, and indeed, we believe the techniques we develop for constrained tensor optimization here may have applications to that setting as well.

\paragraph{Faster agnostic tomography of product states.}
In light of our lower bound, we ask whether there are simpler algorithms for agnostic tomography of product states, perhaps under additional assumptions.
We show that this is true for three natural settings: (1) when the best product state approximation is quite good; (2) when the number of choices for each qubit is discrete; and (3) when the output is allowed to be a matrix product state.

First, we obtain a linear copy and nearly-quadratic time algorithm for agnostic tomography of product states as long as the fidelity of the optimal solution exceeds a fixed constant (namely,~$5/6$):

\begin{theorem}[High-fidelity learning\ifstocproceedings\else, \Cref{thm:high_fidelity_alg}\fi]
\label{thm:high-fidelity-informal}
    There is an algorithm that takes as input a parameter $\eps > 0$ as well as $N = O(n / \eps)$ copies of an $n$-qubit state $\rho$, and has the following guarantees:
    Provided $\OPT  > 5/6 + \eps$, it runs in $O(N n \log n)$ time and outputs a pure product state~$\ket{\psi}$ that satisfies
    \[
    \braket{\psi | \rho | \psi} \geq \OPT   - \eps \,,
    \]
    (except with probability at most~$.01$).
\end{theorem}
In other words, so long as the quality of the product approximation $\OPT$ exceeds $5/6$, there is a strongly polynomial time algorithm for agnostic product state tomography.
This stands in stark contrast to the state of affairs for general $\OPT$, where the hardness result demonstrates such an algorithm is impossible.
The threshold $5/6$ naturally arises from our analysis, but it is an interesting open question to what extent it can be pushed.

We remark that the runtime dependence of the algorithm is linear in $1/\eps$, even though it is easily seen that \textit{estimating} $\OPT$ to $\pm \eps$ requires $\Omega(1/\eps^2)$ samples.
For example, this lower bound holds even in the special case when $\rho = \OPT\ketbra{0}{0} + (1-\OPT)\ketbra{1}{1}$ is a biased coin, and we want to distinguish whether (say) $\OPT = 0.9+\eps$ or $\OPT = 0.9-\eps$.
Our algorithm demonstrates that the task of \textit{finding} a state whose fidelity is within $\eps$ of the optimum may be easier.

%The linear dependence on $1/\eps$ in the sample complexity of \Cref{thm:high-fidelity-informal} is surprising, because it is easily seen that \textit{estimating} $\eta$ to $\pm \eps$ requires $\Omega(1/\eps^2)$ samples.
%This holds even in the special case when $\rho = \eta\ketbra{0}{0} + (1-\eta)\ketbra{1}{1}$ is a biased coin, and we want to distinguish whether (say) $\eta = 0.9+\eps$ or $\eta = 0.9-\eps$.
%\jerry{wait, is this that surprising? this phenomena is at least not unique to the quantum setting: it's known that for instance for the biased coin setting, getting KL / Hellinger (the classical equivalent to fidelity) $\eps$ only requires $1/\eps$ samples.}

Second, we give an efficient algorithm for agnostic tomography, when the class of states is the set of product states where each qubit is drawn from a finite set of possible states:
\begin{theorem}[Learning of a finite class of product states\ifstocproceedings\else, \Cref{thm:discrete}\fi]
\label{thm:discrete-informal}
    For $k = 1, \dots, n$, let $\mathcal{A}_k$ denote a set of single qudit pure states satisfying $|\mathcal{A}_k| \leq s$ and $\abs{\braket{\phi | \phi'}} \leq 1-\delta$ for all distinct $\ket{\phi}, \ket{\phi'} \in \mathcal{A}_k$.
    Let $\mathcal{A} = \mathcal{A}_1 \otimes \cdots \otimes \mathcal{A}_n$, and for any $n$-qudit quantum state $\rho$, let $\OPT_{\mathcal{A}} = \OPT_{\mathcal{A}}(\rho) = \max_{\ket{\pi} \in \mathcal{A}} \braket{\pi | \rho | \pi}$.
    Then there is an algorithm which, given as input $\eps > 0$ and $N = \poly((ns)^{\log(1/\eps)/\delta})$ copies of an $n$-qudit state~$\rho$, runs in $\poly(N)$ time and outputs the classical description of some $\ket{\psi} \in \calA$ satisfying
    \[
        \braket{\psi | \rho | \psi} \geq \OPT_{\mathcal{A}} - \eps \,,
    \]
    (except with probability at most $.01$).
\end{theorem}
Stated plainly, so long as there are a finite set of possible states, and these states are all pairwise separated, then there is an efficient algorithm for agnostic tomography for this class of product states.
We note that, similar to Theorem~\ref{thm:main-informal}, our algorithm actually outputs all good solutions.
This result also directly generalizes prior work of Grewal, Iyer, Kretschmer, and Liang~\cite{gikl24}, which studied the special case where each $\mathcal{A}_k$ is the set of $1$-qubit stabilizer states.
A very similar result was also obtained independently in \cite{cgyz24}, albeit with quite different techniques.

Third, we give an algorithm for learning a good matrix-product state approximation of a given state $\rho$.
Matrix product states with small bond dimension can be used to efficiently describe systems of multiple particles where particles share a small (but non-zero) amount of entanglement, and are ubiquitous in quantum many-body physics~\cite{pvwc07,Schollwock11}.
We give an algorithm for \emph{agnostic} tomography of matrix product states.

\begin{theorem}[Agnostic (improper) learning of matrix product states\ifstocproceedings\else, \Cref{thm:mps}\fi]
    \label{thm:mps-informal}
    Let $n, d, r$ be positive integers, and let $\mathrm{MPS}_{n,d,r}$ be the class of matrix product states on $n$ qudits of local dimension $d$ with bond dimension $r$\ifstocproceedings\else\ (\cref{def:mps})\fi.
    For any state $\rho \in (\C^{d \times d})^{\otimes n}$, let $\OPT_{r} = \OPT_{n, d, r} (\rho) = \max_{\ket{\phi} \in \mathrm{MPS}_{n, d, r}} \braket{\phi | \rho | \phi}$ be the maximum fidelity any such MPS has with $\rho$.
    There is an algorithm which, given as input $\eps > 0$ and $N = \poly (n, d, r, 1/\eps)$ copies of an unknown $n$-qudit state $\rho$, runs in time $\poly(N)$ and outputs the classical description of a matrix product state $\ket{\widehat{\phi}}$ of bond dimension $dn^2 \cdot \poly (r, 1/ \eps)$ such that 
    \[
        \braket{\widehat{\phi} | \rho | \widehat{\phi}} \geq \OPT_r - \eps \,,
    \]
    (except with probability at most $.01$).
\end{theorem}

We can relate this task back to learning the closest product state by taking $r = 1$ and $d = 2$; then, $\mathrm{MPS}_{n,d,r}$ is the class of product states over qubits, and our algorithm is able to output a matrix product state with bond dimension $n^2 \poly(1/\eps)$ whose fidelity with $\rho$ is at least $\OPT - \eps$.
This gives an improper learner for product states, ``improper'' referring to our output not being a product state but instead a low-entanglement state.
Our main result \cref{thm:main-informal} is a \emph{proper} learner for product states.
In the error regimes of our lower bound, this gives an instance where improper agnostic learning is efficient, but proper agnostic learning is $\textsf{NP}$-hard.
In general, the output of this algorithm is an MPS with a bond dimension of at least $r n^2$, which achieves a fidelity which is optimal with respect to MPSs with bond dimension $r$; this dependence on $n$ in particular seems to be what makes this result more straightforward than proper learning of MPSs.

For this task, we recognize that the algorithm of Cramer, Plenio, Flammia, Somma, Gross, Bartlett, Landon-Cardinal, Poulin, and Liu~\cite{cpfsgb10} to learn an MPS also works when the input state is not an MPS, but merely has large constant fidelity with an MPS; our contribution is to generalize it to the agnostic case and perform the necessary analysis.
\ifstocproceedings\else We give a more detailed discussion in \cref{sec:mps}.\fi

\paragraph{Techniques.}
All of these results, as well as \cref{thm:main-informal}, are all based on a common algorithmic framework, which may have applications more broadly.
At a very high level, our algorithms sweep through the qubits one at a time, and generate a set of candidates for good solutions on the qubits seen so far.
This cover is then used as the starting point for generating candidates over the subsequent qubits.
The main algorithmic challenge is in making extending the cover efficient.
In the case of \cref{thm:main-informal}, extending this cover is intimately connected to tensor optimization, as mentioned above.
To achieve our faster algorithms~\cref{thm:high-fidelity-informal,thm:discrete-informal,thm:mps-informal}, our key insight is that there are relatively simple and ``greedy'' techniques that allow us to extend this cover.

Our algorithms interleave classical computation with a particular quantum subroutine: the only way we access $\rho$ is to perform tomography on various subspaces of subsystems, e.g.\ to estimate $\Pi \tr_S(\rho) \Pi$ for $S$ some subset of qubits and $\Pi$ a projector onto a subspace of $\poly(n)$ dimension (which can be represented efficiently with a quantum circuit).
Our algorithms apply this subroutine to various choices of $\Pi$ and $S$, which are adaptively chosen after classical computation on the output of the previous tomography routines.
Since such a tomography subroutine can be performed with single-copy measurements, our algorithm can also be performed with only single-copy measurements.
However, the adaptivity of this algorithm appears inherent: the classical shadows formalism~\cite{hkp20} is the standard technique to allow algorithms like these to perform all of their measurements up-front, but doing this comes at the cost of exponential running time, which we cannot tolerate.

\subsection{Related work}

%\begin{verbatim}
%1. Harrow-Montanaro and follow-ups
%    * testing MPSes
%    * Why don't they work for our purposes here?
%2. Agnostic learning literature
%    * GIKL papers
%    * Anshu Arunachalam survey discussion of agnostic tomography
%    * Sitan preprint
%    * Note that these are all "discrete" classes, so those settings are a bit different
%3. Is there any relevant classical literature to discuss?
%    * Learning mixtures of product distributions over discrete domains
%    *\ainesh{Literature on polynomial optimization and csps}? 
%4. Open directions?
%\end{verbatim}

\paragraph{Concurrent work.} 
In independent and concurrent work, Chen, Gong, Ye, and Zhang~\cite{cgyz24} give an algorithm for agnostic tomography of a finite set of product states, attaining a near-identical result to \cref{thm:discrete-informal} via completely unrelated techniques.
They also give an improved algorithm when the product states are stabilizer; we are also able to get a similar improvement in this setting\ifstocproceedings\else (see \cref{rmk:discrete-shadows} for more details)\fi.

\paragraph{Agnostic tomography.}
The notion of agnostic tomography was introduced by Grewal, Iyer, Kretschmer, and Liang~\cite{gikl24}, though similar notions have been considered under the notion ``quantum hypothesis selection''~\cite{bo24} and in the PAC-learning setting; we refer to the survey \cite{aa23} for a thorough discussion.
Recent work has given agnostic tomography algorithms for stabilizer product states~\cite{gikl24} and stabilizer states~\cite{cgyz24}.
These algorithms use unrelated techniques.

\paragraph{Product state testing.}
A notable related algorithm is the product state test, which, using copies of a state $\rho$ is able to distinguish the cases that $\OPT = 1$ from $\OPT \leq 1 - \eps$.
This algorithm, introduced by \cite{mkb05} and analyzed by Harrow and Montanaro~\cite{hm13}, plays an important role in complexity theory, being used to prove that $\mathsf{QMA}(2) = \mathsf{QMA}(k)$ for $k > 2$.
Though the algorithm suffices for testing, it cannot be used to estimate $\OPT$ when $\OPT$ is bounded away from $1$~\cite[Appendix B]{hm13}.
For similar reasons, it also does not seem to help with the task of finding good product states.

\paragraph{Product state approximations in other contexts.}
Our algorithm shows that it is possible to estimate the ``geometric measure of entanglement'' of a given pure state in polynomial time.
This measure of entanglement, defined by Wei and Goldbart \cite{wg03,vprk97}, has seen significant investigation as a measure of multipartite entanglement.
This interest comes from this measure's potential to capture aspects of entanglement in condensed matter systems which cannot be captured by the typical, bipartite measures of entanglement~\cite{wdmvg05,odv08,ow10}.
See the survey of De Chiara and Sanpera for further discussion \cite{ds18}.
%Citation to nmm09 is a little less relevant because it's for mixed states.
However, research has been limited by computational intractability, so our work may give a possible avenue to expand its scope via quantum simulation.

%\ewin{did this sloppily, check if it reads ok}
Mean-field approximations also arise naturally in contexts where we want to understand things like ground states of many-body systems, and only have a handle on product states~\cite{bh16}.
For example, to prepare ground states of many-body systems, current heuristic phase estimation methods have a running time which depends on the fidelity between the ground state and an input product state~\cite{garnet23}.

\paragraph{Agnostic learning of product distributions.}
In some ways, the problem we consider here is the quantum analog of the well-studied problem of agnostic learning of product distributions on the hypercube.
In its most basic form, we are given samples from a distribution that is close to a product distribution over the hypercube, and the goal is to learn the optimal product distribution approximation.
Efficient algorithms for this problem were given in~\cite{diakonikolas2019robust,lai2016agnostic}.
However, these algorithms only work when their version of $\OPT$ is sufficiently large; in classical learning theory, the regime when $\OPT$ is small is known as \emph{list learning}, and efficient algorithms for list learning of product distributions are also known; see, e.g.~\cite{charikar2017learning,kothari2018robust}.
However, the guarantees they obtain are quite incomparable to ours, and their techniques do not have a meaningful parallel in the quantum setting.

\paragraph{Polynomial optimization.}
Polynomial optimization over the sphere is hard in general.
Multiplicative approximations for optimizing low-degree polynomials in the worst case are well-understood (see \cite{bhattiprolu2017weak} and references therein).   
However, polynomial optimization has still found prominent applications in classical learning problems in the last decade.
The polynomials that naturally appear in these settings do not tend to be worst-case, and admit significantly better approximations.
Optimizing low-degree polynomials (often subject to polynomial constraints) has become a key algorithmic primitive in dictionary learning~\cite{barak2015dictionary}, tensor decomposition~\cite{hopkins2015tensor}, robustly learning Gaussian mixture models~\cite{moitra2010settling,belkin2015polynomial, liu2021settling, bakshi2022robustly} and private and list-decodable learning~\cite{hopkins2023robustness, karmalkar2019list, raghavendra2020list, bakshi2021list}.
These techniques have also found applications in quantum tasks, such as best separable state~\cite{barak2017quantum} and learning quantum Hamiltonians~\cite{blmt24,narayanan2024improved}. 
An interesting feature of our algorithm, compared to this other work, is that we do not establish uniqueness of some strong structure arising from the underlying parameters.
Instead, we output a (non-unique) cover over solutions, and use polynomial optimization as a subroutine to produce such a cover.

Optimizing low-degree polynomials over the hypercube also leads to approximation algorithms for constraint satisfaction problems on dense and low-threshold rank graphs~\cite{barak2011rounding, raghavendra2012approximating, manurangsi2016birthday} and high-dimensional expanders~\cite{alev2019approximating}. These results roughly proceed via solving a sum-of-squares relaxation of a polynomial maximization problem, and obtain additive error that scales proportional to $\eps$ times the $\ell_2$-norm of the coefficients of the polynomial and runs in $n^{\poly(1/\eps)}$ time.   Similar techniques have also appeared in the context of refuting random CSP's~\cite{raghavendra2017strongly}. 

A closely related problem is that of optimizing random polynomials over the sphere, which has deep connections to statistical physics and admits an additive-error guarantee under full replica-symmetry breaking~\cite{subag2021following}.
While our optimization problem does not involve random polynomials, we show that we can optimize low-degree polynomials up to small additive-error efficiently.

\section{Technical overview}
\label{sec:tech}

We now cover the key technical ideas of our algorithms.
\ifstocproceedings
\else
The precise version of the main algorithm can be found in \cref{alg:outer_loop}, which describes its outer loop, and \cref{alg:inner_loop}, which describes the main subroutine.
Further explanation of the subroutine can be found in \cref{subsec:finding}.
\fi

\paragraph{Why naive approaches fail.}
Given an $n$-qubit quantum state with density matrix $\rho \in \C^{2^n \times 2^n}$, we want to find the product state that maximizes fidelity with $\rho$.
The obvious algorithm that one might try to learn the closest product state is to take the best pure state approximation to each of its single-qubit subsystems.
This algorithm works if $\rho$ itself is a pure product state.
However, the single-qubit subsystems do not contain enough information to deduce the best product state, even when the fidelity of $\rho$ with the best product state is very close to $1$.
This phenomenon is why many naive approaches give exponentially poor approximations to the optimal value.

An illustrative example is to consider $\rho = \ketbra{\psi}{\psi}$ where $\ket{\psi}$ is the state proportional to
\begin{equation*}
    \ket{\psi} \propto  \sqrt{1-\eps} \ket{0^n} + \sqrt{\eps} \ket{+^n}
\end{equation*}
for some small constant $\eps$.
Because $\braket{+^n|0^n} = 2^{-n/2}$, $\ket{\psi}$ as written is exponentially close to normalized.
The fidelity with the product state $\ket{0^n}$ can be computed explicitly:\footnote{It follows from one of our later results \ifstocproceedings \else(\Cref{thm:local_opt_works_better}) \fi that the maximum product state fidelity with $\ket{\psi}$ is exponentially close to $1 - \eps$.}
\[
    \braket{0^n| \rho|0^n} = \abs{\braket{\psi|0^n}}^2 = \mparen{\frac{\sqrt{1-\eps} + \sqrt{\eps/2^n}}{\norm{\sqrt{1+\eps} \ket{0^n} + \sqrt{\eps} \ket{+^n}}_2}}^2 \ge 1 - \eps\,.
\]
In the limit of large $n$, the one-qubit density matrices of $\ket{\psi}$ all approach
\[
\rho_i = \begin{bmatrix}
    1 - \eps/2 & \eps/2\\
    \eps/2 & \eps/2
\end{bmatrix}
\]
We will see that there is a distinct state $\ket{\psi'}$ that is also $\eps$-close to a product state, and has identical reduced density matrices, but for which $\ket{0^n}$ is a very bad product state approximation.
Take an eigendecomposition of $\rho_i$ as
\[
\rho_i = p_1\ketbra{v_1}{v_1} + p_2\ketbra{v_2}{v_2}\,,
\]
with $p_1 > p_2$.
The state
\[
\ket{\psi'} = \sqrt{p_1}\ket{v_1}^{\otimes n} + \sqrt{p_2}\ket{v_2}^{\otimes n}
\]
also has at least $1-\eps$ fidelity with a product state (namely $\ket{v_1}^{\otimes n}$), and has all its local density matrices equal to $\rho_i$.
However, calculation shows that $\abs{\braket{\psi'|0^n}}^2$ decays exponentially to $0$ in the limit of large $n$, because both $\ket{v_1}$ and $\ket{v_2}$ are constant-far from $\ket{0}$.
So, there is not enough information in the one-qubit reduced density matrices to learn the best product state approximation.

\paragraph{Barriers to agnostic product tomography.}
The hard case above illuminates broader challenges inherent to this task.
We are concerned with optimizing the fidelity $\braket{\pi | \rho | \pi}$ over the class of product states; however, fidelity is typically quite poorly behaved.
For example, almost all product states have exponentially small fidelity with $\rho$, which is too small to detect, and the fidelity landscape can have many local optima which thwart local search algorithms, like those based on convex optimization.
This ill-behavedness is a well-established phenomenon related to the ``barren plateau'' problem in quantum machine learning~\cite{mbsbn18}.

%fail since there are several sub-optimal local minima the algorithm is likely to get stuck in; indeed, one can show that there are instances for which there are super-polynomially many local minima.
The regime where the optimal fidelity is a small constant like $0.1$ is particularly challenging since, unlike the case where $\OPT$ is near $1$, there are many well-separated globally optimal solutions.
This lack of uniqueness presents basic issues for us: even if we manage to traverse the fidelity landscape and find many locally-optimal product states with fidelity $0.1$, how can we conclude that we are done, and certify that there is no product state with fidelity $0.2$?

\paragraph{Maintaining a cover over good product states.}
Our crucial insight is that we can efficiently maintain a cover over all product states that have large fidelity.
This insight is enabled by the following observations:
\begin{enumerate}
    \item If a product state $\ket{\pi}$ has good fidelity with $\rho$, then its restriction to a subsystem $S$ has good fidelity with the partial trace of $\rho$ onto the subsystem: $\braket{\pi | \rho | \pi} \leq \braket{\pi_S | \rho_S | \pi_S}$.
    \item The number of product states with good fidelity with $\rho$ and which have pairwise small fidelity with each other is small.
\end{enumerate}
The first observation means that we do not have to optimize fidelity over the entire space of product states: just those which are extensions of good product states over a subsystem.
In short, we can build a set of good product states qubit by qubit.
The second observation means that, instead of maintaining \emph{all} good product states, of which there could be exponentially many, it suffices to maintain a small number of well-separated good product states.
In short, it suffices to maintain a cover. 

A priori, it is even unclear whether a small cover over such product states exists.
Our main technical contribution is to establish the existence of such a cover and demonstrate that it can be computed efficiently.
Our algorithm starts with a cover over good product states for $\rho_{[1]}$, the state on qubit 1, and iteratively expands the cover a single qubit at a time.
In particular, we show that given a cover for qubits $1, 2, \dots, m-1$, extending it to qubit $m$ can be reduced to polynomial optimization problems over the sphere with a dimension-independent number of $\ell_2$ and $\ell_\infty$ constraints.

For the remainder of the section, we outline our approach to efficiently maintain a cover.

\paragraph{Fidelity and tangent distance.}
We begin by introducing a pa\-ram\-e\-tri\-za\-tion of product states which is used throughout the paper.
For a $n$-dimensional vector of complex numbers, $\vec{z} \in \C^n$, we denote its corresponding product state by
\begin{equation*}
    \ket{\pi_{\vec{z}}} = \bigotimes_{i=1}^n \frac{\ket{0} + z_i\ket{1}}{\sqrt{1 + \abs{z_i}^2}}.
\end{equation*}
Looking ahead, we ultimately want to optimize over these $z_i$'s, so we want a notion of cover which behaves nicely with respect to this parametrization.
Fidelity is well-known to be an unwieldy notion of distance between quantum states and is typically hard to analyze.
So, instead of constructing a cover directly using fidelity, we introduce an alternate measure between product states:

\begin{definition}[Tangent distance\ifstocproceedings\else, \cref{def:tangent-distance}\fi] 
Given two product states $\ket{\pi_{\vec{z}}}$ and $\ket{\pi_{\vec{a}}}$, the tangent distance between them is defined as
\begin{equation*}
    \dtan(\ket{\pi_{\vec{z}}}, \ket{\pi_{\vec{a}}}) = \norm[\Big]{\frac{\Vec{z} - \Vec{a} }{1 + \Vec{z}^*\Vec{a}}}_2 = \parens[\Big]{\sum_{i=1}^n \abs[\Big]{\frac{z_i - a_i}{1 + z_i^* a_i}}^2}^{1/2}\,.
\end{equation*}
\end{definition}
We call it ``tangent distance'' because, for a single qubit, this measure corresponds to $\abs{\tan(\theta)}$, where $\theta$ is the angle between the two states on the Bloch sphere.
This notion of distance satisfies several desiderata, including being invariant under single-qubit unitaries and being equal to $\twonorm{\vec{z}}$ when $\vec{a} = \vec{0}$ \ifstocproceedings\else(see \cref{sec:tan-distance} for details)\fi. Importantly, tangent distance can be related to fidelity as follows\ifstocproceedings\else\ (see \cref{lem:dtan_fidelity_relationship} for a proof)\fi:
\begin{equation}\label{eq:dtan-intro}
    \log\left(\frac{1}{\abs{\braket{\pi_{\vec{z}} | \pi_{\vec{a}}}}^2}\right) \leq \dtan(\ket{\pi_{\vec{z}}}, \ket{\pi_{\vec{a}}})^2 \leq \frac{1}{\abs{\braket{\pi_{\vec{z}} | \pi_{\vec{a}}}}^2} - 1\,.
\end{equation}

Now, we can introduce our notion of cover under tangent distance\footnote{
    Though our algorithm naturally produces a cover with respect to tangent distance, one can use \eqref{eq:dtan-intro} to convert the guarantees to those involving fidelity.
}\ifstocproceedings\else\ (\cref{def:product-cover})\fi: a cover $\calC$ over product states is $(\eta, \eps)$-good for a state $\rho$ if the following hold
\begin{itemize}
    \item \textbf{Good fidelity:} For all product states $\ket{\pi} \in \calC$, $\braket{\pi | \rho | \pi} \geq \eta - \eps$;
    \item \textbf{Separation:} For all distinct $\ket{\pi} , \ket{\pi'}  \in \calC$, $\dtan\parens{ \ket{\pi} , \ket{\pi'} } \geq 2/\eta$;
    \item \textbf{Coverage:} For any product state $\ket{\phi}$ such that $\braket{\phi | \rho | \phi} \geq \eta$, there exists a $\ket{\pi} \in \calC$ such that $\dtan\parens{\ket{\phi} , \ket{\pi} } \leq 3/\eta$.
\end{itemize}
We design an algorithm which, given $\eta$ and $\eps < \eta/3$, outputs a $(\eta, \eps)$-good cover, where every product state $\ket{\pi_{\vec{z}}}$ in the cover is described by its parametrization $\vec{z}$.
In particular, this gives a product state with fidelity $\geq \eta - \eps$, assuming a product state with fidelity $\eta$ exists.
By performing binary search on $\eta$, one can use this to find a product state with fidelity $\geq \OPT - \eps$, as stated in \cref{thm:main-informal}.
%An ideal cover over product states should have three properties: (a) for all product states in the cover, the fidelity with the input state, $\ket{\rho}$, is roughly  $\eta$, (b) no two product states in the cover are too close to each other, and (c) for any product state $\ket{\phi}$ that has fidelity $\eta$ with the input state, there exists a product state in the cover with that is close to $\ket{\phi}$.

\paragraph{Existence of small covers.} 
Our first step is to show that the size of an $(\eta, \eps)$-good cover is at most $6/\eta$\ifstocproceedings\else\ (see \cref{claim:cover_size})\fi.
Let $\mathcal{C} = \braces{\braket{\pi^{(i)}}}_i$ be an $(\eta, \eps)$-good cover.
For intuition, suppose the product states in the cover were not just well-separated but orthogonal.
Then each captures a different part of the ``mass'' of $\rho$.
That is,
\begin{align*}
    1 = \tr(\rho) \geq \sum_{i} \braket{\pi^{(i)} | \rho | \pi^{(i)}} \geq \abs{\mathcal{C}}(\eta - \eps) \geq \abs{\mathcal{C}} (2\eta/3),
\end{align*}
where the last two inequalities use the good fidelity property of the cover and that $\eps < \eta/3$, respectively.
In general, $\sum_{i} \braket{\pi^{(i)} | \rho | \pi^{(i)}}$ is equal to $\tr(MM^\dagger \rho)$ for $M$ the matrix whose columns are the states in the cover $\ket{\pi^{(i)}}$.
Then,
\begin{align*}
    \abs{\mathcal{C}} (2\eta/3) \leq \tr(MM^\dagger \rho) \leq \opnorm{MM^\dagger} = \opnorm{M^\dagger M} \leq 1 + \abs{\mathcal{C}} (\eta/2),
\end{align*}
giving the bound $\abs{\mathcal{C}} \leq 6/\eta$.
In the last step, we used the well-separated condition: the diagonal entries of $M^\dagger M$ are $1$, the off-diagonal ones have magnitude at most $\eta/2$ by \cref{eq:dtan-intro}, and by the Gershgorin circle theorem the operator norm of $MM^\dagger$ is bounded by the maximum sum of magnitudes of any of its rows.

We can further show how to construct an $(\eta, \eps)$-good cover algorithmically\ifstocproceedings\else\ (\cref{alg:outer_loop})\fi.
We do this by iteratively forming an $(\eta, \eps)$-good cover for $\rho_{[m]}$, the partial trace of $\rho$ onto qubits $1$ through $m$, for $m$ from $1$ to $n$.
We can construct a good cover for $\rho_{[m]}$ greedily: start with $\mathcal{C}_m$ empty, and look for a violation of the coverage property.
When we find one, add the corresponding $\ket{\phi}$ to $\mathcal{C}_m$, and repeat.
Because we know an $(\eta, \eps)$-good cover on $m-1$ qubits $\mathcal{C}_{m-1}$, we can restrict our search to just look over product states $\ket{\phi}$ whose first $m-1$ qubits are close in tangent distance to an element of $\mathcal{C}_{m-1}$.
This makes the problem of finding a violation tractable, since we only have to search in the neighborhood of some ``root'' product state.
In particular, we show that it suffices to solve the following optimization problem\ifstocproceedings\else\ (see \cref{claim:outer})\fi.
\begin{equation} \label{program:intro}
\begin{aligned}
    & \underset{\vec{z} \in \C^{m}}{\text{maximize}}
    & & \bra{\pi_{\vec{z}}} \rho \ket{\pi_{\vec{z}}} \\
    & \text{subject to}
    & & \dtan(\ket{\pi_{\vec{z}}}, \ket{\pi_{\vec{a}}}) \geq 2/\eta \text{ for all } \ket{\pi_{\vec{a}}} \in \calC_m, \\
    &&& \dtan(\ket{\pi_{\vec{z}}}, \ket{\pi_{\vec{0}}}) \leq 4/\eta.
\end{aligned}
\tag{Tangent Cover}
\end{equation}
The precise soundness and completeness guarantees needed are shown in \ifstocproceedings the full algorithm\else \cref{alg:outer_loop}\fi; the constraints allow for significant slack.
Note that the second constraint is equivalent to $\twonorm{\vec{z}} \leq 4/\eta$.

%As alluded to earlier, the algorithm for constructing a cover is to iterate through the qubits.  We assume we have maintained a cover, $\calC$, which constitutes all of the product states that have high fidelity with the input on the first $m$ qubits and we want to extend it to include the $(m+1)$-th qubit. For each product state in the cover, we try to extend the set of states on $m+1$ qubits that continue to have high fidelity with the input state. For notational convenience, let $\ket{0^m}$ be one such product state on the first $m$ qubits (which can be ensured by rotating the space). Extending the cover for $\ket{0^m}$ is equivalent to the following optimization problem: 

\paragraph{Constructing covers and polynomial optimization.}
Now, we consider the task of solving \eqref{program:intro}.
Solving this even in the simplest case is not straightforward.
An example to keep in mind is the following: suppose we are adding our first state to $\mathcal{C}_n$, which is currently empty.
So, there are no ``farness'' constraints, the first kind of constraint in the program.
Then, let $\rho = \proj{\psi}$, where $\ket{\psi}$ is a superposition over computational basis strings with Hamming weight $0$ and $d$:
\begin{align*}
    \ket{\psi} &= \sqrt{\gamma} \ket{0^n} + \sqrt{\frac{1-\gamma}{\binom{n}{d}}}\sum_{\substack{x \in \braces{0,1}^n \\ \abs{x} = d}} \ket{x}.
\end{align*}
We are imagining, say, $\gamma = 0.9 \eta$.
Then $\braket{0^n | \rho | 0^n} = \gamma$, so our root state has good fidelity, but not quite enough to be a violation as we desire.
(This can indeed happen; though $\ket{0^n}$ comes from the cover $\mathcal{C}_{n-1}$, so $\ket{0^{n-1}}$ has fidelity at least $\eta - \eps$, it is extended by one qubit, which can drop the fidelity to $\gamma$ or lower.)
However, for $\vec{z} = \frac{1}{\sqrt{n}} \vec{1}$,
\begin{align*}
    \braket{\pi_{\vec{z}} | \rho | \pi_{\vec{z}}}
    &= (1 + 1/n)^{-n} \parens[\Big]{\sqrt{\gamma} + \sqrt{\frac{\binom{n}{d}(1-\gamma)}{n^d}}}^2 \\
    &\underset{n \text{ large}}{\approx} \frac{1}{e} \parens[\Big]{\sqrt{\gamma} + \sqrt{(1-\gamma)/d!}}^2,
\end{align*}
which can be larger than $\eta$ even for $d = \Theta(\log(1/\eta) / \log\log(1/\eta))$.
Note that $\twonorm{\vec{z}} = 1 \leq 4/\eta$, so it is close enough to the root in \eqref{program:intro}, and our algorithm must be able to recognize this better solution of $\ket{\pi_{\vec{z}}}$.
This demands knowledge of $\rho$ in a (quite large) Hamming ball around the root product state.
Further, by changing the signs of the $\ket{x}$'s in $\ket{\psi}$, \eqref{program:intro} can encode dense $d$-CSP instances.
This suggests that the right approach is a reduction to polynomial optimization.

First, we consider solving \eqref{program:intro} when $\mathcal{C}_m$ is empty.
We can reduce this to low-degree polynomial optimization over the sphere.
We start by observing that it suffices to consider the projection of $\rho$ on basis states of low Hamming weight.
Concretely, let $\Pi_{\geq d}$ be the projection onto the subspace of Hamming weight at least $d = O(1/\eta + \log(1/\eps))$.
Then, we show that, provided $\norm{\vec{z}} \leq 4/\eta$ as in \eqref{program:intro},
\begin{equation*}
    \norm{ \Pi_{\geq d}  \ket{\pi_{\vec{z}}} } \leq \eps \,.
\end{equation*}
This is a Chernoff bound, since the squared norm is the probability that the $n$ Bernoulli random variables sums to at least $d$, where the probabilities come from the $\vec{z}$'s\ifstocproceedings\else\ (see \cref{lem:low_weight} with $\mu = 4/\eta$)\fi.
So, it suffices to perform state tomography for $\rho$ on the space of low Hamming weight vectors $\rho_{d} = \Pi_{< d} \rho \Pi_{<d}$, which is computationally efficient because this subspace has dimension $O(n^d)$\ifstocproceedings\else\ (\cref{lem:subspace-tomography})\fi.
We can use $\rho_d$ in place of $\rho$ in the objective because
\begin{equation*}
%\label{intro:closeness}
    \begin{split}
        \abs{\bra{\pi_{\vec{z}}} \rho \ket{\pi_{\vec{z}}} - \bra{\pi_{\vec{z}}} \rho_d \ket{\pi_{\vec{z}}}} 
    &= \abs[\Big]{\Tr\parens[\Big]{\rho(\proj{\pi_{\vec{z}}} - \Pi_{<d} \proj{\pi_{\vec{z}}} \Pi_{<d})}} \\
    &\leq \opnorm{\proj{\pi_{\vec{z}}} - \Pi_{<d} \proj{\pi_{\vec{z}}} \Pi_{<d}} \\
    &\leq 2\twonorm{\ket{\pi_{\vec{z}}} - \Pi_{<d} \ket{\pi_{\vec{z}}}} \\
    &= 2\twonorm{\Pi_{\geq d} \ket{\pi_{\vec{z}}}} \leq \eps\,.
    \end{split}
\end{equation*}
Further, once we have our estimate of $\rho_d$, the objective function is fully specified explicitly; the rest of the algorithm is classical.
Because $\rho_d$ is only supported on low Hamming weight, $\bra{\pi_{\vec{z}}} \rho_d \ket{\pi_{\vec{z}}}$ is a low-degree polynomial up to a normalization factor:
\ifstocproceedings
\begin{align}
\label{eqn:intro-expand-objective}
    \bra{\pi_{\vec{z}}} &\rho_d \ket{\pi_{\vec{z}}} 
    = \nonumber\\
    &\underbrace{ \frac{1}{ \prod_{ i \in [m] } (1+ \abs{z_i}^2) } }_{\eqref{eqn:intro-expand-objective}.(1) } \underbrace{ \sum_{ x, x' \in \braces{0,1}^m } \bra{x} \rho_d \ket{x'} \Paren{ \vec{z}^* }^x \Paren{ \vec{z} }^{x'}}_{\eqref{eqn:intro-expand-objective}.(2)} \,.
\end{align}
\else
\begin{equation}
\label{eqn:intro-expand-objective}
    \bra{\pi_{\vec{z}}} \rho_d \ket{\pi_{\vec{z}}} 
    = \underbrace{ \frac{1}{ \prod_{ i \in [m] } (1+ \abs{z_i}^2) } }_{\eqref{eqn:intro-expand-objective}.(1) } \underbrace{ \sum_{ x, x' \in \braces{0,1}^m } \bra{x} \rho_d \ket{x'} \Paren{ \vec{z}^* }^x \Paren{ \vec{z} }^{x'}}_{\eqref{eqn:intro-expand-objective}.(2)} \,.
\end{equation}
\fi
\eqref{eqn:intro-expand-objective}.(2) is a degree-$2d$ polynomial in the $z_i$'s and their complex conjugates.
Further, when the $\abs{z_i}$'s are small, we can approximate term \eqref{eqn:intro-expand-objective}.(1) by $\exp\Paren{- \norm{z}^2_2}$, which is a bounded scalar variable that we can hardcode into our constraints.
While the $\abs{z_i}$'s won't always be small, we can guess the ones that are large, fix their value and use the above approximation on the rest.
This is where we pick up an $\ell_\infty$ constraint on the $z_i$'s, since we must enforce that entries which we do not guess are small.
This reduces the algorithm to solving problems of the following form.
\begin{equation*}
%\label{eqn:polynomial-intro}
    \max_{\substack{\twonorm{\vec{z}}=1 \\ \infnorm{\vec{z}} \leq \mu }} p(\vec{z}) = \max_{\substack{\twonorm{\vec{z}}=1 \\ \infnorm{\vec{z}} \leq \mu }} \sum_{ x, x' \in \braces{0,1}^m } \bra{x} \rho_d \ket{x'} \Paren{ \vec{z}^* }^x \Paren{ \vec{z} }^{x'}.
\end{equation*}
Optimizing low-degree polynomials over the sphere is known to be hard to approximate up to polynomial factors in the worst-case, even when the degree is $4$~\cite{barak2012hypercontractivity,bhattiprolu2017weak}.
However, in our case, $p(\vec{z})$ is not an arbitrary low-degree polynomial, but is quite `small': the $\ell_2$ norm of the coefficients $\braket{x | \rho_d | x'}$ is bounded, since it is $\fnorm{\rho_d} \leq \fnorm{\rho} \leq \tr(\rho) = 1$.
%To see this, observe that each coefficient is of the form $\abs{ \bra{x} \rho_d \ket{x'} }^2$ which picks up the squared magnitude of a single entry in $\rho_d$. Since $\rho_d$ has trace and therefore Frobenius norm $\leq 1$, it follows that the $\ell_2$-norm coefficients of the polynomial in term \eqref{eqn:intro-expand-objective}.(2) is at most $1$.
We will show that, in this case, obtaining additive error that scales with $\eps$ admits a polynomial time algorithm.
Additive-error approximations to max $k$-CSPs also admit a similar guarantee.

In the general case, we must also deal with the farness constraints in \eqref{program:intro}, $\dtan(\ket{\pi_{\vec{z}}}, \ket{\pi_{\vec{a}}}) \geq 2/\eta$ for a small number of $\ket{\pi_{\vec{a}}}$'s.
Recall that $\dtan\Paren{ \ket{\pi_{\vec{a}}} , \ket{\pi_{\vec{z}}} }^2 = \sum_{i \in [n]} \abs{ \frac{z_i - a_i }{ 1+ z_i^* a_i } }^2$ by definition.
We will not try to directly optimize over these constraints. 
We use a similar trick as before and show that when the $\abs{z_i}$'s are small, $\dtan\Paren{ \ket{\pi_{\vec{a}}} , \ket{\pi_{\vec{z}}} } \approx \twonorm{\vec{a} - \vec{z}}$\ifstocproceedings\else\ (see \cref{lem:dtan-l2})\fi.
These constraints essentially only enforce what $\vec{z}$ can be in the low-dimensional subspace spanned by the $\vec{a}$'s.
So, we can guess the choice of $\vec{z}$ on this subspace (in addition to the coordinates for which $\abs{z_i}$ is large), and for each guess, solve the problem with the guess hardcoded into the constraints.
Putting all the steps together, we show that we can reduce the problem of extending the cover to a polynomial optimization problem, where the $\ell_2$ norm of the coefficients is bounded by $1$, subject to $\ell_2$ and $\ell_\infty$ constraints. 

%We will approximate \ref{program:intro} by a different optimization problem where the objective function and the constraints are low-degree polynomials, while only losing an additive $\eps$ in the objective value of the optimal solution.
%First, note that na\"ively writing the aforementioned optimization problem would require exponential time, since $\rho$ is an arbitrary exponential-size density matrix.

\paragraph{Optimizing low-degree polynomials over the sphere.} 
We now focus on the algorithmic problem of optimizing a low-degree polynomial over the sphere subject to $\ell_2$ and $\ell_\infty$ constraints.

Our results for polynomial optimization can be thought of as analogs of maximizing dense $k$-CSPs, only the domain is the sphere instead of the hypercube.
The underlying algorithms for max $k$-CSPs are either based on brute-force search over a dimension-independent number of variables followed by greedily completing the solution or global correlation rounding~\cite{mathieu2008yet,yaroslavtsev2014going, manurangsi2016birthday, alev2019approximating}. 
One may expect the correlation rounding algorithms for max $k$-CSPs to generalize straightforwardly to optimize low-degree polynomials over the sphere up to additive error.
However, the existing analysis~\cite{manurangsi2016birthday} would merely translate to outputting a product state with fidelity $\Omega(\OPT) - \eps$, as opposed to $\OPT - \eps$.
One can also try to extend the correlation rounding algorithm of Alev, Jeronimo and Tulsiani~\cite{alev2019approximating} to the sphere, but their algorithm obtains a doubly-exponential dependence on $k$.
In contrast, our approach is closer in spirit to the brute-force style algorithm for max $k$-CSPs, and allows for additional $\ell_2$ and $\ell_\infty$ constraints. 
Translating our algorithm back to optimizing dense polynomials over the hypercube, we can show that we obtain yet another algorithm that achieves additive error guarantees.

To get the key ideas across, we first consider the unconstrained polynomial optimization problem, reformulated as maximizing the injective norm of a tensor: 
\begin{equation*}
\max_{x \in \C^{m}, \norm{\vec{x}}_2 = 1}\langle T, \vec{x}^{\otimes k} \rangle 
\end{equation*}
for a tensor $T$ with $\norm{T}_F \leq 1$.  While it is hard to obtain a multiplicative approximation to tensor optimization problems, we show that we can obtain an additive  $\eps \cdot \norm{T}_F$ approximation in $n^{\poly(1/\eps)}$ time\ifstocproceedings\else\ (see \cref{thm:opt-main} for a formal statement)\fi. We begin by observing that there is a $\poly(1/\eps)$-dimensional subspace such that projecting $x$ onto this subspace suffices to obtain an $\eps \norm{T}_F$ approximation to the optimum objective value. To see this, we can unfold the tensor $T$ along the first mode to obtain a $m \times m^{k-1}$ matrix $M$. We can then compute the singular value decomposition of $M$ and let $\lambda_1 \geq \lambda_2 \geq \dots \geq \lambda_m\geq 0$ be the resulting singular values. There are at most $r = \ceil{1/\eps^2}$  singular values larger than $\eps \cdot \Norm{T}_F$. Let $\Pi^{(1)}_{>\eps}$ be the projection on the top-$r$ subspace of $M$. Then, 
\begin{equation*}
    \Abs{ \angles[\big]{ \vec{x} , M (\textsf{vec}(\vec{x}^{\otimes k-1}))} -  \angles[\big]{  \Pi^{(1)}_{>\eps } \vec{x} , M (\textsf{vec}(\vec{x}^{\otimes k-1})) } } \leq \eps \Norm{T}_F. 
\end{equation*}
Now, we can repeat the same argument for the remaining modes to obtain projectors $\Pi^{(2)}_{>\eps}, \dots , \Pi^{(k)}_{>\eps}$.
Setting $\Pi_{>\eps}$ to project onto the union of the spans of $\Pi^{(1)}, \dots , \Pi^{(k)}$, we conclude that
\[
    \Abs{ \angles[\big]{T , \vec{x}^{\otimes k }}  -  \angles[\big]{T , \Pi_{>\eps} \vec{x}^{\otimes k }} } \leq k \eps\,.
\]
In short, \emph{the polynomial can be approximated by projecting $\vec{x}$ onto a small-sized subspace.}
To solve this tensor optimization problem, it suffices to brute-force over a fine-enough net on this constant-dimensional subspace and pick the vector that obtains the maximal value.

In general, we need to deal with an optimization problem that involves a dimension-independent number of $\ell_2$ constraints and an $\ell_\infty$ constraint\ifstocproceedings\else\ (see \cref{def:opt-domain})\fi. We handle the $\ell_2$ constraints by simply projecting onto the union of the subspace $\Pi$ and the subspace corresponding to the span of the $\ell_2$ constraints. It remains to handle the $\ell_\infty$ constraint, which takes the form $\norm{\vec{x}}_{\infty} \leq \mu$ for some constant $\mu > 0$. Only $1/\mu^2$ of the coordinates can saturate the $\ell_{\infty}$ constraint.  Since $\mu > 0$ is a constant, we can brute-force over which coordinates saturate the $\ell_\infty$ constraint.  Ultimately, we can still reduce the constrained optimization problem to checking over a net in a constant-dimensional subspace. \ifstocproceedings\else We refer the reader to \cref{sec:opt} for more details\fi.

\paragraph{Hardness for agnostic product tomography.} 
Our lower bound starts from the hardness of computing (asymmetric) tensor spectral norm for $4$-tensors\ifstocproceedings\else\ (see Theorem~\ref{thm:tensor-pca-hardness})\fi.  In particular, for an $n \times n \times n \times n$ tensor $T$, computing the spectral norm to additive error $\norm{T}_F/\poly(n)$ is $\mathsf{NP}$-hard.
We attain our result by reducing tensor optimization to product state learning, essentially inverting the reduction discussed earlier.
The main idea is to set the unknown state $\rho = \proj{\psi}$ where 
\[
    \ket{\psi} = \frac{1}{\norm{T}_F} \sum_{i,j,k,l \in [m]} T_{ijkl} \ket{e_i} \ket{e_j} \ket{e_k} \ket{e_l} \,.
\]
Then, we can show that finding the $4m$-qubit product state 
\ifstocproceedings
\[
\ket{\pi_{\vec{x}}}\ket{\pi_{\vec{y}}}\ket{\pi_{\vec{u}}}\ket{\pi_{\vec{v}}}
\]
\else
$\ket{\pi_{\vec{x}}}\ket{\pi_{\vec{y}}}\ket{\pi_{\vec{u}}}\ket{\pi_{\vec{v}}}$
\fi
with optimal fidelity is essentially equivalent to maximizing the tensor form $\langle T, \vec{x} \otimes \vec{y} \otimes \vec{u} \otimes \vec{v} \rangle$.
The only additional difficulty is that this equivalence only holds if $T$ is sufficiently flat; our reduction thus requires an additional step where we embed our input $T$ in a larger space and randomly rotate it to make all its entries small without changing the optimal value.

\paragraph{Faster algorithms.}
In light of the lower bound, one can still ask what additional structure yields faster algorithms.
We consider three additional settings: the high-fidelity regime (high overlap with a product state), a bounded number of choices for each qubit, and matrix-product states.
In all of these settings, we follow the same overall strategy of sweeping over the qubits, but maintaining a cover becomes significantly easier:
\begin{enumerate}
    \item In the high-fidelity regime, the cover can be made to be only \emph{one} state;
    \item In the finite-choices setting, we can simply maintain \emph{all} good product states in the class, instead of a cover over them;
    \item In the MPS setting, we can make our cover one state, though one with a large bond dimension, in some sense capturing many good product states.
\end{enumerate}
For this overview, we focus on the high-fidelity setting.
Here, the optimal solution is unique and we do not require maintaining a net.
% We show that maintaining a cover over product states is unnecessary, because there is always a unique local optimum of product state fidelity.
Instead, we only need to maintain a single candidate product state as we sweep across the qubits, performing local updates until convergence.

To illustrate how and why local optimization works, suppose for simplicity that $\rho = \ketbra{\psi}{\psi}$ is a pure state; the mixed state case is similar but involves some additional parameters.
Imagine that $\ket{0^n}$ is the current candidate product state (in some appropriately chosen basis), and consider what happens when we express $\ket{\psi}$ in the low-Hamming weight subspace:
\[
\ket{\psi} = \alpha\ket{0^n} + \delta\ket{v_1} + \beta\ket{v_{\ge 2}}\,.
\]
Above, we assume without loss of generality that $\ket{v_1}$ is a normalized state on the subspace of Hamming weight $1$, $\ket{v_{\ge 2}}$ is a normalized state on the subspace of Hamming weight at least $2$, and $\alpha$, $\beta$, $\delta$ are all nonnegative reals.
It is helpful to express $\ket{\psi}$ this way because $\delta$ quantifies local updates that we can make to improve fidelity.
By rotating qubit $i$ of our candidate product state away from $\ket{0}$, we can increase the product state fidelity from $\abs{\braket{0^n|\psi}}^2 = \alpha^2$ to $\alpha^2 + \delta^2\abs{\braket{e_i|v_1}}^2$, where $\ket{e_i}$ is the string with $1$ in position $i$ and $0$s elsewhere.
Our goal, then, will be to establish that $\alpha^2$ is close to $\OPT$ whenever $\delta$ is small, because it implies that local optimization converges efficiently: either $\delta$ is large, in which case we can increase the candidate fidelity by a substantial amount, or $\delta$ is small, in which case we are near the global optimum.

We prove this by bounding the contributions to product state fidelity from the Hamming weight $0$, $1$, and $\ge 2$ subspaces separately.
Consider an arbitrary product state $\ket{\pi}$ that, when measured in the computational basis, gives $\ket{0^n}$ with probability $1 - p$ and anything else with probability $p$.
We observe \ifstocproceedings\else(\Cref{cor:p0_vs_p2}) \fi that $\ket{\pi}$ places at most $O(p^2)$ probability on Hamming weight $2$ or larger, and use this to upper bound the overlap between $\ket{\pi}$ and $\ket{\psi}$\ifstocproceedings\else\ (this is a simplified version of \Cref{eq:alpha_minus_p_plus_delta})\fi:
\[
\abs{\braket{\pi|\psi}} \le \alpha(1 - \Omega(p)) + \delta\sqrt{p} + O(\beta p)\,.
\]
Working out the constants, we find that so long as $\alpha^2 \ge 2/3$, the $-\Omega(\alpha p)$ term dominates the $O(\beta p)$ term, leaving us with $\abs{\braket{\pi|\psi}} \le \alpha + \delta\sqrt{p} \le \alpha + \delta$.
So, \textit{every} product state satisfies $\abs{\braket{\pi|\psi}}^2 \le (\alpha + \delta)^2$, and therefore $\OPT \le (\alpha + \delta)^2$.

The above analysis straightforwardly gives rise to a polynomial-time but suboptimal algorithm for finding the closest product state.
We briefly summarize the additional tricks that are required to reduce the sample complexity to linear in $n$.

First, we observe that divide-and-conquer is more efficient than sweeping through one additional qubit at a time.
So, the basic structure of the learning algorithm \ifstocproceedings\else(\Cref{alg:high_fidelity_alg}) \fi is:
\begin{enumerate}
    \item Recursively run the algorithm on the left and right halves of $\rho$, obtaining product states $\ket{\pi_L}$ and $\ket{\pi_R}$ that have fidelity at least $5/6$ with the respective halves.
    \item Take $\ket{\pi} = \ket{\pi_L} \otimes \ket{\pi_R}$, which satisfies $\braket{\pi|\rho|\pi} \ge 2/3$.
    \item Run local optimization on $\ket{\pi}$ until convergence.
\end{enumerate}

Second, we improve the bound on $\abs{\braket{\pi|\psi}}$ when $\alpha^2$ is much larger than $2/3$.
Namely, we show the alternative bound
\[
\abs{\braket{\pi|\psi}} \le \alpha + O\mparen{\frac{\delta^2}{\alpha^2 - 2/3}}\,,
\]
which ultimately implies that local optimization needs fewer iterations to converge.

Third, we find that one can make larger improvements to the fidelity by updating all $n$ qubits simultaneously, rather than one at a time.
We take $\vec{z} \in \C^n$ to be the vector defined by $z_i = \braket{e_i|\rho|0^n}$, which captures the mass that $\rho$ places on Hamming weight $1$ that is coherent with $\ket{0^n}$.
Then we show that a step from $\ket{0^n}$ to $\ket{\pi_{\vec{z} / 10}}$ increases the fidelity with $\rho$ by $\Omega\mparen{\norm{\vec{z}}_2^2}$.
Note that in the pure state case $\rho = \ketbra{\psi}{\psi}$, this $\vec{z}$ is precisely $\alpha\delta\ket{v_1}$, and therefore the fidelity improvement is $\Omega(\delta^2)$.
For comparison, recall that the improvement from a single-qubit update was only $\delta^2\max_{i \in [n]} \abs{\braket{e_i|v_1}}^2$, which can be as small as $\delta^2 / n$.

Finally, we establish that it suffices to obtain a relative $\ell^2$-error estimate of $\vec{z}$ in order to perform these local updates.
In symbols, if we can produce an estimate $\vec{a}$ satisfying $\norm{\vec{a} - \vec{z}}_2 \le \norm{\vec{z}}_2 / 3$ (say), then we show that $\ket{\pi_{\vec{a} / 10}}$ also increases the fidelity by $\Omega\mparen{\norm{\vec{z}}_2^2}$.
This allows us to cut down some of the cost of the tomography subroutine by varying the error parameter throughout the algorithm, because we can afford to be sloppier when the step size is large.

\ifstocproceedings
\else
\subsection{Notation}
Throughout, $\log_b$ is the logarithm base $b$, and $\log$ is shorthand for the natural logarithm base $e \approx 2.718$.
The complex conjugate of $z \in \C$ is denoted $z^*$.
For a function $f: \R \to \R$, $f'(x)$ denotes its derivative at $x$.
An unspecified polynomially-bounded function of $n$ may be written as $\poly(n)$.

We use standard shorthand for tensor products of quantum basis states. So, for example, $\ket{0}\otimes\ket{1}$ can be written either as $\ket{0}\ket{1}$ or $\ket{01}$.

$[n]$ is defined as the set of integers $\{1,2,\ldots,n\}$. If $S \subseteq [n]$, $\overline{S}$ denotes its complement in $[n]$.
When $\rho$ is an $n$-qubit mixed state, we write $\rho_S \coloneqq \tr_{\overline{S}}(\rho)$ for the reduced state on the qubits indexed by $S$.

The Hamming weight of a binary string $x \in \{0,1\}^n$ is denoted by $|x|$. We write $\ket{e_i}$ for the Hamming weight-$1$ $n$-qubit computational basis state that has $\ket{1}$ in the $i$th position and $\ket{0}$ everywhere else, leaving $n$ implicit from context.

Vectors always have arrows over them (e.g.\ $\vec{v}$), unless they represent a (pure) quantum state, in which case we use ket notation (e.g.\ $\ket{\psi}$).
If $\vec{v} \in \C^n$, $\vec{v}^* \in \C^n$ is its entrywise complex conjugate.
For matrices $A \in \C^{m \times n}$, the conjugate transpose is $A^\dagger \in \C^{n \times m}$.
The Euclidean (or $\ell^2$) norm of a vector $\vec{v} \in \C^n$ is denoted by $\twonorm{\vec{v}} \coloneqq \sqrt{\sum_i \abs{v_i}^2}$, and the $\ell^\infty$ norm is denoted by $\infnorm{\vec{v}} \coloneqq \max_i \abs{v_i}$.
The norms that we use for matrices $A \in \C^{m \times n}$ are the operator norm $\norm{A}_{\op} \coloneqq \sup_{\vec{v} \neq 0} \frac{\norm{A\vec{v}}_2}{\norm{\vec{v}}_2}$, the trace norm $\norm{A}_1 \coloneqq \tr\mparen{\sqrt{A^\dagger A}}$, and the Frobenius norm $\norm{A}_F \coloneqq \sqrt{\tr\mparen{A^\dagger A}}$.
We also use $\norm{A}_F$ for the Frobenius norm of a tensor $A$, which is the square root of the sum of the squared magnitudes of the entries.
\fi
%\will{TODO: Frobenius norm of tensor?}

% \ewin{Global notation:}
% \begin{itemize}
%     \item $n$ is the number of qubits
%     \item $m$ is the number of qubits while sweeping
%     \item $\rho_{[m]}$ is the version of $\rho$ with qubits $m+1$ through $n$ traced out
%     \item $\eta$ is the maximum fidelity
%     \item $\eps$ is the error allowed from the max fidelity
%     \item vectors have arrows on them, $\vec{z}$
%     \item $b$ and $B$ are bounds on the product state cover
%     \item $r$ is the number of things in the net currently (and the number of constraints)
%     \item $S$ is the set of indices for which $\vec{z}$ is large in the reduction step to the optimization problem.
%     \item $d$ is the degree of the polynomials (and the bound on hamming weight strings)
%     \item $z^*$ is the conjugate of $z$, $\overline{S}$ is the set complement of $S$
%     \item $|x|$ is the Hamming weight of a string $x$
%     \item $\ket{e_i}$ is the state with $1$ in position $i$ and $0$s elsewhere
% \end{itemize}

% Need to make consistent:
% \begin{itemize}
%     \item Arrows for vectors
%     \item $\norm{v}_2$ for Euclidean norm (instead of $\norm{v}$)
%     \item $\norm{A}_{\op}$ for spectral norm (instead of $\norm{A}$ or $\infnorm{A}$
%     \item $\norm{A}_1$ for trace norm
%     \item $\log$ for natural logarithm (instead of $\ln$)
% \end{itemize}

% \will{low priority consistency items}
% \begin{itemize}
%     \item $\braket{\pi|\rho|\pi}$ vs.\ $\bra{\pi}\rho\ket{\pi}$
% \end{itemize}

\section{Parametrization of product states}

\begin{definition}[Product state parametrization]
    \label{def:product_state_parametrization}
    A product state $\ket{\pi}$ can be described by $n$ parameters $\Vec{z} \in (\C \cup \{\infty\})^n$ in the extended complex plane.
    We use the notation
    \begin{align*}
        \ket{\pi_{\Vec{z}}} = \bigotimes_{i=1}^n \frac{\ket{0} + z_i\ket{1}}{\sqrt{1 + \abs{z_i}^2}}\,.
    \end{align*}
\end{definition}
Note that this parametrization normalizes global phase such that the $\ket{0^n}$ amplitude, or more generally, the non-zero amplitude with lowest Hamming weight, is real.
Though we define this parametrization for $\vec{z}$ with entries in the extended complex plane, we will generally work with $\vec{z} \in \C^n$ for simplicity: our algorithms incur error, so we will always be able to work with vectors with finite (but possibly very large) entries.
However, this limitation is not necessary, and the diligent reader can extend our results to hold even with parameters at infinity.

\subsection{Tangent distance}
\label{sec:tan-distance}

We now define a distance measure between product states which behaves nicely with respect to our parametrization, which we call the \emph{tangent distance}.  For the sake of building intuition, we first consider the case when $n = 1$:
\begin{definition}[Tangent distance, $n=1$] \label{def:tangent-n1}
    For $z, a \in \C \cup \{\infty\}$, we define the following distance between $1$-qubit states:
    \begin{align*}             
        \dtan(\ket{\pi_{z}},  \ket{\pi_{a}}) = \abs[\Big]{\frac{z_i - a_i}{1 + z_i^* a_i}}\,.
    \end{align*}
    We also have that $\dtan^2(\ket{\pi_{z}},  \ket{\pi_{a}}) = \tan^2(\theta/2)$ where $\theta \in [0,\pi]$ is the angle between the Bloch sphere points for $\ket{\pi_{z}}$ and $\ket{\pi_{a}}$.
\end{definition}
\begin{remark}
    Presumably the Bloch sphere interpretation is well known but we could not find a reference.  
    Here is a quick proof: The Bloch sphere point for $\ket{\pi_z}$ is the inverse stereographic projection of~$z$; namely $\vec{p}_z \coloneqq \frac{1}{1+|z|^2}(z+z^*, z-z^*, 1-|z|^2)$.  
    From this it follows that $\cos \theta = \vec{p}_z \cdot \vec{p}_a = 1 - \frac{2|z-a|^2}{(1+|z|^2)(1+|a|^2)}$.
    Now use $\tan^2(\theta/2) = \frac{1-\cos\theta}{1+\cos\theta}$.
\end{remark}

\begin{definition}[Tangent distance, general $n$]
\label{def:tangent-distance}
    For two product states $\ket{\pi_{\Vec{z}}}$ and $\ket{\pi_{\Vec{a}}}$, we define the distance measure $\dtan(\ket{\pi_{\vec{z}}}, \ket{\pi_{\vec{a}}})$ via
    \begin{align*}
        \dtan^2(\ket{\pi_{\vec{z}}}, \ket{\pi_{\vec{a}}}) = \sum_{i=1}^n \dtan(\ket{\pi_{z_i}}, \ket{\pi_{a_i}})^2 = \norm[\Big]{\frac{\Vec{z} - \Vec{a} }{1 + \Vec{z}^*\Vec{a}}}_2^2\,.%= \parens[\Big]{\sum_{i=1}^n \abs[\Big]{\frac{z_i - a_i}{1 + z_i^* a_i}}^2}^{1/2}
    \end{align*}
    Above, we abuse notation by using product and quotient of vectors to denote their entry-wise product and quotient.
\end{definition}
This distance measure has several nice properties: it satisfies $\dtan(\ket{\pi_{\vec{z}}}, \ket{\pi_{\vec{z}}}) = 0$; it is symmetric; and it is invariant under applying single-qubit unitaries.
This last condition is evident from the Bloch sphere interpretation.
However, it is not a metric: the triangle inequality does not hold, since the distance measure is infinity for orthogonal product states:
\begin{align*}
    \dtan(\ket{+}, \ket{0}) + \dtan(\ket{0}, \ket{-}) = 2 < \dtan(\ket{+}, \ket{-}) = \infty\,.
\end{align*}
Nevertheless, we can think of this as being approximately a metric when the denominator is ``close to constant'', i.e.\ when the two product states, and therefore their parameters, are close.

\begin{remark}[Generalizing to qudits] \label{rmk:qudit}
    We can take the following parametrization over qudits, which takes $(d-1)n$ parameters $\bz \in (\C \cup \{\infty\})^{(d-1)n}$:
    \begin{equation*}
        \ket{\pi_{\bz}} = \bigotimes_{i = 1}^{n} \frac{\ket{0} + \sum_{j = 1}^{d} z_{i, j} \ket{j}}{\sqrt{1 + \norm{\vec{z}_i}^2_2}}\,.
    \end{equation*}
    Here we imagine partitioning $\bz$ into $n$ many vectors $\vec{z_i}$ of size $d-1$.
    The qudit version of tangent distance becomes
    \begin{equation*}
        \dtan(\ket{\pi_{\bz}}, \ket{\pi_{\ba}}) = \parens[\Big]{\sum_{i=1}^n \frac{\norm{\vec{z}_i-\vec{a}_i}^2 + \twonorm{\vec{z}_i}^2\twonorm{\vec{a}_i}^2 - \abs{\angles{\vec{z}_i, \vec{a}_i}}^2}{\abs{1 + \angles{\vec{z}_i, \vec{a}_i}}^2}}^{1/2},
    \end{equation*}
    which we can see by rotating every qudit to reduce to the qubit definition.
    Notice that this reduces to the qubit definition when $d = 2$.
    %For simplicity, we only show how to perform this algorithm over qubits.
    Though not immediate, we anticipate no barriers in generalizing our results to qudits.
    Generalized versions of the lemmas to follow hold, and we can run the algorithms and analyses accordingly.
    %Though tangent distance is no longer approximated by $\twonorm{\bz - \ba}$, we believe such constraints can still be handled by the techniques given in this section to reduce to problems which can be solved by \cref{thm:opt-main}.
\end{remark}

We define tangent distance to be a version of fidelity which behaves more nicely with respect to the parametrization.
Tangent distance is related to fidelity in the following way.

\begin{lemma}[Relationship between tangent distance and fidelity]\label{lem:dtan_fidelity_relationship}
    For all product states $\ket{\pi_{\vec{z}}}$ and $\ket{\pi_{\vec{a}}}$, the following holds:
    \begin{equation*}
        \log\left(\frac{1}{\abs{\braket{\pi_{\vec{z}} | \pi_{\vec{a}}}}^2}\right) \leq \dtan(\ket{\pi_{\vec{z}}}, \ket{\pi_{\vec{a}}})^2 \leq \frac{1}{\abs{\braket{\pi_{\vec{z}} | \pi_{\vec{a}}}}^2} - 1\,.
    \end{equation*}
    Both inequalities are tight.
\end{lemma}
\begin{proof}
    First, because both fidelity and $\dtan$ are invariant under single-qubit unitaries, %(\cref{lem:dtan_invariant}), 
    it suffices to show the inequalities for $\ket{\pi_{\vec{a}}} = \ket{0^n}$.
    Then, $\abs{\braket{\pi_{\vec{z}} | 0^n}}^2 = \prod_{i=1}^n \frac{1}{1+\abs{z_i}^2}$ and $\dtan(\ket{\pi_{\vec{z}}}, \ket{0^n})^2 = \norm{\vec{z}}_2^2$.
    Consequently,
    \begin{align*}
        \log\left(\frac{1}{\abs{\braket{\pi_{\vec{z}} | 0^n}}^2}\right)
        = \sum_{i=1}^n \log(1 + \abs{z_i}^2)
        &\leq \sum_{i=1}^n \abs{z_i}^2
        = \dtan(\ket{\pi_{\vec{z}}}, \ket{0^n})^2
        \\
        &\leq \prod_{i=1}^n (1 + \abs{z_i}^2) - 1
        = \frac{1}{\abs{\braket{\pi_{\vec{z}} | 0^n}}^2} - 1\,.
    \end{align*}
    The first inequality is tight for $\vec{z} = 0^n$ and the right inequality is tight for $\vec{z} = (\infty, 0^{n-1})$.
\end{proof}

The approximation in \cref{lem:dtan_fidelity_relationship} can be made tighter when $\vec{z}$ and $\vec{a}$ are closer together.
For simplicity, we state the following lemma when $\vec{a} = 0^n$.

\begin{lemma}[Approximation of $\dtan$ at small distances]
    \label{lem:dtan-fidelity-approx}
    For a vector $\vec{z} \in \C^n$,
    \begin{align*}
        e^{-\dtan(\ket{\pi_{\vec{z}}}, \ket{0^n})^2}
        \leq  \prod_{i=1}^n \frac{1}{1 + \abs{z_i}^2}
        \leq e^{-\dtan(\ket{\pi_{\vec{z}}}, \ket{0^n})^2 + \sum_{i=1}^n \abs{z_i}^4}\,.
    \end{align*}
    Note that $\prod_{i=1}^n \frac{1}{1 + \abs{z_i}^2} = \abs{\braket{\pi_{\vec{z}}|0^n}}^2$.
    %\ewin{This can also be made into the stronger bound $e^{-\dtan(\ket{\pi_{\vec{z}}}, \ket{0^n})^2} + \sum_{i=1}^n \abs{z_i}^4$, but I don't think we need this.}
\end{lemma}
\begin{proof}
The first inequality is equivalent to the first inequality in \cref{lem:dtan_fidelity_relationship}.
For the second inequality, we use that, for $x \geq 0$,
\begin{align*}
    e^{-x} \leq \frac{1}{1 + x} \leq 1 - x + x^2 \leq e^{-x + x^2}\,.
\end{align*}
So, this gives us that 
\begin{equation*}
    \prod_{i=1}^n \frac{1}{1 + \abs{z_i}^2}
    \leq \prod_{i=1}^n e^{-\abs{z_i}^2 + \abs{z_i}^4}
    = e^{-\dtan(\ket{\pi_{\vec{z}}}, \ket{0^n})^2 + \sum_{i=1}^n \abs{z_i}^4}.
    \qedhere
\end{equation*}
\end{proof}

As a corollary, we can also relate tangent distance to trace distance.

\begin{corollary} \label{cor:dtan_trace_relationship}
    For $\vec{z}, \vec{a} \in \C^n$,
    $\frac12\norm{\proj{\pi_{\vec{z}}} - \proj{\pi_{\vec{a}}}}_1
    = \opnorm{\proj{\pi_{\vec{z}}} - \proj{\pi_{\vec{a}}}}
        \leq \dtan(\ket{\pi_{\vec{z}}}, \ket{\pi_{\vec{a}}})$.
\end{corollary}
\begin{proof}
The equality holds because $\proj{\pi_{\vec{z}}} - \proj{\pi_{\vec{a}}}$ is traceless and rank two.
The inequality holds because
\begin{align*}
    \frac12\norm{\proj{\pi_{\vec{z}}} - \proj{\pi_{\vec{a}}}}_1
    &\leq \sqrt{1 - \abs{\braket{\pi_{\vec{z}}|\pi_{\vec{a}}}}^2} \\
    &\leq \sqrt{1 - e^{-\dtan(\ket{\pi_{\vec{z}}}, \ket{\pi_{\vec{a}}})^2}} \\
    &\leq \dtan(\ket{\pi_{\vec{z}}}, \ket{\pi_{\vec{a}}})\,,
\end{align*}
where the first inequality is Fuchs--van de Graaf, the second is \cref{lem:dtan_fidelity_relationship}, and the third is $1 + x \leq e^x$.
\end{proof}

We can further simplify $\dtan$ to a simple $\ell^2$ norm when $\vec{z}$ and $\vec{a}$ are close entry-wise.
\begin{lemma} \label{lem:dtan-l2}
    Let $\vec{z}, \vec{a} \in \C^n$.% be vectors such that $\max_i \abs{z_i} \abs{a_i} \leq 1/2$.
    Then
    \begin{align*}
        \abs[\Big]{\dtan(\ket{\pi_{\vec{z}}}, \ket{\pi_{\vec{a}}})
        - \twonorm{\vec{z} - \vec{a}}}
        \leq \dtan(\ket{\pi_{\vec{z}}}, \ket{\pi_{\vec{a}}}) (\max_i \abs{z_i}\abs{a_i})\,.
    \end{align*}
\end{lemma}
\begin{proof}
Consider a single qubit $i$.
Then,%\will{where do we use this bound?}
\begin{align} \label{eq:dtan-l2}
    \abs[\Big]{\frac{z_i - a_i}{1 + z_i^*a_i} - (z_i - a_i)}
    = \abs[\Big]{\frac{z_i - a_i}{1 + z_i^*a_i}}\abs{1 - (1 + z_i^*a_i)}
    = \abs[\Big]{\frac{z_i - a_i}{1 + z_i^*a_i}}\abs{z_i}\abs{a_i}\,.
\end{align}
So, we can conclude
\begin{equation*}
    \abs[\Big]{\dtan(\ket{\pi_{\vec{z}}}, \ket{\pi_{\vec{a}}}) - \twonorm{\vec{z} - \vec{a}}}
    \leq \twonorm[\Big]{\frac{\vec{z} - \vec{a}}{1 + \vec{z}^*\vec{a}} - (\vec{z} - \vec{a})}
    \leq \dtan(\ket{\pi_{\vec{z}}}, \ket{\pi_{\vec{a}}}) (\max_i \abs{z_i}\abs{a_i})\,. \qedhere
\end{equation*}
\end{proof}

\subsection{Approximation lemmas}

\begin{lemma}[Low-weight truncation of product states]
\label{lem:low_weight}
    For a product state $\ket{\pi_{\vec{z}}}$, let $\mu = \sum_{i=1}^n \frac{\abs{z_i}^2}{1 + \abs{z_i}^2}$.
    Then for $c \leq \mu$ and $d \geq \mu$,
    \begin{align*}
        \twonorm{\Pi_{\geq d} \ket{\pi_{\vec{z}}}}^2 &\leq e^{-d \log(d/\mu) + (d - \mu)}\,, \\
        \twonorm{\Pi_{\leq c} \ket{\pi_{\vec{z}}}}^2 &\leq e^{-(2\mu - c) \log(2 - c/\mu) + (\mu - c)}\,,
    \end{align*}
    where $\Pi_{\geq d}$ is the projection onto computational basis states $\ket{b}$ such that $\abs{b} \geq d$, and similarly for $\Pi_{\leq c}$.
\end{lemma}
\begin{proof}
Let $p_i = \frac{\abs{z_i}^2}{1 + \abs{z_i}^2}$ be the probability that qubit $i$ of $\ket{\pi_{\vec{z}}}$ outputs $\ket{1}$ when measured in the computational basis.
Consider $n$ independent Bernoulli random variables $Z_1, \dots, Z_n$, where $\Pr[Z_i = 1] = p_i$.
Then, because the qubits of $\pi_{\vec{z}}$ are uncorrelated,
\begin{align*}
    \twonorm{\Pi_{\geq k} \ket{\pi_{\vec{z}}}}^2 = \Pr[Z_1 + \dots + Z_n \geq k]\,.
\end{align*}
Let $\mu = \E[Z_1 + \dots + Z_n] = p_1 + \dots + p_n$.
By Bennett's inequality~\cite[Theorem 2.9]{blm13}, for any $t \geq 0$,
\begin{align*}
    \Pr\bracks[\Big]{\sum_{i=1}^n (Z_i - p_i) \geq t}
    \leq e^{-\mu ((1 + t/\mu) \log(1 + t/\mu) - t/\mu)}
    = e^{-(t + \mu) \log(1 + t/\mu) + t}\,.
\end{align*}
The same bound holds on the other tail:
\begin{align*}
    \Pr\bracks[\Big]{\sum_{i=1}^n (Z_i - p_i) \leq -t}
    \leq e^{-(t + \mu) \log(1 + t/\mu) + t}\,.
\end{align*}
The statement follows upon taking $d = \mu + t$ in the first inequality and $c = \mu - t$ in the second.
\end{proof}

\section{High-fidelity product state learning}

In this section, we give a simple polynomial-time algorithm for product state agnostic learning in the high signal-to-noise regime, where the fidelity of the unknown state $\rho$ with the closest product state $\ket{\pi}$ is sufficiently close to $1$. The algorithm operates by local optimization, where a candidate product state approximation is updated on the Hamming weight-$1$ subspace until convergence. We argue this algorithm's correctness by showing that once all of the local updates are sufficiently small, the algorithm must be close to a global optimum. 
The analysis involves bounding the optimal product state fidelity $\max_{\text{product }\ket{\pi}} \braket{\pi|\rho|\pi}$ in terms of the projection of $\rho$ onto the subspace of Hamming weight $0$ or $1$.

\subsection{Properties of product distributions}

We begin by showing some simple concentration bounds on the Hamming weight of a product distribution.

\begin{lemma}
\label{lem:p0_vs_p1}
    Let $\pi$ be a product distribution over $\{0,1\}^n$, and let $p_0 = \Pr_{x \sim \pi}[x = 0^n]$. Then
    \[
    \Pr_{x \sim \pi}[|x| = 1] \ge -p_0 \log p_0\,,
    \]
    with the understanding that $0 \log 0 = 0$.
\end{lemma}

\begin{proof}
    Let $a_i = \Pr_{x \sim \pi}[x_i = 0]$. If any $a_i$ are equal to $0$, we have $p_0 = 0$ and the lemma is trivial. Otherwise, assuming all $a_i$ are strictly positive,
    \begin{align*}
        \Pr_{x \sim \pi}[|x| = 1]
        &= \sum_{i=1}^n (1 - a_i) \prod_{j \neq i} a_j\\
        &= \sum_{i=1}^n \frac{1 - a_i}{a_i} \prod_{j=1}^n a_j\\
        &= p_0 \sum_{i=1}^n \frac{1 - a_i}{a_i}\\
        &\ge p_0 \sum_{i=1}^n \log(1/a_i) && (x - 1 \ge \log x)\\
        &= -p_0 \log\left(\prod_{i=1}^n a_i \right)\\
        &= -p_0 \log p_0\,. \tag*{\qedhere}
    \end{align*}
\end{proof}

% \begin{lemma}
% \label{lem:questionable_inequality}
%     For all $\sqrt{2/3} \le \alpha \le 1$ and $0 \le p \le 1$, it holds that
%     \[
%     \alpha \sqrt{p} + \sqrt{1 - \alpha^2}\sqrt{1 - p + p \log p} \le \alpha.
%     \]
% \end{lemma}
% \begin{proof}
%     IDK, it works numerically...

%     ALTERNATIVELY: I think we could just show that the bound in \Cref{cor:p0_vs_p2} is at most $2(1 - \sqrt{p_0})^2$.
% \end{proof}

Next, we prove a useful upper bound on the quantity $p_0 \log p_0$ appearing in the previous lemma.

\begin{lemma}
\label{lem:weird_ln_sqrt}
    %\ewin{Here's the thing I can prove.}
    For all $p \in [0, 1]$,
    \[
        1 - p + p\log p \le 2(1 - \sqrt{p})^2\,,
    \]
    with the understanding that $0 \log 0 = 0$.
\end{lemma}

\begin{proof}
    Define
    \[
    f(p) \coloneqq 2(1 - \sqrt{p})^2 - 1 + p - p \log(p)\,.
    \]
    Notice that $f(1) = 0$. Moreover, the derivative of $f$ satisfies:
    \[
        f'(p) = 2 - \frac{2}{\sqrt{p}} - \log p \le 0\,,
    \]
    because $1 - \frac{1}{x} \le \log(x)$ for all $x > 0$ (substituting $x = \sqrt{p}$). So, $f$ is decreasing on $[0, 1]$, and therefore $f(p) \ge 0$ for all $p \in [0, 1]$, which implies the lemma.
\end{proof}

Combining the previous two lemmas gives an upper bound on the probability assigned by a product distribution to strings of Hamming weight at least $2$.

\begin{corollary}
\label{cor:p0_vs_p2}
    Let $\pi$ be a product distribution over $\{0,1\}^n$, and let $p_0 = \Pr_{x \sim \pi}[x = 0^n]$. Then
    \[
    \Pr_{x \sim \pi}[|x| \ge 2] \le 2(1 - \sqrt{p_0})^2\,.
    \]
\end{corollary}
\begin{proof}
    Follows from \Cref{lem:p0_vs_p1} and \Cref{lem:weird_ln_sqrt}.
\end{proof}

\subsection{Characterizing optimal product approximations}

Proved in this section is the theorem below, which is the heart of the algorithm's correctness.

\begin{theorem}
\label{thm:local_opt_works_better}
Consider an arbitrary quantum state $\ket{\psi}$, which we express in the form
\begin{align*}
    \ket{\psi} = \mparen{\alpha \ket{0^n} + \delta \ket{v_1} + \beta \ket{v_{\ge 2}}}\ket{g} + \gamma\ket{\perp}\,.
\end{align*}
Assume without loss of generality that
\begin{itemize}
    \item $\alpha$, $\beta$, $\delta$, and $\gamma$ are all real and nonnegative,
    \item $\ket{v_1}$, $\ket{v_{\ge 2}}$, $\ket{g}$, and $\ket{\perp}$ are unit vectors,
    \item $\ket{v_1}$ is supported only on strings of Hamming weight $1$,
    \item $\ket{v_{\ge 2}}$ is supported only on strings of Hamming weight $2$ or larger, and
    \item $\ket{\perp}$ is orthogonal to all states beginning with $\ket{0^n}$ or ending with $\ket{g}$.
\end{itemize}
Suppose further that $\alpha^2 = 2/3 + c$ for some $c \ge 0$.
Then for all product states $\ket{\pi}$ on $n$ qubits:
\begin{align*}
    \norm{(\bra{\pi} \otimes I)\ket{\psi}}_2^2 \leq \mparen{\alpha + \min\mbrace{\delta, \sqrt{\frac{2}{27}}\frac{\delta^2}{c}}}^2\,.
\end{align*}
\end{theorem}

Before establishing this theorem, let us explain how to interpret it. Suppose we are searching for the best product state approximation to the leftmost $n$ qubits of some state $\ket{\psi}$. Our current candidate for the best product state is $\ket{0^n}$ (in some appropriately chosen product basis), whose fidelity with the leftmost $n$ qubits is currently $\alpha^2 = 2/3 + c$.

First consider the special case where $\ket{\psi}$ itself is an $n$-qubit pure state, in which case we may take $\ket{g} = 1$ and $\gamma = 0$ to write
\[
\ket{\psi} = \alpha \ket{0^n} + \delta\ket{v_1} + \beta \ket{v_{\ge 2}}
\]
as a sum over basis states of Hamming weight $0$, $1$, and $2$ or greater.
Then, the theorem says that if $\ket{\psi}$'s support on Hamming weight $1$ is small (as captured by $\delta$), $\ket{0^n}$ must approximately maximize the product state fidelity with $\ket{\psi}$.

In the general case, where $\ket{\psi}$ has more than $n$ qubits, the theorem similarly shows that $\ket{0^n}$ is an approximate maximizer of product fidelity, but under a slightly different assumption: that $\ket{\psi}$ places small support on states of Hamming weight $1$ \textit{that are coherent with $\ket{0^n}$ on the leftmost $n$ qubits}.
In other words, $\ket{\psi}$ may place large mass on Hamming weight $1$ in $\ket{\perp}$, but this does not affect the bound on product state fidelity.

Now we proceed towards the proof. We first establish two quantitative bounds, whose importance will become clear in the proof of \Cref{thm:local_opt_works_better}.

\begin{lemma}
    \label{lem:p_beta_gamma_r}
    Let $p$, $\beta$, $\gamma$, and $r$ be nonnegative reals satisfying $0 \le p \le 1$, $\beta^2 + \gamma^2 \le 1/3$, and $r \ge \sqrt{2/3}$. Then:
    \[
    \mparen{r\mparen{1 - \frac{p}{2}} + \frac{\beta\sqrt{2}}{2}p}^2 + \gamma^2 p \le \mparen{r\mparen{1 - \frac{p}{2}} + \frac{\sqrt{\beta^2 + \gamma^2}\sqrt{2}}{2}p}^2\,.
    \]
\end{lemma}

\begin{proof}
    %The inequality is easily verified when any of $\beta$, $\gamma$, or $r$ are $0$, so we assume henceforth that they are all strictly positive.
    Expanding out, the desired inequality becomes
    \begin{multline*}
    r^2\mparen{1 - \frac{p}{2}}^2 + \beta\sqrt{2}rp\mparen{1 - \frac{p}{2}} + \frac{\beta^2}{2}p^2 + \gamma^2 p \\ \le r^2\mparen{1 - \frac{p}{2}}^2 + \sqrt{\beta^2 + \gamma^2}\sqrt{2}rp\mparen{1 - \frac{p}{2}} + \frac{\beta^2 + \gamma^2}{2}p^2\,,
    \end{multline*}
    which considerably simplifies to
    % \[
    % \beta\sqrt{2}rp\mparen{1 - \frac{p}{2}} + \gamma^2 p \le \sqrt{\beta^2 + \gamma^2}\sqrt{2}rp\mparen{1 - \frac{p}{2}} + \frac{\gamma^2}{2}p^2.
    % \]
    \[
    \gamma^2 p\mparen{1 - \frac{p}{2}} \le \sqrt{\beta^2 + \gamma^2}\sqrt{2}rp\mparen{1 - \frac{p}{2}} - \beta\sqrt{2}rp\mparen{1 - \frac{p}{2}}\,.
    \]
    Factoring out $p\mparen{1 - \frac{p}{2}} \ge 0$, it suffices to show that
    \[
    \gamma^2 \le r\sqrt{2}\mparen{\sqrt{\beta^2 + \gamma^2} - \beta}\,.
    \]
    We can assume henceforth that $\beta$ and $\gamma$ are both strictly positive, because the inequality above is easily verified if either are zero.
    Multiplying both sides by $\sqrt{\beta^2 + \gamma^2} + \beta$ is equivalent to
    \[
    \gamma^2\mparen{\sqrt{\beta^2 + \gamma^2} + \beta} \le r\sqrt{2}\gamma^2\,,
    \]
    and then we divide by $\gamma^2$ to obtain
    \[
    \sqrt{\beta^2 + \gamma^2} + \beta \le r\sqrt{2}\,.
    \]
    This is implied by the assumptions of the lemma because the left side is at most $\frac{2}{\sqrt{3}}$, and the right side is at least $\frac{2}{\sqrt{3}}$.
\end{proof}

\begin{lemma}
\label{lem:alpha_c}
Suppose that $\alpha^2 = 2/3 + c$ for some $c \ge 0$. Then:
\[
\alpha - \sqrt{2 - 2\alpha^2} \ge c\sqrt{\frac{27}{8}}\,.
\]
\end{lemma}

\begin{proof}
    In other words, we wish to show that
    \[
    f(\alpha) \coloneqq \alpha - \sqrt{2 - 2\alpha^2} - \sqrt{\frac{27}{8}}\mparen{\alpha^2 - \frac{2}{3}} \ge 0\,.
    \]
    Consider the first and second derivatives of $f$:
    \begin{align*}
    f'(\alpha) = 1 + \frac{2\alpha}{\sqrt{2 - 2\alpha^2}} - \sqrt{\frac{27}{2}}\alpha,
    &&
    f''(\alpha) = \frac{\sqrt{2}}{(1-\alpha^2)^{3/2}} - \sqrt{\frac{27}{2}}\,.
    \end{align*}
    A simple calculation shows that $f''(\sqrt{2/3}) > 0$, and $f''(\alpha)$ is clearly increasing in $\alpha$, so $f''(\alpha) > 0$ for all $\alpha \ge \sqrt{2/3}$. $f'(\sqrt{2/3}) = 0$, and $f'(\alpha)$ is increasing in $\alpha$ as a consequence of the positive second derivative, so $f'(\alpha) \ge 0$ for all $\alpha \ge \sqrt{2/3}$. Hence, $f$ is non-decreasing for $\alpha \ge \sqrt{2/3}$. Since $f(\sqrt{2/3}) = 0$, the lemma follows.
\end{proof}

\begin{proof}[{Proof of \Cref{thm:local_opt_works_better}}]
    Write
    \[
    \ket{\pi} = \sqrt{p_0} \ket{0^n} + \sqrt{p_1} \ket{w_1} + \sqrt{p_{\ge 2}}\ket{w_{\ge 2}}\,,
    \]
    where $p_0$, $p_1$, and $p_{\ge 2}$ are probabilities summing to $1$, $\ket{w_1}$ is supported over strings of weight $1$, $\ket{w_{\ge 2}}$ is supported over strings of weight at least $2$, and both $\ket{w_1}$ and $\ket{w_{\ge 2}}$ are unit vectors.
    Then:
    \begin{align*}
        \norm{(\bra{\pi} \otimes I)\ket{\psi}}_2^2
        &= \norm{\mparen{\alpha\sqrt{p_0} + \delta\sqrt{p_1}\braket{w_1|v_1} + \beta\sqrt{p_{\ge 2}}\braket{w_{\ge 2}|v_{\ge 2}}}\ket{g} + \gamma(\bra{\pi} \otimes I)\ket{\perp}}_2^2\\
        &= \norm{\mparen{\alpha\sqrt{p_0} + \delta\sqrt{p_1}\braket{w_1|v_1} + \beta\sqrt{p_{\ge 2}}\braket{w_{\ge 2}|v_{\ge 2}}}\ket{g}}_2^2 + \norm{\gamma(\bra{\pi} \otimes I)\ket{\perp}}_2^2\\
        &= \abs{\alpha\sqrt{p_0} + \delta\sqrt{p_1}\braket{w_1|v_1} + \beta\sqrt{p_{\ge 2}}\braket{w_{\ge 2}|v_{\ge 2}}}^2 + \gamma^2\norm{(\bra{\pi} \otimes I)\ket{\perp}}_2^2\\
        &\le \mparen{\alpha \sqrt{p_0} + \delta\sqrt{p_1} + \beta\sqrt{p_{\ge 2}}}^2 + \gamma^2(1 - p_0)\,,
    \end{align*}
    where in the second line we used the Pythagorean theorem which is valid because $\ket{\perp}$ has no support on $\ket{g}$, and in the last line we applied the triangle inequality and the assumption that $\ket{\perp}$ has no support on $\ket{0^n}$.
    We first turn our attention to bounding (the square root of) the left term. Let $p = p_1 + p_{\ge 2}$ be the probability that the measurement distribution of $\ket{\pi}$ assigns to strings of Hamming weight $1$ or more. Then:
    \begin{align*}
        \alpha \sqrt{p_0} + \delta\sqrt{p_1} + \beta\sqrt{p_{\ge 2}}
        &\le \alpha \sqrt{p_0} + \delta\sqrt{p_1} + \beta\sqrt{2}(1 - \sqrt{p_0}) && (\mathrm{\Cref{cor:p0_vs_p2}})\\
        &= (\alpha - \beta\sqrt{2})\sqrt{p_0} + \delta\sqrt{p_1} + \beta\sqrt{2}\\
        &\le (\alpha - \beta\sqrt{2})\sqrt{1 - p} + \delta\sqrt{p} + \beta\sqrt{2} && (p_0 = 1-p, p_1 \le p)\\
        &\le (\alpha - \beta\sqrt{2})\mparen{1 - \frac{p}{2}} + \delta\sqrt{p} + \beta\sqrt{2} && (\sqrt{1 - p} \le 1 - \frac{p}{2}, \alpha \ge \beta\sqrt{2})\\
        &= \alpha\mparen{1 - \frac{p}{2}} + \frac{\beta\sqrt{2}}{2}p + \delta\sqrt{p}\,.
    \end{align*}
    Substituting, we find that
    \begin{align}
        \norm{(\bra{\pi} \otimes I)\ket{\psi}}_2^2
        &\le \mparen{\alpha\mparen{1 - \frac{p}{2}} + \frac{\beta\sqrt{2}}{2}p + \delta\sqrt{p}}^2 + \gamma^2 p\nonumber\\
        &\le \mparen{\alpha\mparen{1 - \frac{p}{2}} + \frac{\sqrt{\beta^2 + \gamma^2}\sqrt{2}}{2}p + \delta\sqrt{p}}^2 && (\mathrm{\Cref{lem:p_beta_gamma_r}})\nonumber\\
        &\le \mparen{\alpha - \frac{p}{2}\mparen{\alpha - \sqrt{2 - 2\alpha^2}} + \delta\sqrt{p}}^2 && (\alpha^2 + \beta^2 + \gamma^2 \le 1)\nonumber\\
        &\le \mparen{\alpha - p\sqrt{\frac{27}{32}}c + \delta\sqrt{p}}^2\,. && (\mathrm{\Cref{lem:alpha_c}})\label{eq:alpha_minus_p_plus_delta}
    \end{align}
    The use of \Cref{lem:p_beta_gamma_r} in the second line is by choosing $r = \frac{\alpha\mparen{1 - \frac{p}{2}} + \delta\sqrt{p}}{1 - \frac{p}{2}} \ge \alpha \ge \sqrt{2/3}$. (To give some interpretation for this use of \Cref{lem:p_beta_gamma_r}, it effectively says that we can assume without loss of generality that $\gamma = 0$, by placing all of its amplitude on $\beta$ instead.)
    To bound this last quantity, we first observe that
    \begin{align*}
        p\sqrt{\frac{27}{32}}c &\le \sqrt{\frac{3}{32}} && (0 \le p \le 1, 0 \le c \le 1/3)\\
        &< \sqrt{\frac{2}{3}}\\
        &\le \alpha\,,
    \end{align*}
    and therefore the expression inside the parentheses is always positive. So, it suffices to upper bound what is in the parentheses.
    Next, we note that
    \[
    \sqrt{\frac{2}{27}}\frac{\delta^2}{c} - \delta\sqrt{p} + p\sqrt{\frac{27}{32}}c \ge 0\,,
    \]
    because the discriminant of the left side (as a quadratic function of $\delta$) is $0$.
    %\ewin{To get the below, you're using that $\abs{\delta \sqrt{p} - p\sqrt{27/32}c} \leq \alpha$, right? So that it suffices to give an upper bound on the thing inside the square?}
    %\will{good catch. corrected above.}
    Plugging into the bound obtained in \Cref{eq:alpha_minus_p_plus_delta} gives
    \[
        \norm{(\bra{\pi} \otimes I)\ket{\psi}}_2^2 \le \mparen{\alpha + \sqrt{\frac{2}{27}}\frac{\delta^2}{c}}^2\,.
    \]
    Alternatively, we can take
    \[
    \norm{(\bra{\pi} \otimes I)\ket{\psi}}_2^2 \le \mparen{\alpha + \delta}^2
    \]
    by using $0 \le p \le 1$.
    % To conclude, let $d = \sqrt{\frac{2}{27}}\frac{\delta^2}{c}$. If $\alpha^2 + d > 1$, then obviously $\norm{(\bra{\pi} \otimes I)\ket{\psi}}^2 \le 1 < \alpha^2 + 2d$. Otherwise, if $\alpha^2 + d \le 1$, one can verify that $d \le 1 - \alpha^2 \le 2 - 2\alpha$, and therefore
    % \[
    % (\alpha + d)^2 = \alpha^2 + (2\alpha + d)d \le \alpha^2 + 2d.
    % \]
    % So, in both cases we obtain:
    % \[
    %     \norm{(\bra{\pi} \otimes I)\ket{\psi}}^2 \le \alpha^2 + 2d = \alpha^2 + \sqrt{\frac{8}{27}}\frac{\delta^2}{c} .\qedhere
    % \]
\end{proof}

We conclude this subsection by proving a version of \Cref{thm:local_opt_works_better} for mixed states.

\begin{corollary}
    \label{cor:mixed_local_opt_works}
    For an $n$-qubit density matrix $\rho$, define $\Vec{z} \in \C^n$ by
    \[
    z_i = \braket{e_i|\rho|0^n}\,.
    \]
    If $\braket{0^n|\rho|0^n} = 2/3 + c$ for some $c > 0$, then for all product states $\ket{\pi}$:
    \[
    \braket{\pi|\rho|\pi} \le \braket{0^n|\rho|0^n} + \min\mbrace{3\norm{\Vec{z}}_2, \frac{\norm{\Vec{z}}_2^2}{c}}\,.
    \]
\end{corollary}

\begin{proof}
    Take $\ket{\psi}$ to be a purification of $\rho$ in the form of \Cref{thm:local_opt_works_better}:
    \[
    \ket{\psi}= \mparen{\alpha \ket{0^n} + \delta \ket{v_1} + \beta \ket{v_{\ge 2}}}\ket{g} + \gamma\ket{\perp}\,.
    \]
    Observe that
    \[
    \delta^2 = \sum_{i=1}^n \frac{\abs{\braket{e_i|\rho|0^n}}^2}{\alpha^2} = \frac{\norm{\Vec{z}}_2^2}{\alpha^2}\,,
    \]
    because $z_i = \braket{e_i|\rho|0^n} = \alpha \delta \braket{e_i|v_1}$.
    Hence:
    \begin{align*}
    \braket{\pi|\rho|\pi} 
    &= \norm{(\bra{\pi} \otimes I)\ket{\psi}}_2^2\\
    &\le \mparen{\alpha + \frac{\norm{\Vec{z}}_2}{\alpha}}^2 && (\mathrm{\Cref{thm:local_opt_works_better}})\\
    &= \alpha^2 + 2\norm{\Vec{z}}_2 + \frac{\norm{\Vec{z}}_2^2}{\alpha^2}\\
    &= \alpha^2 + \mparen{2 + \frac{\delta}{\alpha}}\norm{\Vec{z}}_2\\
    &\le \alpha^2 + 3\norm{\Vec{z}}_2\,. && (\delta \le \sqrt{1/3}, \alpha \ge \sqrt{2/3})
    \end{align*}
    We can also use the other half of \Cref{thm:local_opt_works_better} to obtain:
    \begin{align*}
    \braket{\pi|\rho|\pi} 
    &= \norm{(\bra{\pi} \otimes I)\ket{\psi}}_2^2\\
    &\le \mparen{\alpha + \sqrt{\frac{2}{27}}\frac{\norm{\Vec{z}}_2^2}{\alpha^2 c}}^2 && (\mathrm{\Cref{thm:local_opt_works_better}})\\
    &\le \mparen{\alpha + \frac{\norm{\Vec{z}}_2^2}{3\alpha c}}^2 && (\alpha \ge \sqrt{2/3})\\
    &= \alpha^2 + \frac{2\norm{\Vec{z}}_2^2}{3c} + \frac{\norm{\Vec{z}}_2^4}{9\alpha^2 c^2}\,.
    \end{align*}
    To complete the proof, assume without loss of generality that $\frac{\norm{z}_2^2}{c} \le 1/3$, as otherwise the statement is trivial. Then:
    \begin{align*}
    \braket{\pi|\rho|\pi} &\le \alpha^2 + \frac{2\norm{\Vec{z}}_2^2}{3c} + \frac{\norm{\Vec{z}}_2^2}{27\alpha^2 c}\\
    &\le \alpha^2 + \frac{2\norm{\Vec{z}}_2^2}{3c} + \frac{\norm{\Vec{z}}_2^2}{18 c} && (\alpha^2 \ge 2/3)\\
    &\le \alpha^2 + \frac{\norm{z}_2^2}{c}\,.\tag*{\qedhere}
    \end{align*}
\end{proof}

\subsection{Bounding local updates}

\Cref{cor:mixed_local_opt_works} shows that the maximum product fidelity can be bounded in terms of the $\ell^2$ norm of a certain vector $\vec{z}$, where $\vec{z}$ captures mass that $\rho$ places coherently between $\ket{0^n}$ and strings of Hamming weight $1$. In this subsection, we show a sort of converse: that there always exists a product state $\ket{\pi}$ whose increase in fidelity compared to $\ket{0^n}$ is proportional to $\norm{\vec{z}}_2^2$. So, we can use $\vec{z}$ to guide our local optimization algorithm until convergence.

We first need a claim showing that $\norm{\vec{z}}_2$ is bounded:

\begin{claim}
    \label{claim:norm_z_at_most_half}
    For an $n$-qubit density matrix $\rho$, define $\Vec{z} \in \C^n$ by
    \[
    z_i = \braket{e_i|\rho|0^n}\,.
    \]
    Then $\norm{\Vec{z}}_2 \le 1/2$.
\end{claim}

\begin{proof}
    Following the proof of \Cref{cor:mixed_local_opt_works}, let $\ket{\psi}$ to be a purification of $\rho$ in the form of \Cref{thm:local_opt_works_better}:
    \[
    \ket{\psi}= \mparen{\alpha \ket{0^n} + \delta \ket{v_1} + \beta \ket{v_{\ge 2}}}\ket{g} + \gamma\ket{\perp}\,.
    \]
    Recall that
    \[
    \delta^2 = \sum_{i=1}^n \frac{\abs{\braket{e_i|\rho|0^n}}^2}{\alpha^2} = \frac{\norm{\Vec{z}}_2^2}{\alpha^2}\,,
    \]
    and therefore $\norm{\Vec{z}}_2 = \alpha\delta$. Since $\alpha^2 + \delta^2 \le 1$, $\norm{\Vec{z}}_2$ is maximized when $\alpha = \delta = \sqrt{2}$, and the claim follows.
\end{proof}

Now we show how to make local updates based on $\vec{z}$. To understand the theorem below, think of $\vec{a}$ as a ``good enough'' approximation to $\vec{z}$. (We require an approximation because quantum tomography algorithms are necessarily non-exact.) Then the theorem shows that, using the product state parametrization from \Cref{def:product_state_parametrization}, moving in the direction of $\vec{a}$ always increases the fidelity by an amount proportional to $\norm{\vec{a}}_2^2$.

\begin{theorem}
\label{thm:local_opt_improvement_with_error}
For an $n$-qubit density matrix $\rho$, define $\Vec{z} \in \C^n$ by
\[
z_i = \braket{e_i|\rho|0^n}\,.
\]
Then if $\norm{\Vec{a} - \Vec{z}}_2 \le \norm{\Vec{a}}_2 / 2$,
\[
\braket{\pi_{\Vec{a}/10}|\rho|\pi_{\Vec{a}/10}} \ge \braket{0^n|\rho|0^n} + \frac{\norm{\Vec{a}}_2^2}{20}\,.
\]
\end{theorem}

\begin{proof}
    Consider a purification $\ket{\psi}$ of $\rho$ in a form similar to \Cref{thm:local_opt_works_better}:
    \[
    \ket{\psi} = \mparen{\alpha \ket{0^n} + \sum_{i=1}^n \frac{\braket{e_i|\rho|0^n}}{\alpha}\ket{e_i} + \beta \ket{v_{\ge 2}}}\ket{g} + \gamma\ket{\perp}\,.
    \]
    We assume without loss of generality that
    \begin{itemize}
        \item $\alpha = \sqrt{\braket{0^n|\rho|0^n}}$ and $\beta$ and $\gamma$ are nonnegative reals,
        \item $\ket{v_{\ge 2}}$, $\ket{g}$, and $\ket{\perp}$ are unit vectors,
        \item $\ket{v_{\ge 2}}$ is supported only on strings of Hamming weight $2$ or larger, and
        \item $\ket{\perp}$ is orthogonal to all states beginning in $\ket{0^n}$ or ending in $\ket{g}$.
    \end{itemize}

    Let us express $\ket{\pi_{\Vec{a}/10}}$ in the form:
    \[
    \ket{\pi_{\Vec{a}/10}} = \sqrt{p_0}\ket{0^n} + \sqrt{p_1}\ket{w_1} + \sqrt{p_{\ge 2}}\ket{w_{\ge 2}}\,,
    \]
    where $p_0$, $p_1$, and $p_{\ge 2}$ are probabilities summing to $1$, $\ket{w_1}$ is supported over strings of weight $1$, $\ket{w_{\ge 2}}$ is supported over strings of weight at least $2$, and both $\ket{w_1}$ and $\ket{w_{\ge 2}}$ are unit vectors.
    %\ewin{same here, maybe need to say that the $p$'s are real (amplitude on $0^n$ is real and non-negative by parametrization)}\will{done.}
    We note that:
    \[
    \braket{\pi_{\Vec{a}/10}|e_i} = \frac{a_i^*/10}{\prod_{j=1}^n \sqrt{1 + \abs{a_i/10}^2}} = \frac{a_i^*\sqrt{p_0}}{10}\,.
    \]
    So, the fidelity is equal to:
    \begin{align*}
        \braket{\pi_{\Vec{a}/10}|\rho|\pi_{\Vec{a}/10}} &=
        \norm{(\bra{\pi_{\Vec{a}/10}} \otimes I)\ket{\psi}}_2^2\\
        &= \norm[\Bigg]{\mparen{\alpha\sqrt{p_0} + \sum_{i=1}^n \frac{a_i^*\sqrt{p_0}}{10}\frac{\braket{e_i|\rho|0^n}}{\alpha} + \beta\sqrt{p_{\ge 2}}\braket{w_{\ge 2}|v_{\ge 2}}}\ket{g} + \gamma(\bra{\pi} \otimes I)\ket{\perp}}_2^2\\
        &= \norm[\Bigg]{\mparen{\alpha\sqrt{p_0} + \frac{\sqrt{p_0}\braket{\Vec{a}, \Vec{z}}}{10\alpha} + \beta\sqrt{p_{\ge 2}}\braket{w_{\ge 2}|v_{\ge 2}}}\ket{g} + \gamma(\bra{\pi} \otimes I)\ket{\perp}}_2^2\\
        &= \norm[\Bigg]{\mparen{\alpha\sqrt{p_0} + \frac{\sqrt{p_0}\braket{\Vec{a}, \Vec{z}}}{10\alpha} + \beta\sqrt{p_{\ge 2}}\braket{w_{\ge 2}|v_{\ge 2}}}\ket{g}}_2^2 + \norm{\gamma(\bra{\pi} \otimes I)\ket{\perp}}_2^2
    \end{align*}
    by the Pythagorean theorem, because $\ket{\perp}$ has no support on $\ket{g}$.
    We can lower bound this by:
    \begin{align*}
        \braket{\pi_{\Vec{a}/10}|\rho|\pi_{\Vec{a}/10}}
        &\ge \abs[\Bigg]{\alpha\sqrt{p_0} + \frac{\sqrt{p_0}\braket{\Vec{a}, \Vec{z}}}{10\alpha} + \beta\sqrt{p_{\ge 2}}\braket{w_{\ge 2}|v_{\ge 2}}}^2\\
        &= \abs[\Bigg]{\alpha\sqrt{p_0} + \frac{\sqrt{p_0}\mparen{
        \norm{\vec{a}}_2^2 - \braket{\vec{a},\vec{a} - \vec{z}}
        }}{10\alpha} + \beta\sqrt{p_{\ge 2}}\braket{w_{\ge 2}|v_{\ge 2}}}^2\\
        &\ge \mparen{\alpha\sqrt{p_0} + \frac{\sqrt{p_0}\mparen{
        \norm{\vec{a}}_2^2 - \abs{\braket{\vec{a},\vec{a} - \vec{z}}}
        }}{10\alpha} - \beta\sqrt{p_{\ge 2}}}^2 && (\text{Triangle inequality})\\
        &\ge \mparen{\alpha\sqrt{p_0} + \frac{\sqrt{p_0}\mparen{
        \norm{\vec{a}}_2^2 - \norm{\vec{a}}_2\norm{\vec{a} - \vec{z}}_2
        }}{10\alpha} - \beta\sqrt{p_{\ge 2}}}^2  && (\text{Cauchy-Schwarz})\\
        &\ge \mparen{\alpha\sqrt{p_0} + \frac{\sqrt{p_0}
        \norm{\vec{a}}_2^2}{20\alpha} - \beta\sqrt{p_{\ge 2}}}^2  && (\norm{\Vec{a} - \Vec{z}}_2 \le  \norm{\Vec{a}}_2/2)\\
        &\ge \mparen{\alpha \sqrt{p_0} + \frac{\sqrt{p_0}
        \norm{\vec{a}}_2^2}{20\alpha} - \beta\sqrt{2}(1 - \sqrt{p_0})}^2 && (\mathrm{\Cref{cor:p0_vs_p2}})\\
        &\ge \mparen{\alpha p_0 + \frac{p_0
        \norm{\vec{a}}_2^2}{20\alpha} - \beta\sqrt{2}(1 - p_0)}^2\,. && (p_0 \le 1)
    \end{align*}

    Let us now obtain some bounds on $p_0$.
    From \Cref{lem:dtan-fidelity-approx}, we know that
    \[
    p_0 = \prod_{i=1}^n \frac{1}{1 + \abs{a_i/10}^2} \ge e^{-\dtan(\ket{\pi_{\Vec{a}/10}}, \ket{0^n})^2} = e^{-\norm{\Vec{a}/10}_2^2} \ge 1 - \frac{\norm{\Vec{a}}_2^2}{100}\,.
    \]
    By \Cref{claim:norm_z_at_most_half}, because $\norm{\Vec{z}}_2 \le \frac{1}{2}$, we know that
    \[
    \norm{\Vec{a}}_2\le \norm{\Vec{z}}_2 + \norm{\Vec{a} - \Vec{z}}_2 \le \frac{1 + \norm{\Vec{a}}_2}{2}\,,
    \]
    and therefore $\norm{\Vec{a}}_2 \le 1$.
    This further implies $p_0 \ge 0.99$. Using these two bounds on $p_0$, we obtain the desired lower bound:
    \begin{align*}
        \braket{\pi_{\Vec{a}/10}|\rho|\pi_{\Vec{a}/10}}
        &\ge \mparen{\alpha\mparen{1 - \frac{\norm{\Vec{a}}_2^2}{100}} + \frac{0.99\norm{\Vec{a}}_2^2}{20\alpha} - \beta\sqrt{2}\frac{\norm{\Vec{a}}_2^2}{100}}^2\\
        &\ge \mparen{\alpha + \frac{0.0495}{\alpha}\norm{\Vec{a}}_2^2 - \frac{1 + \sqrt{2}}{100}\norm{\Vec{a}}_2^2 }^2 && (\alpha, \beta \le 1)\\
        &\ge \mparen{\alpha + \frac{0.0395 - \sqrt{2}/100}{\alpha}\norm{\Vec{a}}_2^2}^2 && (\alpha \le 1)\\
        &\ge \mparen{\alpha + \frac{1}{40\alpha}\norm{\Vec{a}}_2^2}^2\\
        &\ge \alpha^2 + \frac{\norm{\Vec{a}}_2^2}{20}\\
        &= \braket{0^n|\rho|0^n} + \frac{\norm{\Vec{a}}_2^2}{20}\,.\tag*{\qedhere}
    \end{align*}
\end{proof}

\subsection{Local optimization algorithms}

Let's briefly take a step back to show the power of the combination of \Cref{cor:mixed_local_opt_works}  and \Cref{thm:local_opt_improvement_with_error}.
Say that (in some appropriately chosen basis), we've found that
\[
\braket{0^n|\rho|0^n} = 2/3 + c\,.
\]
Suppose that the largest fidelity achievable with $\rho$ by any product state is $2/3 + D$.
Then \Cref{cor:mixed_local_opt_works} shows that
\[
2/3 + D \le 2/3 + c + \min\mbrace{3\norm{\vec{z}}_2, \frac{\norm{\vec{z}}_2^2}{c}}\,,
\]
or equivalently
\[
    \norm{\Vec{z}}_2^2 \ge \max\mbrace{\frac{(D-c)^2}{9}, c(D - c)}\,,
\]
which further implies
\[
    \norm{\Vec{z}}_2^2 \ge \frac{D(D - c)}{10}\,.
\]
To simplify things, imagine that we were able to learn $\vec{z}$ \textit{exactly}.
Then \Cref{thm:local_opt_improvement_with_error} shows that we can find a product state $\ket{\pi}$ whose fidelity with $\rho$ is at least
\[
\braket{\pi|\rho|\pi} \ge 2/3 + c + \frac{\norm{z}_2^2}{20}+\,.
\]
Combining these two shows that after a single local update according to \Cref{thm:local_opt_improvement_with_error}, $c$ increases at least as fast as
\[
c \to c + \frac{D(D - c)}{200}\,.
\]
If we make such local updates repeatedly, the fidelity of our product state with $\rho$ converges towards $2/3 + D$ exponentially quickly. This is the high-level idea behind our local optimization procedure, \Cref{alg:local_opt_no_logs} below.

We first take a small detour to show how to learn an approximation of $\Vec{z}$ in a sample- and time-efficient manner, via a simple modification of the classical shadows protocol~\cite{hkp20}.

\begin{lemma}
    \label{lem:estimating_z}
    For an $n$-qubit density matrix $\rho$, define $\Vec{z} \in \C^n$ by
    \[
    z_i = \braket{e_i|\rho|0^n}\,.
    \]
    Then there is a procedure to find $\Vec{a}$ satisfying $\norm{\Vec{a} - \Vec{z}}_2 \le \eps$ with probability $1 - \delta$ that uses $O(\frac{n}{\eps^2} \log \frac{1}{\delta})$ copies of $\rho$ and $O\mparen{\frac{n^2 \log n}{\eps^2}\log\frac{1}{\delta} + n\log^2 \frac{1}{\delta}}$ time.
\end{lemma}

\begin{proof}
    Note that the real part of $\vec{z}$ is
    \[
    \Re(z_i) = \Tr\mparen{ 
    \frac{\ketbra{0^n}{e_i} + \ketbra{e_i}{0^n}}{2}
    \rho
    },
    \]
    and the imaginary part is
    \[
    \Im(z_i) = \Tr\mparen{
    \rho
    \frac{i\ketbra{0^n}{e_i} - i\ketbra{e_i}{0^n}}{2}
    \rho
    }\,.
    \]
    So, it suffices to obtain an $\ell^2$-error estimate of the $2n$ observables
    \[
    \mbrace{\frac{\ketbra{0^n}{e_i} + \ketbra{e_i}{0^n}}{2}
    , \frac{i\ketbra{0^n}{e_i} - i\ketbra{e_i}{0^n}}{2}}_{i \in [n]}\,,
    \]
    which we will call $O_1,\ldots,O_{2n}$.
    
    For some $N$ and $K$ that we choose later, we will use $NK$ copies of $\rho$ to produce estimates $\hat{o}_i(N, K)$ of each $o_i = \tr(O_i \rho)$. Our success criterion will then be
    \[
    \sum_{i=1}^{2n} (o_i - \hat{o}_i(N, K))^2 \le \eps^2\,.
    \]
    To do so, we use the classical shadows framework of~\cite{hkp20}. Consider measuring $\rho$ in a random Clifford basis $U$, obtaining outcome $\ket{\hat{b}}$. Following~\cite[Eqs. (S16) and (S5)]{hkp20}, we define
    \[
    \hat{\rho} \coloneqq (2^n + 1)U^\dagger \ketbra{\hat{b}}{\hat{b}}U - I
    \]
    to be the classical shadow, and
    \[
    \hat{o}_i = \tr(O_i \hat{\rho})
    \]
    the estimator corresponding to $O_i$.
    The key fact shown in~\cite[Lemma 1 and Eq. (S16)]{hkp20} is
    \[
    \E[\hat{o}_i] = o_i \qquad\qquad\text{and}\qquad\qquad \Var[\hat{o}_i] \le 3\tr(O_i)^2 = \frac{3}{2}\,.
    \]
    So, the expected sum of squared deviations is at most:
    \[
    \E\mbracket{\sum_{i=1}^{2n} (o_i - \hat{o}_i)^2} = \sum_{i=1}^{2n} \E\mbracket{ (o_i - \hat{o}_i)^2} \le 3n\,.
    \]
    If we take $N$ classical shadows $\hat{\rho}_1,\ldots,\hat{\rho}_N$ and compute the mean of their estimators, as in~\cite[Eq. (S11)]{hkp20}:
    \[
    \hat{o}_i(N, 1) \coloneqq \frac{1}{N} \sum_{j=1}^N \tr(O_i \hat{\rho}_j)\,,
    \]
    the expected sum of squared deviations is decreased by a factor of $N$:
    \[
    \E\mbracket{\sum_{i=1}^{2n} (o_i - \hat{o}_i(N, 1))^2} \le \frac{3n}{N}\,.
    \]
    Choosing $N = \frac{n}{27\eps^2}$, by Markov's inequality we will have
    \[
    \Pr\mbracket{\sum_{i=1}^{2n} (o_i - \hat{o}_i(N, 1))^2 \le \frac{\eps^2}{3}} \ge \frac{2}{3}\,.
    \]
    To boost the success probability from $2/3$ to $1 - \delta$, we can use a median-of-means trick by letting $\vec{\hat{o}}(N,K)$ be the ``median'' of $K = O(\log \frac{1}{\delta})$ independent samples from $\vec{\hat{o}}(N,1)$; see for example~\cite[Proposition 2.4]{hkot23}. (Note: it is important that the ``median'' in this case is performed with respect to the entire vectors $\vec{\hat{o}}(N,1)$ rather than entrywise on $\hat{o}_i(N,1)$; see~\cite{hkot23} for details. The important feature is that the measure of error, $\ell^2$ distance, is a metric.)

    It remains to bound the runtime. Naively, operating on the classical shadows takes exponential time. However, because all of the observables are supported only the subspace of Hamming weight $0$ or $1$, we can encode the information about this subspace into a register of just $O(\log n)$ qubits and perform the classical shadows there. For each sample of $\rho$, we append a register of $1 + \lceil \log_2(n + 1) \rceil$ qubits initialized to $\ket{0}$. We coherently set the first appended qubit to $\ket{1}$ conditioned on having Hamming weight $0$ or $1$ in the $n$-qubit register. Then, conditioned on having Hamming weight $0$ or $1$, we populate the remainder of the register with the binary representation of the qubit that was set to $\ket{1}$, or zero if no such qubit exists.
    Now we can take the classical shadows entirely within this $O(\log n)$-qubit register.

    Mapping $\rho$ into this $O(\log n)$-qubit register takes $O(n\log n)$ time per classical shadow.
    Next, sampling the random Clifford $U$ and then measuring to get $\ket{\hat{b}}$ takes $O(\log^2 n)$ time per classical shadow~\cite{vdB20-clifford}. We then compute the entire amplitude vector of $U^\dagger \ket{\hat{b}}$ in $\C^{O(n)}$, which takes time $O(n \log n)$ per classical shadow~\cite[Algorithm 2]{dSSY23-fast}. Since the $O_i$'s are sparse, computing each $\hat{o}_i = \tr(O_i\hat{\rho})$ takes $O(1)$ time by using $U^\dagger \ket{\hat{b}}$ as a lookup table. Finally, as noted in~\cite[Proposition 2.4]{hkot23}, the ``median-of-means procedure'' takes $O(K^2)$ times the cost of computing the distance between two vectors $\vec{\hat{o}}(N,1)$, for a total of $O(K^2n)$ time.\footnote{$O(K^2n)$ is a worst-case bound, but the average-case time complexity is easily seen to be $O(Kn)$. We believe that a more careful accounting of this factor over the iterations of \Cref{alg:local_opt_no_logs} could remove the $\log^2\frac{1}{\delta}$ factors that appear in the runtime bounds later in this section.} The overall runtime is therefore:
    \[
    O(NK(n \log n) + K^2n) = O\mparen{\frac{n^2 \log n}{\eps^2}\log\frac{1}{\delta} + n\log^2 \frac{1}{\delta}}\,.\qedhere
    \]
\end{proof}

We remark that, instead of the lemma above, one could simply appeal to~\cite[Theorem 1]{hkp20} as a black box to estimate each $z_i$ to additive error $\sqrt{\frac{\eps}{n}}$. However, our approach saves a $\log(n)$ factor in the sample complexity because we only need to approximate $\vec{z}$ in $\ell^2$ distance, rather than $\ell^\infty$.

We can now state our algorithm for optimizing toward a global maximizer of product state fidelity.

\begin{longfbox}[breakable=false, padding=1em, margin-top=1em, margin-bottom=1em]
\begin{algorithm}
\label{alg:local_opt_no_logs}
    (Local product state optimization).
\begin{description}
    \item[Input:] Copies of an $n$-qubit state $\rho$, an $n$-qubit product state $\ket{\pi}$, $\eps \in (0, 1/3]$, $\delta \in (0, 1)$, $C \in [0, 1/3]$
    \item[Promise:] $\braket{\pi|\rho|\pi} \ge 2/3$, and $\max_{\text{product }\ket{\pi'}} \braket{\pi'|\rho|\pi'} \ge 2/3 + C$
    \item[Output:] A product state $\ket{\pi}$ satisfying $\braket{\pi|\rho|\pi} \ge \max_{\text{product }\ket{\pi'}} \braket{\pi'|\rho|\pi'} - \eps$ with probability at least $1 - \delta$
    \item[Procedure:] \mbox{}
    \begin{algorithmic}[1]
        \State $C' \coloneqq \max\{\eps, C\}$ \;
        \State $m \coloneqq \left\lceil \frac{1}{2}\log\mparen{\frac{90}{C'\eps}} \right\rceil$ \;
        \Loop
            \State Find a product unitary $U$ such that $U\ket{\pi} = \ket{0^n}$ \;
            \State Define $\Vec{z}$ by $z_i = \braket{e_i|U\rho U^\dagger|0^n}$ for $i \in [n]$\;

            \For{$\lambda = 1,\ldots,m$}\label{line:for_loop_lambda}
                \State $\ell_\lambda \coloneqq m + 1 - \lambda$ \;
                \State $\delta_\lambda \coloneqq \delta 2^{-\ell_\lambda}\mparen{\left\lceil 5\eps e^{2m} \right\rceil + \left\lceil \frac{900\ell_\lambda}{C'} \right\rceil}^{-1}$
                \State Use \Cref{lem:estimating_z} to find $\Vec{a}$ satisfying $\norm{\Vec{a} - \Vec{z}}_2 \le e^{-\lambda}$ with prob.\ $1 - \delta_\lambda$ \label{line:a_minus_z_bound_no_logs}\;
                \If{$\norm{\Vec{a}}_2 \ge 2e^{- \lambda}$}
                    \State $\ket{\pi} \coloneqq U^\dagger \ket{\pi_{\Vec{a}/10}}$ \label{line:update_qubit_no_logs}\; 
                    \State Exit for-loop\;
                \ElsIf{$\lambda = m$} 
                    \State \Return $\ket{\pi}$ \label{line:a_not_big_enough}\;
                \EndIf
            \EndFor
        \EndLoop
    \end{algorithmic}
\end{description}
\end{algorithm}
\end{longfbox}

We note that the parameter $C$ is effectively optional: one can always set $C = 0$ and the algorithm will be correct. However, the algorithm becomes more efficient when $C$ is larger.

We will show the correctness of \Cref{alg:local_opt_no_logs} in a series of smaller steps. First we lower bound the improvement from updates:

\begin{claim}
    \label{claim:fidelity_always_increases_no_logs}
    In a given non-terminating iteration of the outer loop of \Cref{alg:local_opt_no_logs}, if \Cref{line:a_minus_z_bound_no_logs} does not err, then $\braket{\pi|\rho|\pi}$ increases by at least $\frac{\norm{\vec{a}}_2^2}{20}$.
\end{claim}

\begin{proof}
    $\ket{\pi}$ only changes in \Cref{line:update_qubit_no_logs}, so consider an iteration in which the algorithm reaches \Cref{line:update_qubit_no_logs}, and therefore $\norm{\Vec{a}}_2 \ge 2e^{- \lambda} \ge 2\norm{\vec{a} - \vec{z}}_2$ for some $\lambda$. Then we may appeal to \Cref{thm:local_opt_improvement_with_error} to conclude that
    \begin{align*}
        \braket{\pi_{\Vec{a}/10}|U\rho U^\dagger|\pi_{\Vec{a}/10}} - \braket{\pi|\rho|\pi} &= \braket{\pi_{\Vec{a}/10}|U\rho U^\dagger|\pi_{\Vec{a}/10}} - \braket{0^n|U\rho U^\dagger|0^n}\\
        &\ge \frac{\norm{\Vec{a}}_2^2}{20}\,.\qedhere
    \end{align*}
\end{proof}

The helper lemma below will be used to show that the product state fidelity converges to within $\eps$ of the optimum in roughly $O\mparen{\log \frac{1}{\eps}}$ iterations of the outer loop.

\begin{lemma}
    \label{lem:exponential_decay}
    Let $\{x_i\}_{i \in \N}$ be a sequence satisfying $x_0 \ge c \ge 0$ and
    \[
    x_{i+1} \ge \min\left\{x_i + \frac{D(D - x_i)}{r}, D - \eps\right\}
    \]
    for some $r > D$.
    If we define
    \[
    k \coloneqq \frac{r}{D}\log\mparen{\frac{D - c}{\eps}}\,,
    \]
    then for all $i \ge k$, $x_i \ge D - \eps$.
\end{lemma}

\begin{proof}
    We assume without loss of generality that $c \le D - \eps$, because otherwise $k$ is negative and the statement clearly holds.

    We first note that $x_i \ge t$ implies
    \[
    x_{i+1} \ge \min\left\{t + \frac{D(D - t)}{r}, D - \eps\right\}\,,
    \]
    because $f(t) = t + \frac{D(D - t)}{r}$ is increasing on $[0, D]$.
    So, if we define the sequence $\{y_i\}_{i \in \N}$ by $y_0 = c$ and
    \[
    y_{i+1} = y_i + \frac{D(D - y_i)}{r}\,,
    \]
    then $x_i \ge \min \{y_i, D - \eps\}$, and thus it suffices to show that $y_{i} \ge D - \eps$ for all $i \ge k$.

    We can write the definition of $y$ equivalently as
    \[
    D - y_{i+1} = D - y_i - \frac{D(D - y_i)}{r} = \mparen{1 - \frac{D}{r}}(D - y_i)\,.
    \]
    In other words, $D - y_i$ decays exponentially as
    \[
    D - y_{i} = \mparen{1 - \frac{D}{r}}^{i}(D - c)\,.
    \]
    Thus choosing
    \[
    k' = \frac{\log\mparen{\frac{\eps}{D - c}}}{\log\mparen{1 - \frac{D}{r}}}
    \]
    guarantees that $y_i \ge D - \eps$ for all $i \ge k'$. Because
    \begin{align*}
        k' &= -\frac{\log\mparen{\frac{D-c}{\eps}}}{\log\mparen{1 - \frac{D}{r}}}\\
        &\le \frac{r}{D}\log\mparen{\frac{D - c}{\eps}} && (\log(1+x) \le x, \log\mparen{\frac{D-c}{\eps} \ge 0})\\
        &\le k\,,
    \end{align*}
    the lemma follows.
\end{proof}

\begin{lemma}
    \label{lem:local_opt_iters_to_converge}
    Suppose that
    \[
    \max_{\mathrm{product}\ \ket{\pi'}}\braket{\pi'|\rho|\pi'} = 2/3 + D\,,
    \]
    and suppose $\braket{\pi|\rho|\pi} \ge 2/3 + c \ge 2/3$ in some iteration of \Cref{line:a_minus_z_bound_no_logs}'s outer loop. 
    Then conditioned on \Cref{line:a_minus_z_bound_no_logs} never erring thereafter,
    \Cref{alg:local_opt_no_logs} returns a $\ket{\pi}$ satisfying the output condition (i.e., $\braket{\pi|\rho|\pi} \ge 2/3 + D - \eps$) within at most
    \[
    \left\lceil5\eps e^{2m}\right\rceil + \max\mbrace{0, \left\lceil\frac{450}{D}\log\mparen{\frac{D - c}{\eps}}\right\rceil}
    \]
    additional iterations of the outer loop.
\end{lemma}

\begin{proof}
    \Cref{claim:fidelity_always_increases_no_logs} shows that in each non-terminating iteration, $\braket{\pi|\rho|\pi}$ increases by at least $\frac{\norm{\vec{a}}_2^2}{20} \ge \frac{e^{-2m}}{5}$. So, if $c \ge D - \eps$, then the algorithm must halt within $\lceil 5\eps e^{2m} \rceil$ additional iterations (as the fidelity can never exceed $2/3 + D$). Conversely, if $c < D - \eps$, then it suffices to show that the algorithm achieves $\braket{\pi|\rho|\pi} \ge 2/3 + D - \eps$ within the initial $\left\lceil\frac{450}{D}\log\mparen{\frac{D - c}{\eps}}\right\rceil$ iterations of the outer loop, because thereafter the algorithm must halt within $\lceil 5\eps e^{2m} \rceil$ additional iterations. So, we assume henceforth that $c < D - \eps$.

    For $i \in \N$, define $x_i$ so that $\braket{\pi|\rho|\pi} = 2/3 + x_i$ immediately after the end of $i$ additional iterations of the outer loop, with the convention that $\braket{\pi|\rho|\pi} = 2/3 + x_0 \ge 2/3 + c$ at the start of the algorithm.
    
    Suppose that $x_i \le D - \eps$, and consider what happens in iteration $i+1$. We know that $x_i \ge c \ge 0$ by \Cref{claim:fidelity_always_increases_no_logs}.
    Then \Cref{cor:mixed_local_opt_works} tells us that
    \[
    2/3 + D \le 2/3 + x_i + \min\mbrace{3\norm{\vec{z}}_2,\frac{\norm{\Vec{z}}_2^2}{x_i}}\,,
    \]
    or equivalently
    \[
    \norm{\Vec{z}}_2^2 \ge \max\mbrace{\frac{(D-x_i)^2}{9}, x(D - x_i)}\,,
    \]
    which then gives
    \begin{equation}
    \label{eq:norm_z_in_terms_of_D}
    \norm{\Vec{z}}_2^2 \ge \frac{D(D-x_i)}{10}.
    \end{equation}

    We claim that within this $(i+1)$th iteration of the outer loop, the algorithm finds $\Vec{a}$ satisfying $\norm{\Vec{a}}_2 \ge e^{1 - \lambda} \ge 2\norm{\Vec{a} - \Vec{z}}_2$ for \textit{some} $\lambda \le m$, as otherwise upon reaching \Cref{line:a_not_big_enough} we would have
    \begin{align*}
        \sqrt{\frac{C'\eps}{10}}
        &\le 
        \sqrt{\frac{D\eps}{10}} && (D \ge C \text{ and } D \ge x_i + \eps \ge \eps \text{ by supposition})\\
        &\le 
        \sqrt{\frac{D(D - x_i)}{10}} && (D \ge x_i + \eps\text{ by supposition}))\\
        &\le \norm{\vec{z}}_2 && (\mathrm{\Cref{eq:norm_z_in_terms_of_D}})\\
        &\le \norm{\vec{a}}_2 + \norm{\vec{a} - \vec{z}}_2 && (\text{Triangle inequality})\\
        &< 2e^{- m} + e^{-m}\\
        &\le \sqrt{\frac{C'\eps}{10}}\,,
    \end{align*}
    a contradiction.

    We know that $\norm{\Vec{a}}_2 \ge \norm{\Vec{z}}_2 - \norm{\Vec{a} - \vec{z}}_2$ by the triangle inequality, and since $\norm{\Vec{a} - \vec{z}}_2 \le \norm{\Vec{a}}_2/2$, we get that
    \begin{align*}
        \norm{\Vec{a}}_2 &\ge \frac{2}{3}\norm{\Vec{z}}_2\\
            &\ge \frac{2}{3}\sqrt{\frac{D(D-x_i)}{10}}\,. && (\mathrm{\Cref{eq:norm_z_in_terms_of_D}})
    \end{align*}
    Plugging into \Cref{cor:mixed_local_opt_works},
    \begin{align*}
        \braket{\pi_{\Vec{a}/10}|U\rho U^\dagger|\pi_{\Vec{a}/10}} - \braket{\pi|\rho|\pi} &= \braket{\pi_{\Vec{a}/10}|U\rho U^\dagger|\pi_{\Vec{a}/10}} - \braket{0^n|U\rho U^\dagger|0^n}\\
        &\ge \frac{\norm{\Vec{a}}_2^2}{20}\\
        &\ge \frac{D(D - x_i)}{450}\,.
    \end{align*}

    We have effectively established that
    \[
    x_{i+1} \ge \min\left\{x_i + \frac{D(D - x_i)}{450}, D - \eps\right\}\,.
    \]
    \Cref{lem:exponential_decay} shows that for all $i \ge k$, $x_{i} \ge D - \eps$, where
    \[
    k = \frac{450}{D}\log\mparen{\frac{D - c}{\eps}}\,,
    \]
    which proves the lemma.
\end{proof}

\Cref{lem:local_opt_iters_to_converge} is almost sufficient to establish \Cref{alg:local_opt_no_logs}'s correctness; we only need to show that the total error probability is at most $\delta$. For that, we bound the total number of calls to the tomography subroutine (\Cref{line:a_minus_z_bound_no_logs}).

\begin{lemma}
    \label{lem:iters_per_lambda}
    Consider a run of 
    \Cref{alg:local_opt_no_logs} where \Cref{line:a_minus_z_bound_no_logs} never errs. Then for a given $\lambda$, \Cref{line:a_minus_z_bound_no_logs} is executed at most
    \[
    \left\lceil5\eps e^{2m}\right\rceil + \left\lceil\frac{900\ell_\lambda}{C'}\right\rceil
    \]
    times before \Cref{alg:local_opt_no_logs} halts.
\end{lemma}

\begin{proof}
    Define $D \coloneqq \max_{\mathrm{product}\ \ket{\pi'}}\braket{\pi'|\rho|\pi'} - 2/3$ as in \Cref{lem:local_opt_iters_to_converge}. We break into cases depending on $D$ and $\lambda$.

    Case 1: $D < \eps$. Then the total number of calls to \Cref{line:a_minus_z_bound_no_logs} is bounded by the number of iterations of the outer loop, which by \Cref{lem:local_opt_iters_to_converge} is at most
    \[
        \left\lceil5\eps e^{2m}\right\rceil + \max\mbrace{0, \left\lceil\frac{450}{D}\log\mparen{\frac{D}{\eps}}\right\rceil}
        = \left\lceil5\eps e^{2m}\right\rceil\,.
    \]

    Case 2: $D \ge \eps$, $\lambda = 1$. Again we bound the number of calls to \Cref{line:a_minus_z_bound_no_logs} by the number of iterations of the outer loop, using \Cref{lem:local_opt_iters_to_converge} and that $1/3 \ge D \ge \max\{C, \eps\} = C'$:
    \begin{align*}
        \left\lceil5\eps e^{2m}\right\rceil + \max\mbrace{0, \left\lceil\frac{450}{D}\log\mparen{\frac{D}{\eps}}\right\rceil}
        &\le \left\lceil5\eps e^{2m}\right\rceil + \left\lceil\frac{450}{C'}\log\mparen{\frac{1}{3\eps}}\right\rceil \\
        &\le \left\lceil5\eps e^{2m}\right\rceil + \left\lceil\frac{900m}{C'}\right\rceil\\
        &= \left\lceil5\eps e^{2m}\right\rceil + \left\lceil\frac{900\ell_\lambda}{C'}\right\rceil\,.
    \end{align*}

    Case 3: $D \ge \eps$, $\lambda \ge 2$. Consider an iteration of the algorithm in which $\braket{\pi|\rho|\pi} = 2/3 + x$. 
    In order for the for-loop (\Cref{line:for_loop_lambda}) to reach iteration $\lambda$, by the triangle inequality we must have
    \begin{align*}
    \norm{\vec{z}}_2 &\le \norm{\vec{a}}_2 + \norm{\vec{a}-\vec{z}}_2\\
    &\le 2e^{-(\lambda - 1)} + e^{-(\lambda - 1)}\\
    &= 3e^{\ell_\lambda - m}\\
    &\le \sqrt{\frac{C'\eps}{10}}e^{\ell_\lambda}\,.
    \end{align*}
    We know from \Cref{cor:mixed_local_opt_works} that
    \[
    2/3 + D \le 2/3 + x + \min\mbrace{3\norm{\vec{z}}_2,\frac{\norm{\Vec{z}}_2^2}{x}}\,,
    \]
    or equivalently
    \[
    \norm{\Vec{z}}_2^2 \ge \max\mbrace{\frac{(D-x)^2}{9}, x(D - x)}\,,
    \]
    which then gives
    \[
    \norm{\Vec{z}}_2^2 \ge \frac{D(D-x)}{10}\,.
    \]
    Combining, we find that
    \[
    \frac{D-x}{\eps} \le \frac{C'}{D}e^{2\ell_\lambda} \le e^{2\ell_\lambda}\,,
    \]
    because $D \ge \max\{\eps, C\} = C'$.

    We have shown that \Cref{line:a_minus_z_bound_no_logs} is executed for this particular $\lambda$ only when $x \ge D - \eps e^{2\ell_\lambda}$. 
    Since $x$ never drops below $0$ (\Cref{claim:fidelity_always_increases_no_logs}), we can let $c \coloneqq \max\{0,D - \eps e^{2\ell_\lambda}\}$ and appeal to
    \Cref{lem:local_opt_iters_to_converge} to bound the number of additional iterations for which the algorithm can run by
    \[
    \left\lceil5\eps e^{2m}\right\rceil + \max\mbrace{0,\left\lceil\frac{450}{D}\log\mparen{\frac{D - c}{\eps}}\right\rceil}\,.
    \]
    The lemma follows by substituting $c \ge D - \eps e^{2\ell_\lambda}$ and $C' \le D$.
\end{proof}

\begin{corollary}
    The total failure probability of \Cref{alg:local_opt_no_logs} is at most $\delta$.
\end{corollary}

\begin{proof}
    By a union bound, we can use \Cref{lem:iters_per_lambda} to bound the contribution of each run of \Cref{line:a_minus_z_bound_no_logs} for a given $\lambda$:
    \begin{align*}
        \sum_{\lambda = 1}^m \delta_\lambda \mparen{\left\lceil5\eps e^{2m}\right\rceil + \left\lceil\frac{900\ell_\lambda}{C'}\right\rceil}
        &= \delta\sum_{\lambda = 1}^m 2^{-\ell_\lambda}\\
        &= \delta\sum_{\ell = 1}^{m} 2^{- \ell} && (\ell \coloneqq m + 1 - \lambda)\\
        &< \delta\sum_{\ell = 1}^{\infty} 2^{-\ell}\\
        &= \delta\,.\tag*{\qedhere}
    \end{align*}
\end{proof}

\begin{corollary}
    \label{cor:local_opt_complexity}
    Assuming \Cref{line:a_minus_z_bound_no_logs} never errs, the total sample complexity \Cref{alg:local_opt_no_logs} is at most
    \[
    O\mparen{\frac{n}{\eps C'^2}\log\frac{1}{\delta C'}}
    \]
    and the runtime is
    \[
    O\mparen{\frac{n^2 \log n}{\eps C'^2}\log\frac{1}{\delta C'} + \frac{n}{C'}\log^2 \frac{1}{\eps C'}\log^2\frac{1}{\delta C'}}\,,
    \]
    recalling that $C' = \max\{\eps,C\}$.
\end{corollary}

\begin{proof}
    For a given $\lambda$, a single run of \Cref{line:a_minus_z_bound_no_logs} has sample complexity $O\mparen{ne^{2\lambda}\log \frac{1}{\delta_\lambda}}$ according to \Cref{lem:estimating_z}. A simple calculation shows that
    \[
    \log \frac{1}{\delta_\lambda} \le O\mparen{\ell_\lambda + \log \frac{\ell_\lambda}{\delta C'}} \le O\mparen{\ell_\lambda + \log \frac{1}{\delta C'}}\,.
    \]
    Using \Cref{lem:iters_per_lambda}, \Cref{line:a_minus_z_bound_no_logs} is called a total of
    \[
    \left\lceil5\eps e^{2m}\right\rceil + \left\lceil\frac{900\ell_\lambda}{C'}\right\rceil \le O\mparen{\frac{\ell_\lambda}{C'}}
    \]
    times.
    The total sample complexity is therefore
    \begin{align}
        \sum_{\lambda = 1}^m O\mparen{ne^{2\lambda} \log\mparen{\frac{1}{\delta_\lambda}} \frac{\ell_\lambda}{C'}}
        &\le \sum_{\lambda = 1}^m O\mparen{ne^{2\lambda} \mparen{\ell_\lambda + \log \frac{1}{\delta C'}} \frac{\ell_\lambda}{C'}}\nonumber\\
        &\le \sum_{\ell = 1}^{m} O\mparen{ne^{2m + 2 - 2\ell} \mparen{\ell + \log \frac{1}{\delta C'}} \frac{\ell}{C'}} && (\ell \coloneqq m + 1 - \lambda)\nonumber\\
        &\le \frac{n}{\eps C'^2} \sum_{\ell = 1}^{m} O\mparen{e^{- 2\ell} \mparen{\ell + \log \frac{1}{\delta C'}}\ell}\nonumber\\
        &\le O\mparen{\frac{n}{\eps C'^2}\log\frac{1}{\delta C'}}\,.\label{eq:local_opt_total_sample_complexity}
    \end{align}

    We turn to the time complexity.
    By \Cref{lem:estimating_z}, the runtime of a single call to \Cref{line:a_minus_z_bound_no_logs} with chosen $\lambda$ is
    \[
    (\text{Sample complexity}) \cdot O(n \log n) + O\mparen{n\log^2 \frac{1}{\delta_\lambda}}.
    \]
    Since the $O(n \log n)$ is independent of $\lambda$, it is easy to bound the contribution of the left term to the runtime: we multiply the total sample complexity (\Cref{eq:local_opt_total_sample_complexity}) by $O(n \log n)$, yielding $O\mparen{\frac{n^2 \log n}{\eps C'^2}\log\frac{1}{\delta C'}}$.
    So, we focus on bounding the contribution of the right term, which is
    \begin{align}
        \sum_{\lambda = 1}^m O\mparen{n \log^2\mparen{\frac{1}{\delta_\lambda}} \frac{\ell_\lambda}{C'}}
        &\le \sum_{\lambda = 1}^m O\mparen{n \mparen{\ell_\lambda + \log \frac{1}{\delta C'}}^2 \frac{\ell_\lambda}{C'}}\nonumber\\
        &\le \sum_{\ell = 1}^{m} O\mparen{n \mparen{\ell + \log \frac{1}{\delta C'}}^2 \frac{\ell}{C'}} && (\ell \coloneqq m + 1 - \lambda)\nonumber\\
        &\le \frac{n}{C'} \sum_{\ell = 1}^{m} O\mparen{ \mparen{\ell + \log \frac{1}{\delta C'}}^2 \ell}\nonumber\\
        &\le \frac{n}{C'} \sum_{\ell=1}^{m} O\mparen{\ell^3 + \ell\log^2\frac{1}{\delta C'}} && ((a+b)^2 \le O\mparen{a^2 + b^2})\nonumber\\
        &\le \frac{n}{C'} O\mparen{m^4 + m^2\log^2\frac{1}{\delta C'}}\nonumber\\
        &\le O\mparen{\frac{n}{\eps C'^2}} + O\mparen{\frac{n}{C'}\log^2 \frac{1}{\eps C'}\log^2\frac{1}{\delta C'}}\,. && (m = O\mparen{\log\frac{1}{\eps C'}})\label{eq:local_opt_runtime_log_squared_part}
    \end{align}
    The left part of \Cref{eq:local_opt_runtime_log_squared_part} gets absorbed into the $O\mparen{\frac{n^2 \log n}{\eps C'^2}\log\frac{1}{\delta C'}}$.
\end{proof}

\subsection{Divide and conquer}

As written, \Cref{alg:local_opt_no_logs} assumes that we begin with a product state $\ket{\pi}$ having fidelity at least $2/3$ with $\rho$. This is not a true learning algorithm, then, because such a state $\ket{\pi}$ might not be known in advance. Nevertheless, we can straightforwardly generalize \Cref{alg:local_opt_no_logs} to a learning algorithm that only takes copies of $\rho$ as input, using a divide-and-conquer approach. This works as a consequence of the following lemma:

\begin{lemma}
\label{lem:fidelity_tensor}
    Let $\rho_{AB}$ be a state on a systems $A$ and $B$.
    Suppose that $\braket{\phi|\rho_A|\phi} \ge 1 - \eps_1$ and $\braket{\pi|\rho_B|\pi} \ge 1 - \eps_2$.
    Then
    \[
    \braket{\phi\pi|\rho_{AB}|\phi\pi} \ge 1 - \eps_1 - \eps_2\,.
    \]
\end{lemma}
\begin{proof}
    %\will{Shouldn't there be a dumber way to prove this lemma? Is it even ``obvious''? isn't this literally a union bound}
    Extend $\ket{\pi} \eqqcolon \ket{\pi_1}$ to an orthonormal basis $\{\ket{\pi_i} \mid i \in [d]\}$ for $B$. Since the partial trace can be computed by summing over any orthonormal basis, we have
    \begin{align*}
    1 - \eps_1 &\le
    \braket{\phi|\rho_A|\phi}\\
    &= 
    \sum_{i=1}^d \braket{\phi\pi_i|\rho_{AB}|\phi\pi_i}\\
    &\le \braket{\phi\pi|\rho_{AB}|\phi\pi} + \sum_{i=2}^{d}\braket{\pi_i|\rho_{B}|\pi_i}\\
    &= \braket{\phi\pi|\rho_{AB}|\phi\pi} + 1 - \braket{\pi|\rho_B|\pi}\\
    &\le \braket{\phi\pi|\rho_{AB}|\phi\pi} + \eps_2\,.
    \end{align*}
    The lemma follows by rearranging.
\end{proof}

The full learning algorithm is below; its correctness is self-explanatory.
Note that we must assume the existence of a product state with fidelity above $5/6$, instead of $2/3$ in the previous algorithm.
This is because of the loss incurred from combining the two halves via \Cref{lem:fidelity_tensor}.

\begin{longfbox}[breakable=false, padding=1em, margin-top=1em, margin-bottom=1em]
\begin{algorithm}
\label{alg:high_fidelity_alg}
    (High-fidelity product state agnostic learning).
\begin{description}
    \item[Input:] Copies of an $n$-qubit state $\rho$, $\eps \in (0, 1/6]$, $\delta \in (0, 1)$
    \item[Promise:] There exists a product state $\ket{\pi}$ satisfying $\braket{\pi|\rho|\pi} \ge 5/6 + \eps$
    \item[Output:] A product state $\ket{\pi}$ satisfying $\braket{\pi|\rho|\pi} \ge \max_{\text{product }\ket{\pi'}} \braket{\pi'|\rho|\pi'} - \eps$ with probability at least $1 - \delta$
    \item[Procedure:] \mbox{}
    \begin{algorithmic}[1]
        \State \; \Comment{First find a product state $\ket{\pi}$ such that $\braket{\pi|\rho|\pi} \ge 2/3$}
        \If{$n$ = 1}
            \State Use tomography to find a $\ket{\pi}$ satisfying $\braket{\pi|\rho|\pi} \ge 2/3$ with prob.\ $1 - \delta/2$ \label{line:rough_single_qubit_tomography}\;
        \Else
            \State $\rho_L \coloneqq$ the left half of $\rho$ \;
            \State $\ket{\pi_L} \coloneqq $ \Cref{alg:high_fidelity_alg}$(\rho_L, \eps, \delta/4)$ \; \Comment{$\braket{\pi_L|\rho_L|\pi_L} \ge 5/6$}
            \State $\rho_R \coloneqq$ the right half of $\rho$ \;
            \State $\ket{\pi_R} \coloneqq$ \Cref{alg:high_fidelity_alg}$(\rho_R, \eps, \delta/4)$ \; \Comment{$\braket{\pi_R|\rho_R|\pi_R} \ge 5/6$}
            \State $\ket{\pi} \coloneqq \ket{\pi_L} \otimes \ket{\pi_R}$ \; \Comment{$\braket{\pi|\rho|\pi} \ge 2/3$ by \Cref{lem:fidelity_tensor}}
        \EndIf
        \State\Return \Cref{alg:local_opt_no_logs}$(\rho, \ket{\pi}, \eps, \delta/2, 1/3 + \eps)$ \Comment{Error $\le \delta/4 + \delta/4 + \delta/2 = \delta$}\label{line:call_local_opt}\;
    \end{algorithmic}
\end{description}
\end{algorithm}
\end{longfbox}

\begin{theorem}
    \label{thm:high_fidelity_alg}
    \Cref{alg:high_fidelity_alg} has sample complexity
    \[
    O\mparen{\frac{n}{\eps} \log\frac{1}{\delta}}
    \]
    and runs in time
    \[
    O\mparen{\frac{n^2\log n}{\eps}\log\frac{1}{\delta} + n\log n\log^2\frac{1}{\eps}\log^2\frac{1}{\delta}} \le O\mparen{\frac{n^2\log n}{\eps}\log^2\frac{1}{\delta}}\,.
    \]
\end{theorem}

\begin{proof}
    For the sample complexity, observe that copies of $\rho$ can be shared between the two recursive calls to \Cref{alg:high_fidelity_alg}, because a copy of $\rho$ is both a copy of $\rho_L$ and a copy of $\rho_R$.
    
    When $n = 1$, it is clear that \Cref{line:rough_single_qubit_tomography} can be performed using $O\mparen{\log\frac{1}{\delta}}$ time and samples. For example, one can estimate the coordinates of $\rho$ on the Bloch sphere to some small constant precision by measuring repeatedly in the $X$, $Y$, and $Z$ bases. Therefore, the contribution of \Cref{line:rough_single_qubit_tomography} to sample complexity and runtime across recursive calls to \Cref{alg:high_fidelity_alg} are, respectively, at most $O\mparen{\log\frac{n^2}{\delta}} \le O\mparen{\log n\log\frac{1}{\delta}}$ and 
    $O\mparen{n\log\frac{n^2}{\delta}} \le O\mparen{n\log n\log\frac{1}{\delta}}$.
    These are dominated by the other terms.

    \Cref{line:call_local_opt} accounts for the remaining algorithmic complexity.
    By \Cref{cor:local_opt_complexity}, because we pick $C > 1/3$, the overall sample complexity due to \Cref{alg:local_opt_no_logs} across recursive calls is at most
    \[
        \sum_{i=0}^{\lceil\log_2 n\rceil} O\mparen{\frac{n2^{-i}}{\eps}\log\frac{4^i}{\delta}} \le \frac{n}{\eps}\log\frac{1}{\delta}\sum_{i=0}^{\lceil\log_2 n\rceil} O\mparen{\frac{i}{2^i}} \le O\mparen{\frac{n}{\eps} \log\frac{1}{\delta}}\,.
    \]
    The time complexity bound is
    \[
        \sum_{i=0}^{\lceil\log_2 n\rceil} 2^i \cdot O\mparen{\frac{\mparen{n2^{-i}}^2 \log(n2^{-i})}{\eps}\log\frac{4^i}{\delta} + n2^{-i}\log^2 \frac{1}{\eps}\log^2\frac{4^i}{\delta}}\,.
    \]
    We handle the two terms within the summation separately. The left term is bounded by:
    \begin{equation}
        \label{eq:high_fidelity_left_runtime}
        \frac{n^2 \log n}{\eps}\log\frac{1}{\delta} \sum_{i=0}^{\lceil\log_2 n\rceil} O\mparen{\frac{i}{2^i}} \le O\mparen{\frac{n^2\log n}{\eps}\log\frac{1}{\delta}}\,.
    \end{equation}
    The right term is bounded by:
    \begin{align*}
        n\log^2\frac{1}{\eps} \sum_{i=0}^{\lceil\log_2 n\rceil} O\mparen{\log^2\frac{4^i}{\delta}}
        &\le n\log^2\frac{1}{\eps} \sum_{i=0}^{\lceil\log_2 n\rceil} O\mparen{i^2 + \log^2\frac{1}{\delta}} \tag{$(a+b)^2 \le O\mparen{a^2 + b^2}$} \\
        &\le O\mparen{n\log^3n \log^2\frac{1}{\eps} + n\log n\log^2\frac{1}{\eps}\log^2\frac{1}{\delta}}\,.
    \end{align*}
    The theorem follows because $n\log^3n \log^2\frac{1}{\eps}$ is bounded by \Cref{eq:high_fidelity_left_runtime}.
\end{proof}

\section{Agnostic learning of product states}
\label{sec:reduction-to-poly-opt}

In this section, we give our general algorithm for finding product states which have good fidelity with an input state $\rho$.
Our output will take the form of a ``good'' product state cover, as given below.

\begin{definition}[Good product state cover] \label{def:product-cover}
    A collection of pure product states on $m$ qubits, $\mathcal{C} = \{\ket{\pi_{i}}\}_{i}$ is a \emph{$(\eta, \eps, b, B)$-good cover for a state $\rho$} if 
    \begin{enumerate}
        \item For all $\ket{\pi} \in \mathcal{C}$, $\bra{\pi} \rho \ket{\pi} \geq \eta - \eps$.
        \item For all distinct $\ket{\pi}, \ket{\pi'} \in \mathcal{C}$, $\dtan(\ket{\pi}, \ket{\pi'}) \geq b$.
        \item For all product states on $m$ qubits, $\ket{\phi}$, such that $\bra{\phi}\rho\ket{\phi} \geq \eta$, there exists $\ket{\pi} \in \mathcal{C}$ such that $\dtan(\ket{\phi}, \ket{\pi}) \leq B$.
    \end{enumerate}
    We will eventually take $b = 2/\eta$ and $B = 3/\eta$, so for brevity a $(\eta, \eps)$-good cover refers to a $(\eta, \eps, 2/\eta, 3/\eta)$-good cover.\footnote{
        The properties our parameters must satisfy are that $\eta - \eps - \frac{1}{b} > 0$ (for \cref{claim:cover_size}) and that $B > b$ (for the greedy algorithm to succeed).
    }
\end{definition}

It may not yet be clear that a good product state cover for $\rho$ even exists.
When $B \geq b$, a greedy approach works here: start with $\mathcal{C} = \varnothing$, and while property 3 does not hold for $\mathcal{C}$, add the violating $\ket{\phi}$ to $\mathcal{C}$.
This approach will be used and adapted to be more tractable in \cref{alg:outer_loop}.
We also soon show that the size of good product state convers is small (\cref{claim:cover_size}).
First, we state the theorem we will prove in this section.

\begin{theorem}[Agnostic learning of product states] \label{thm:main} 
    Let $\rho$ be an $n$-qubit state and suppose we are given error parameters $\eta \in (0,1)$, $\eps \in (0, \eta/3)$, and $\delta \in (0, 1)$.
    Then there is an algorithm which, with probability $\geq 1-\delta$, outputs an $(\eta, \eps)$-good product state cover for $\rho$.
    The algorithm uses $N \leq (\poly(n))^{1/\eta^2 + \log\frac{1}{\eps}}\poly(\log\frac{1}{\delta})$ copies of $\rho$, $\poly(n, 1/\eta, \log\frac{1}{\eps})$ quantum gates per copy of $\rho$, and $n^{\poly(1/\eps)} \poly(\log\frac{1}{\delta})$ additional classical overhead.
\end{theorem}

\begin{remark}[Applications of \cref{thm:main}]
    \label{rmk:applications_of_thm}
    We use this remark to note some immediate corollaries of the above algorithm.

    First, it can be used to deduce $\OPT = \max_{\ket{\pi}} \bra{\pi} \rho \ket{\pi}$, the maximum fidelity a product state has with $\rho$, to a specified error $2\eps$.
    This is because the cover $\mathcal{C}$ output by the algorithm contains at least one product state with fidelity at least $\eta - \eps$, if a product state with fidelity at least $\eta$ exists.
    So, we can start with $\eta = 1/2$ and perform binary search on the choice of $\eta$, reducing $\eta$ when the output cover is empty and increasing it when it is non-empty.
    After $O(\log(1/\eps))$ iterations (and with an appropriate choice of $\delta$), $\eta$ will be an $\eps$-good estimate.
    This also gives a product state $\ket{\pi}$ such that $\bra{\pi} \rho \ket{\pi} \geq \OPT - 2\eps$.
    The running time for this algorithm is also $n^{\poly(1/\eps)} \poly(\log\frac{1}{\delta})$, since if $\eta$ ever drops below $\eps$, $0$ is a suitable output, and the number of copies used is at most $(\poly(n))^{1/(\max(\OPT, \eps))^2 + \log\frac{1}{\eps}}\poly(\log\frac{1}{\delta})$.
\end{remark}

%Now we show that a good product state cover is a well behaved object.
%This includes showing that it can not be too large, and that we can construct a net in the tangent distance.
%We first show that a good product state cover is not too large.
Now, we show that a good product state cover cannot be too large.

\begin{claim}[Size of a good cover] \label{claim:cover_size}
    Let $\mathcal{C}$ be a $(\eta, \eps, b, B)$-good cover for a state with density matrix $\rho$.
    Then $\abs{\mathcal{C}} \leq \frac{1}{\eta - \eps - \frac{1}{b}}$, provided $\eta - \eps - \frac{1}{b} > 0$.
\end{claim}
\begin{proof}
    Let $\mathcal{C} = \{\ket{\pi^{(i)}}\}_i$, and let $M$ be the matrix whose $i$th column is $\ket{\pi^{(i)}}$.
    Consider the Gram matrix $M^{\dagger}M$; its $(i, j)$th entry is $\braket{\pi^{(i)} | \pi^{(j)}}$.
    When $i = j$, this entry is $1$, and otherwise, the entry has bounded magnitude:
    \begin{align*}
        \abs{\braket{\pi^{(i)} | \pi^{(j)}}} \leq \frac{1}{\dtan(\ket{\pi^{(i)}}, \ket{\pi^{(j)}})} \leq \frac{1}{b},
    \end{align*}
    where the first step follows from \cref{lem:dtan_fidelity_relationship} and the second from property 2 of the definition of a good product state cover.
    So, every row has one diagonal entry of value $1$ and the other $\abs{\mathcal{C}} - 1$ entries are bounded by $\frac{1}{b}$.
    A consequence of the Gershgorin circle theorem is that the operator norm of a symmetric matrix is bounded by the maximum sum of absolute values of all the entries in a single column.
    So,
    \begin{equation*}
        \opnorm{M}^2 = \opnorm{M^{\dagger} M} \leq 1 + \frac{|C|}{b}\,.
    \end{equation*}
    This gives us the upper bound $ \leq 1 + |\mathcal{C}| / b$. We can similarly lower bound the operator norm of $M$:
    \begin{align*}
        \opnorm{M}^2 &= \opnorm{MM^\dagger}
        \geq \Tr[MM^{\dagger} \rho]
        = \sum_{i} \bra{\pi_i} \rho \ket{\pi_i}
        \geq |\mathcal{C}|(\eta - \eps) \,.
    \end{align*}
    Here we use the definition of the operator norm and property 1 of the definition of a good product state cover.
    Putting both bounds together, we have that
    \begin{align*}
        \abs{\mathcal{C}} \leq \frac{1}{\eta - \eps - \frac{1}{b}}
    \end{align*}
    as desired.
\end{proof}

\subsection{Finding a good product state cover}
Now we present an iterative algorithm that builds a good product state cover as it sweeps along the registers of the input state.

\begin{longfbox}[breakable=false, padding=1em, margin-top=1em, margin-bottom=1em]
\begin{algorithm}[Extending a good product state cover]
\label{alg:outer_loop}\mbox{}
    \begin{description}
    \item[Input:] Copies of a state $\rho$; parameters $\eta \in (0, 1)$ and $\eps \in (0, \frac{\eta}{3})$
    \item[Output:] $\mathcal{C}$, a $(\eta, \eps)$ good product state cover for $\rho$.
    \item[Subroutine:] We assume the existence of an oracle which, given a set of product states $\mathcal{C}$, copies of the state $\rho$, and a ``root'' product state $\ket{\varphi}$, either outputs a classical description of a product state $\ket{\pi}$ or $\bot$.
    The output is guaranteed to satisfy
    \begin{enumerate}[label=(\alph*)]
        \item $\bra{\pi}\rho\ket{\pi} \geq \eta - \eps$;
        \item For all $\ket{\pi'} \in \mathcal{C}_{k}$, $\dtan(\ket{\pi}, \ket{\pi'}) \geq 2/\eta$;
    \end{enumerate}
    If there is a product state $\ket{\pi}$ such that
    \begin{enumerate}[label=(\alph*')]
        \item $\bra{\pi}\rho\ket{\pi} \geq \eta$;
        \item For all $\ket{\pi'} \in \mathcal{C}_{k}$, $\dtan(\ket{\pi}, \ket{\pi'}) \geq 3/\eta$;
        \item $\dtan(\ket{\pi}, \ket{\varphi}) \leq 4/\eta$;
    \end{enumerate}
    then the oracle is guaranteed to not output $\bot$.
    \item[Procedure:] \mbox{}
    \begin{algorithmic}[1]
        \State Let $\mathcal{N} \subset \C^2$ be a $1$-tangent distance net over qubits;
        \State Let $\mathcal{C}_0 = \{1\}$;
        \For{$k$ from $1$ to $n$}
            \LComment{Create an $(\eta, \eps)$-good product state cover for $\rho_{[k]}$}
            \State Let $\mathcal{C}_{k} = \varnothing$;
            \Loop
                \State Call the oracle on $\mathcal{C}_k$ and $\rho_{[k]}$ and all $\ket{\varphi} \in \mathcal{C}_{k-1} \otimes \mathcal{N}$;
                \label{line:add-to-cover}
                \State If any of the calls return a product state $\ket{\pi} \in (\mathbb{C}^2)^{\otimes k}$, add it to $\mathcal{C}_k$;
                \State Otherwise, exit the loop;
            \EndLoop
        \EndFor
        \State Output $\mathcal{C}_n$;
    \end{algorithmic}
    \end{description}
\end{algorithm}
\end{longfbox}

\begin{remark}[An explicit 1-tangent distance net]
    In the algorithm, we need a net over qubit states $\net$ such that, for every state $\ket{\phi}$, there is a $\ket{\varphi} \in \net$ such that $\dtan(\ket{\phi}, \ket{\varphi}) < 1$.
    By \cref{def:tangent-n1}, considering the states on the Bloch sphere, this means that the angle $\theta$ between $\ket{\phi}$ and $\ket{\varphi}$ is less than $\pi/2$.
    So, we can form such a net just by picking the states on the axes of the Bloch sphere: $\net = \braces{\ket{0}, \ket{1}, \ket{+}, \ket{-}, \ket{+\ii}, \ket{-\ii}}$.
\end{remark}

\begin{claim} \label{claim:outer}
    \cref{alg:outer_loop} outputs a $(\eta, \eps)$ good product state cover for $\rho$, requiring $O(\frac{n}{\eta^2})$ runs of the subroutine and $O(\frac{n^2}{\eta^2})$ classical overhead.
\end{claim}
\begin{proof}
%Consider a modified version of the algorithm where requirement (c) in \cref{line:add-to-cover} is removed.
We will show that $\mathcal{C}_k$ is an $(\eta, \eps)$-good product state cover for $\rho_{[k]}$, by induction on $k$.
We consider the oracle's behavior in \cref{line:add-to-cover}, when run on $\mathcal{C}_k$ and $\rho_{[k]}$.
First, we observe that, for all $k$, because conditions (a) and (b) of the subroutine guarantee are identical to properties 1 and 2 of \cref{def:product-cover} with respect to $\rho_{[k]}$, $\mathcal{C}_k$ will always obey such properties.
So, it suffices to show that $\mathcal{C}_k$ also obeys property 3 at the end of the ``repeat'' loop.

We first show for $k = 1$.
Consider some product state $\ket{\phi}$ such that $\bra{\phi} \rho_{[k]} \ket{\phi} \geq \eta$.
Then $\ket{\phi}$ satisfies condition (a') and (c') in \cref{line:add-to-cover} for some element of $\mathcal{C}_0 \otimes \mathcal{N}$: since $\mathcal{N}$ is a net, there is some $\ket{\varphi} \in \mathcal{C}_0 \otimes \mathcal{N} = \mathcal{N}$ such that $\dtan(\ket{\varphi}, \ket{\phi}) < 1 \leq 4/\eta$.
Because the oracle output $\bot$ when run on $\mathcal{C}_k$, this means that (b') must not be satisfied for $\ket{\phi}$.
In other words, some $\ket{\pi} \in \mathcal{C}_1$ satisfies $\dtan(\ket{\pi}, \ket{\phi}) < \frac{3}{\eta}$.
This shows that $\mathcal{C}_1$ satisfies property 3 of the product state cover.

For $k > 1$, again consider a product state $\ket{\phi} = \ket{\phi_1} \dots \ket{\phi_k}$ such that $\bra{\phi} \rho_{[k]} \ket{\phi} \geq \eta$.
Then $\ket{\phi}$ satisfies condition (a').
Further, it satisfies condition (c') for some $\ket{\varphi} \in \mathcal{C}_{k-1} \otimes \mathcal{N}$: since $\mathcal{C}_{k-1}$ is a good product state cover and $(\bra{\phi_1}\dots\bra{\phi_{k-1}}) \rho_{[k-1]} (\ket{\phi_1}\dots\ket{\phi_{k-1}}) \geq \eta$, there exists a product state $\ket{\nu}$ in $\mathcal{C}_{k-1}$ such that $\dtan(\ket{\phi_1}\dots\ket{\phi_{k-1}}, \ket{\nu}) \leq 3/\eta$.
Then, there is some $\ket{\varphi_k}$ such that $\dtan(\ket{\phi}, \ket{\nu}\ket{\varphi_k}) \leq \sqrt{(\frac{3}{\eta})^2  + 1} \leq \frac{4}{\eta}$ as claimed.

Thus, because of the guarantee of the oracle, after the repeat loop terminates, condition (b') cannot be true for $\ket{\phi}$.
So, there is a $\ket{\pi} \in \mathcal{C}_k$ such that $\dtan(\ket{\phi}, \ket{\pi}) < 3/\eta$.
This shows that $\mathcal{C}_k$ satisfies property 3 of the product state cover.

Every product state cover satisfies $\abs{\mathcal{C}_k} \leq 1/(\eta - \eps - \eta/2) \leq 6/\eta$ by \cref{claim:cover_size}, so the subroutine only needs to be run at most $n(6/\eta)^2\abs{\mathcal{N}}$ times.
The only additional overhead is the task of storing the cover, which takes $O(n)$ time with a classical computer per element added.
\end{proof}

\subsection{Finding candidate product states} \label{subsec:finding}

Now, we specify how to perform \cref{line:add-to-cover} in \cref{alg:outer_loop}.
We restate the goal of that subroutine here.

\begin{lemma} \label{lem:IT}
    Suppose we are given a set of $r$ product state constraints $\{(\vec{a}^{(s)}, b)\}_{s \in [r]}$ where $\vec{a}^{(s)} \in \C^m$ and $b > 0$, a description of a known ``root'' product state $\ket{\varphi} \in (\C^2)^{\otimes m}$, and error parameters $\eta \in (0,1)$, $\eps \in (0, \eta/3)$, and $\delta \in (0,1)$.
    Then there is an algorithm which, with probability $\geq 1-\delta$, successfully performs the subroutine as specified in \cref{alg:outer_loop}: it outputs either $\bot$ or a $\vec{z} \in \C^m$ such that
    \begin{enumerate}[label=(\alph*)]
        \item $\bra{\pi_{\vec{z}}}\rho\ket{\pi_{\vec{z}}} \geq \eta - \eps$;
        \item For all $s \in [r]$, $\dtan(\ket{\pi_{\vec{z}}}, \ket{\pi_{\vec{a}^{(s)}}}) \geq b$.
    \end{enumerate}
    If there is a product state $\ket{\pi}$ such that
    \begin{enumerate}[label=(\alph*')]
        \item $\bra{\pi}\rho\ket{\pi} \geq \eta$;
        \item For all $s \in [r]$, $\dtan(\ket{\pi_{\vec{z}}}, \ket{\pi_{\vec{a}^{(s)}}}) \geq 1.5b$;
        \item $\dtan(\ket{\pi}, \ket{\varphi}) \leq B$;
    \end{enumerate}
    then the output is guaranteed to not be $\bot$.
    The algorithm uses $N \leq (\poly(m))^{(B^2 + \log\frac{1}{\eps})}\log\frac{1}{\delta}$ copies of $\rho$, $\poly(m, B, \log\frac{1}{\eps})$ quantum gates per copy of $\rho$, and $m^{\poly(r, B, b, 1/b, 1/\eps)} \poly(\log\frac{1}{\delta})$ additional classical overhead.
\end{lemma}
From this, we can immediately conclude our main result:
\begin{proof}[Proof of \cref{thm:main}]
    By \cref{claim:outer}, to construct the cover, it suffices to call \cref{lem:IT} $\poly(n, 1/\eta)$ times with parameters $m \leq n$, $B = 4/\eta$, $b = 3/\eta$, and a set of product state constraints where $r = O(1/\eta)$ by \cref{claim:cover_size}.
    There is a failure of probability associated to each run of the subroutine, but the failure probability parameter can be rescaled such that the probability of all calls succeeding is at least $\geq 1-\delta$.
    This gives the stated running time, and associated quantum complexities.
\end{proof}

The rest of this section is devoted to proving \cref{lem:IT}.
The algorithm is given in \cref{alg:inner_loop}; we spend the rest of this subsection describing the intuition for this algorithm.
We prove the desired complexity bounds in \cref{claim:fast}, and then we prove the above guarantees in \cref{claim:sound} and \cref{claim:complete}.
The complexity bound requires the use of a polynomial optimization routine, which is described and proved later, in \cref{sec:opt}.

\paragraph{Algorithm intuition.}
Our goal is to search within tangent distance $B$ of the root state for a product state that has good overlap with $\rho$, if one exists.
Further, we have additional constraints that the state we find be far away from a collection of $r$ product states in tangent distance.
Without loss of generality, we can take the root state to be $\ket{\pi_{\vec{0}}} = \ket{0^n}$.
Then it suffices to find an (approximate) solution to the following optimization problem, which finds the product state with the best fidelity under these constraints:
\begin{equation} \label{program:true}
\begin{aligned}
    & \underset{\vec{z} \in \C^m}{\text{maximize}}
    & & \bra{\pi_{\vec{z}}} \rho \ket{\pi_{\vec{z}}} \\
    & \text{subject to}
    & & \dtan(\ket{\pi_{\vec{z}}}, \ket{\pi_{\vec{a}^{(s)}}}) \geq b \text{ for all } s \in [r], \\
    &&& \dtan(\ket{\pi_{\vec{z}}}, \ket{\pi_{\vec{0}}}) \leq B.
\end{aligned}
\tag{P1}
\end{equation}

We would like to reduce this task to one of optimizing a low-degree polynomial under simple constraints.
The criterion $\dtan(\ket{\pi_{\vec{z}}}, \ket{\pi_{\vec{0}}}) \leq B$ is equivalent to $\twonorm{\vec{z}} \leq B$.
However, the other tangent distance constraints do not simplify so easily.
We can simplify them for small coordinates, and because $\vec{z}$ is bounded norm, most coordinates are small; we encode this into the program:
\begin{equation} \label{program:guess}
\begin{aligned}
    & \underset{\vec{z},\,\diffy{S \subseteq [m]}}{\text{maximize}}
    & & \bra{\pi_{\vec{z}}} \rho \ket{\pi_{\vec{z}}} \\
    & \text{subject to}
    && \diffy{\abs{\vec{z}_i} \leq \mu \text{ for all } i \not\in S} \\
    %&&& \diffy{\abs{\vec{a}_i^{(s)}} \leq \blacklozenge \text{ for all } i \not\in S} \\
    &&& \diffy{\abs{S} \leq B^2/\mu^2} \\
    &&& \dtan(\ket{\pi_{\vec{z}}}, \ket{\pi_{\vec{a}^{(s)}}}) \geq b \text{ for all } s \in [r], \\
    &&& \diffy{\twonorm{\vec{z}} \leq B}.
\end{aligned}
\tag{P2}
\end{equation}
This program gives an equivalent solution to \eqref{program:true}.
Now we can approximate the farness constraints by an $\ell^2$ constraint for the coordinates outside of $S$.
\begin{equation} \label{program:subspace}
\begin{aligned}
    & \underset{\vec{z},\,S \subset [m]}{\text{maximize}}
    & & \bra{\pi_{\vec{z}}} \rho \ket{\pi_{\vec{z}}} \\
    & \text{subject to}
    && \abs{\vec{z}_i} \leq \mu \text{ for all } i \not\in S \\
    &&& \abs{S} \leq B^2/\mu^2 \\
    %&&& [C'] = \{s \in [C]: \abs{\vec{a}_i^{(s)}} \leq (b + 1)/ \kappa \text{ for all } i \not\in S\}\\
    &&& \diffy{\dtan(\ket{\pi_{\vec{z}_S}}, \ket{\pi_{\vec{a}^{(s)}_S}})^2 + \norm{\vec{z}_{\overline{S}} - \vec{a}_{\overline{S}}^{(s)}}_2^2 \geq 1.5 b^2} \text{ for all } s \in [r], \\
    &&& \twonorm{\vec{z}} \leq B.
\end{aligned}
\tag{P3}
\end{equation}
Our first key lemma (\cref{claim:guess-to-subspace}) shows that this is a stronger constraint than that of \eqref{program:guess}, but still satisfies for completeness, so the program will still find a good product state under the desired circumstances.

Finally, we would like to approximate the objective function $\braket{\pi_{\vec{z}} | \rho | \pi_{\vec{z}}}$ by a low-degree polynomial.
We do this by replacing $\rho$ with a truncation: for $d \in [m]$, let $\rho_{d} = \Pi_{\leq d} \rho \Pi_{\leq d}$, where $\Pi_{\leq d}$ is the projector onto computational basis strings with Hamming weight less than or equal to $d$.
Then $\braket{\pi_{\vec{z}} | \rho_d | \pi_{\vec{z}}}$ is a degree-$d$ polynomial multiplied by the normalization $\prod_{i \in [m]} \frac{1}{1 + \abs{z_i}^2}$.
Because most coordinates of $\vec{z}$ are small, this normalization can be captured by a simpler quantity: $\prod_{i \in \overline{S}} \frac{1}{1 + \abs{z_i}^2} \approx e^{-\twonorm{\vec{z}_{\overline{S}}}^2}$.
After taking these approximations, we have the following optimization problem.
\begin{equation} \label{program:function}
\begin{aligned}
    & \underset{\vec{z},\,S \subset [m]}{\text{maximize}}
    & & \diffy{p_{\vec{z}_S}(\vec{z}_{\overline{S}}) = \frac{e^{-\twonorm{\vec{z}_{\overline{S}}}^2}}{\prod_{i \in S} (1 + \abs{z_i}^2)} \sum_{x, x' \in \{0,1\}^m} \bra{x} \rho_{d} \ket{x'} (\vec{z}^*)^{x}(\vec{z})^{x'}} \\
    & \text{subject to}
    && \abs{z_i} \leq \mu \text{ for all } i \not\in S \\
    &&& \abs{S} \leq B^2/\mu^2 \\
    &&& \dtan(\ket{\pi_{\vec{z}_S}}, \ket{\pi_{\vec{a}^{(s)}_S}})^2 + \norm{\vec{z}_{\overline{S}} - \vec{a}_{\overline{S}}^{(s)}}_2^2 \geq 1.5b^2 \text{ for all } s \in [r], \\
    &&& \twonorm{\vec{z}} \leq B.
\end{aligned}
\tag{P4}
\end{equation}
Above, for a vector $\Vec{z} \in \mathbb{C}^{m}$ and string $x \in \{0, 1\}^{m}$, we use the notation $(\Vec{z})^{x} = \prod_{x_i = 1} z_i$.
Our second key lemma (\cref{claim:subspace-to-function}) shows that the objective functions of \eqref{program:function} and \eqref{program:subspace} are close on the domain.

\eqref{program:function} is the optimization program implicitly being run by the eventual algorithm (\cref{alg:inner_loop}).
There are three major differences.
First, we do not know $\rho_d$, so we perform tomography on $\rho$ to get some estimate $\varrho_d$ which we use instead (\cref{lem:subspace-tomography}).
This is efficient because we are performing tomography on a small subspace $\Pi_d$.
Second, the above program still has poorly behaved elements.
However, these are all located in the subspace $V^{(S)}$ spanned by the coordinates of $S$ and the $\vec{a}^{(s)}$'s.
If we knew the true value of $\vec{z}$ on $V^{(S)}$ (and the true value of $\twonorm{\vec{z}_{\overline{S}}}$), then the objective function becomes a degree-$d$ polynomial and the constraints all become simple $\ell_2$ and $\ell_\infty$ constraints, which is a form we can solve (\cref{thm:opt-main}).
However, this is a low-dimensional subspace, so we can simply guess the value of $\vec{z}$ on this subspace, and pick the best solution among all guesses.
Finally, we need to allow for some tolerance $\toleps$ in our argument, because of the error in guessing and the error of the final polynomial optimization algorithm.

\subsection{Algorithm and showing running time}
We begin with a lemma for the form of state tomography we need for the algorithm.
We did not attempt to optimize this, but we note that this is the only ``quantum'' part of the algorithm, and any tomography algorithm could be used here.
Here, we choose one which only performs single-copy Clifford measurements.
This is also the only part of the algorithm which is random, and so this is the only place where there is a probability of failure.
For our correctness proofs, we will assume that this tomography step always completes successfully.

\begin{lemma}[Computationally efficient state tomography] \label{lem:subspace-tomography}
    There is an algorithm which, given copies of $\rho$ and a parameter $d$, outputs a positive semi-definite matrix $\hat{\rho}$ with trace at most 1 and supported on the strings of Hamming weight $\leq d$, such that $\opnorm{\hat{\rho} - \Pi_{\leq d}\rho \Pi_{\leq d}} < \eps$.
    This algorithm uses $N = O(\frac{(10n)^{2d}}{\eps^2} \log\frac{1}{\delta})$ copies of $\rho$, along with $\poly(n, d)$ quantum gates per copy of $\rho$ and $\poly(N)$ classical processing.
\end{lemma}
\begin{proof}
    First, similarly to \Cref{lem:estimating_z}, for every copy of $\rho$, we can attach $d \ceil{\log_2(n + 1)} + O(\log d)$ ancilla qubits, and then apply a $\poly(n, d)$-sized circuit which performs the following unitary.
    For a set $J \subseteq [n]$ with elements $i_1 < i_2 < \dots < i_{\abs{J}}$, the circuit maps $\ket{0^d} \ket{e_J} \mapsto \ket{i_1}\ket{i_2}\dots\ket{i_{\abs{J}}} \ket{0^{d - \abs{J}}} \ket{e_J}$ when $\abs{J} \leq d$, and does nothing to $\ket{0^d}\ket{e_J}$ when $\abs{J} > d$.
    Here, we use $e_J$ to denote the $n$-qubit state which is $\ket{1}$ on the qubits in $J$ and $\ket{0}$ otherwise.
    After applying this circuit and discarding the $n$-qubit state, the density matrix of the $d \ceil{\log_2(n + 1)}$ ancilla qubits contains $\Pi_{\leq d} \rho \Pi_{\leq d}$ as a submatrix.
    So, it suffices to perform tomography on this density matrix of size $D \leq (10n)^d$, which we denote $\sigma$, and output the corresponding matrix, which is $\Pi \sigma \Pi$ for some projector $\Pi$ on some subset of computational basis states.

    We use a simple, gate-efficient tomography algorithm.
    Consider the process that measures a random Clifford circuit, $C$, performs $C^{\dagger}$ on the $\sigma$ and measures in the computational basis to get $\ket{i}$, and then taking the estimator $((D+1) C \proj{i} C^{\dagger} - \id)$.
    If we perform this process $N$ times and then average the estimator to form $\hat{\sigma}$, it satisfies the following \cite[Section 1.5.2]{lowe2021learning}:
    \begin{equation*}
        \E \fnorm{\hat{\sigma} - \sigma}^2 \leq \frac{D^2 + D - 1}{N} \,.
    \end{equation*}
    In the language of Flammia and O'Donnell~\cite{fo24}, this is a state estimation algorithm with Frobenius-squared rate $\frac{D^2 + D - 1}{N}$, and so by \cite[Proposition 3.10]{fo24}, with only a constant factor loss in the guarantee, we can modify the algorithm such that it always outputs a matrix $\hat{\sigma}$ which is positive semi-definite and satisfies $\tr(\hat{\rho}) = 1$.
    Further, by \cite[Proposition 3.11]{fo24}, the guarantee can be upgraded to the guarantee that $\fnorm{\hat{\sigma} - \sigma}^2 \leq \frac{c(D^2 + D - 1)}{N} \log\frac{1}{\delta}$ with probability $\geq 1-\delta$, for some sufficiently large constant.
    Since $\opnorm{X} \leq \fnorm{X}$, for some choice of $N = O(\frac{D}{\eps}\log\frac{1}{\delta})$, we have that $\opnorm{\hat{\sigma} - \sigma} < \eps$ with probability $\geq 1-\delta$.
    Finally, this operator norm bound also bounds the operator norm distance for every submatrix, so we can conclude that for $\hat{\rho}$ formed by re-labeling the rows and columns of $\Pi \hat{\sigma} \Pi$, $\opnorm{\hat{\rho} - \Pi_{\leq d} \rho \Pi_{\leq d}} < \eps$.
\end{proof}

For the following algorithm, and throughout this section, for a vector $\vec{v}$ we use the notation $\vec{v}_S \in \C^{\abs{S}}$ to denote the vector restricted to the coordinates of $S$, and for two vectors $\vec{u}_S, \vec{v}_{\overline{S}}$, the notation $(\vec{u}_S, \vec{v}_{\overline{S}})$ denotes the vector which is $\vec{u}$ on the coordinates corresponding to $S$ and $\vec{v}$ on the coordinates corresponding to $\overline{S}$.

\begin{longfbox}[breakable=false, padding=1em, margin-top=1em, margin-bottom=1em]
\begin{algorithm}[Adding a new element to the cover]
\label{alg:inner_loop}\mbox{}
    \begin{description}
    \item[Input:] A set of product state constraints $\{(\vec{a}^{(s)}, b)\}_{s \in [r]}$; copies of an $m$-qubit state $\rho$; an explicit known ``root'' product state $\ket{\varphi} \in (\C^2)^{\otimes m}$; parameters $B, \eta, \eps, \delta$.
    \item[Output:] Either $\bot$ or a $\vec{z} \in \C^m$ with guarantees as in \cref{lem:IT}
    \item[Procedure:]\mbox{}
    \begin{algorithmic}[1]
        \State Recenter $\ket{\varphi}$ to be $\ket{0^m}$;
        \State Let $\appreps = \eps/100$;
        \Comment{$\appreps$ is for approximation errors.}
        \State Choose $\mu \leq \frac{1}{10} \min(b, \frac{1}{b}, \frac{\sqrt{\appreps}}{B})$;
        \Comment{$\mu$ is the eventual $\ell_\infty$ bound; we need the $b$ upper bound for \cref{claim:guess-to-subspace}, and the $\sqrt{\eps}/B^2$ upper bound for \cref{claim:subspace-to-function}}
        \State Let $d = 10 B^2 + \log\frac{2}{\appreps}$;
        \Comment{$d$ is the degree of the eventual polynomial}
        \State Let $\toleps \leq \min(\gamma, \mu^4 / \sqrt{m}, 0.01\eps) = \poly(1/m, \mu, \eps)$;
        \Comment{Here, $\gamma$ is the tolerance parameter coming from \cref{thm:opt-main}; the other parts are used in the proof of soundness (\cref{claim:sound}).}
        \State Perform tomography on $\rho_d = \Pi_{\leq d}\rho \Pi_{\leq d}$ to get a description of $\varrho_d$ such that $\opnorm{\varrho_d - \rho_d} \leq \appreps$ with probability $\geq 1-\delta$, where $\Pi_{\leq d} = \sum_{\abs{x} \leq d} \ketbra{x}{x}$ (\cref{lem:subspace-tomography});
        \Comment{Note that we will use the additional properties of $\varrho_d$ guaranteed by \cref{lem:subspace-tomography}.
        In particular, $\varrho_d$ is supported on the image of $\Pi_{\leq d}$.}
        \ForAll{subsets $S \subseteq [m]$ of size $\leq B^2/\mu^2$}
            \State Let $\vsub$ be the subspace spanned by $\braces{\vec{a}^{(s)}}_{s \in [r]}$ and the computational basis vectors associated to $S$;
            \State Let $\net_S$ be an $\toleps$-net over the space $\calF_S$, where
            \begin{multline} \label{eq:final-feasible}
                \calF_S = \Big\{\vec{v} \in \vsub,\, \nu \in [0, B] \,\Big|\, \twonorm{\vec{v}_S}^2 + \nu^2 \leq B^2,\, \twonorm{\vec{v}} \leq B,\, \\
                \nu^2 - \twonorm{\vec{v}_{\overline{S}}}^2 + \twonorm{\vec{v}_{\overline{S}} - \vec{a}_{\overline{S}}^{(s)}}^2 \geq 1.5b^2 - \dtan(\ket{\pi_{\vec{v}_S}}, \ket{\pi_{\vec{a}^{(s)}_S}})^2 \text{ for all } s \in [r]
                \Big\}.
            \end{multline}
            \Comment{The net enforces that the output is far from the $\vec{a}^{(s)}$'s; the error parameter needs to be at most $\toleps$, and is used to show completeness (\cref{claim:complete}). The formal criteria the net satisfies is given in \cref{claim:net}.}
            \For{$(\vec{v}, \nu) \in \net_S$}
                \State Consider the domain
                \begin{multline} \label{program:final-domain}
                    \mathcal{D}_{\toleps} = \braces{
                        \vec{z}_{\overline{S}} \in \C^{\abs{\overline{S}}}
                        \mid \abs{\twonorm{\vec{z}_{\overline{S}}} - \nu} \leq \toleps,\,
                        \twonorm{\Pi_\vsub \vec{z} - \vec{v}} \leq \toleps, \\
                        \infnorm{\vec{z}_{\overline{S}}} \leq \mu + \toleps
                        \text{ for } \vec{z} = (\vec{v}_S,\,\vec{z}_{\overline{S}})};
                \end{multline}
                \State Use \cref{thm:opt-main} to find a $\vec{y}_{\overline{S}} \in \mathcal{D}_{2\toleps}$ such that
                \begin{multline} \label{program:final}
                    p_{\vec{v}_S, \nu}(\vec{y}_{\overline{S}}) \geq \max_{\vec{z}_{\overline{S}} \in \mathcal{D}_{\toleps}} p_{\vec{v}_S, \nu}(\vec{z}_{\overline{S}}) - \appreps \\
                    \text{for } p_{\vec{z}_S, \nu}(\vec{z}_{\overline{S}}) = \frac{e^{-\nu^2}}{\prod_{i \in S} (1 + \abs{z_i}^2)} \sum_{x, x' \in \{0,1\}^m} \bra{x} \varrho_d \ket{x'} (\vec{z}^*)^{x}\vec{z}^{x'},
                    \tag{P0}
                \end{multline}
                \State Add $\vec{y} = (\vec{v}_{S}, \vec{y}_{\overline{S}})$ to the list of candidate solutions, along with its objective value $p_{\vec{y}} \gets p_{\vec{v}_S, \nu}(\vec{y}_{\overline{S}})$;
            \EndFor
        \EndFor
        \State Let $\vec{u}$ be the candidate solution achieving the largest objective value $p_{\vec{u}}$;
        \State If $p_{\vec{u}} \geq \eta - \eps/2$, output it; output $\bot$ otherwise.
    \end{algorithmic}
    \end{description}
\end{algorithm}
\end{longfbox}

\begin{claim} \label{claim:fast}
    For some sufficiently large $C > 1$, \cref{alg:inner_loop} requires $N \leq m^{C(B^2 + \log\frac{1}{\eps})}\log\frac{1}{\delta}$ copies of $\rho$, $\poly(m, B, \log\frac{1}{\eps})$ quantum gates per copy of $\rho$, $(\poly(B, b, m)/\toleps)^{B^2/\mu^2 + r + 1}$ calls to the optimization problem \eqref{program:final}, and $(N + (\poly(B, b, m)/\toleps)^{B^2/\mu^2 + r})^C$ additional classical overhead.
\end{claim}
\begin{proof}
    The only step of the algorithm which is quantum is the tomography step, so the quantum complexities follow from \cref{lem:subspace-tomography} with $d = 10B^2 + \log\frac{2}{\appreps}$.

    The number of calls to the optimization problem is at most the number of subsets iterated over times a bound on the size of every net $\net_S$.
    The number of subsets of cardinality at most $B^2/\mu^2$ is at most $(m+1)^{B^2/\mu^2}$.
    By \cref{claim:net}, the cardinality of $\net_S$ is at most $(\frac{1}{c\toleps})^{\abs{S} + r + 1}$ for some $c = \poly(1/B, 1/b, m)$.
    %To bound the size of the net, we first recall that there is an $\zeta$-net in $\ell_2$ for the $\ell_2$ unit ball of $D$ dimensions with size $(\frac{2}{\zeta} + 1)^{D}$ \cite[Corollary 4.2.13]{vershynin18}.
    %We want a net for the set $\calF_S$ with the property that, for every $(\vec{v}, \nu) \in \calF_S$, there is a $(\vec{v}', \nu') \in \net_S$ such that $\twonorm{\vec{v} - \vec{v}'} \leq \toleps$ and $\abs{\nu - \nu'} \leq \toleps$.
    %To get this, we can treat $(\vec{v}, \nu)$ as a vector of dimension $\leq \abs{S} + r + 1$ (where $\abs{S} + r$ is a bound on the dimension of the subspace $\vsub$); a $\toleps$-net over $\calF_S$ in this sense implies desired guarantees.
    %Since the norm of such a vector $(\vec{v}, \nu)$ is bounded by $2B$, we can conclude that a net for $\calF_S$ exists with cardinality at most $(\frac{8}{\toleps} + 1)^{\abs{S} + r + 1}$ \cite[Exercise 4.2.10]{vershynin18}.
    Multiplying the two together, we get $(m+1)^{B^2/\mu^2}(\frac{1}{c\toleps})^{B^2 / \mu^2 + r + 1}$ as desired.

    The running time is dominated by the tomography algorithm and the construction of the nets $\net_S$; both take polynomial time, by \cref{lem:subspace-tomography} and \cref{claim:net}.
\end{proof}

\begin{claim} \label{claim:opt-corollary}
    A $y_{\overline{S}}$ as in \eqref{program:final} can be solved in $(B / \toleps)^{\poly(d, r, 1/\appreps, 1/\mu)}$ time.
    Thus, the classical overhead of \cref{alg:inner_loop} is $m^{\poly(r, B, b, 1/b, 1/\eps)} \poly(\log\frac{1}{\delta})$.
\end{claim}
\begin{proof}
This is a corollary of \cref{thm:opt-main}.
Recall that we are considering the objective function
\begin{align*}
    p_{\vec{z}_S, \nu}(\vec{z}_{\overline{S}}) &= \frac{e^{-\nu^2}}{\prod_{i \in S} (1 + \abs{z_i}^2)} \sum_{x, x' \in \{0,1\}^m} \bra{x} \varrho_d \ket{x'} (\vec{z}^*)^{x}\vec{z}^{x'}. \\
    &= e^{-\nu^2} \sum_{x,x' \in \braces{0,1}^{\abs{\overline{S}}}} \bra{x}\parens[\Big]{(\id_{\overline{S}} \otimes \bra{\pi_{\vec{z}_{S}}}_S) \varrho_d (\id_{\overline{S}} \otimes \ket{\pi_{\vec{z}_{S}}}_S)} \ket{x'} (\vec{z}_{\overline{S}}^*)^x (\vec{z}_{\overline{S}})^{x'}
\end{align*}
Defining $\sigma = (\id_{\overline{S}} \otimes \bra{\pi_{\vec{z}_{S}}}_S) \varrho_d (\id_{\overline{S}} \otimes \ket{\pi_{\vec{z}_{S}}}_S)$, we can express $p_{\vec{z}_S, \nu}(z_{\overline{S}})$ as
\begin{align*}
    p_{\vec{z}_S, \nu}(\vec{z}_{\overline{S}})
    = e^{-\nu^2} (\prod_{i \in \overline{S}} (1 + \abs{z_i}^2)) \bra{\pi_{\vec{z}_{\overline{S}}}} \sigma \ket{\pi_{\vec{z}_{\overline{S}}}},
\end{align*}
We can further write it in terms of tensors $T^{(k)} \in (\C^n)^{\otimes 2k}$.
\begin{align*}
    p_{\vec{z}_S, \nu}(\vec{z}_{\overline{S}})
    &= e^{-\nu^2}\parens[\Big]{T^{(0)} + \angles{T^{(1)}, z_{\overline{S}}^* \otimes z_{\overline{S}}} + \dots + \angles{T^{(d)}, (z_{\overline{S}}^*)^{\otimes d} \otimes z_{\overline{S}}^{\otimes d}}}, \\
    T^{(k)}_{i_1\dots i_k j_1 \dots j_k} &= \begin{cases}
        \frac{1}{(k!)^2}\bra{e_{i_1,\dots,i_k}}\sigma\ket{e_{j_1,\dots,j_k}} & i_1, \dots, i_k \text{ are distinct},\, j_1,\dots,j_k \text{ are distinct} \\
        0 & \text{otherwise}
    \end{cases}
\end{align*}
Note that $\fnorm{T^{(k)}} \leq \frac{1}{(k!)} \fnorm{\sigma}$, since for every entry of $\sigma$ there are $(k!)^2$ corresponding entries in $T^{(k)}$ containing it, scaled down by a factor of $(k!)^2$.
In order to apply \cref{thm:opt-main}, we can further break this up into two cases.
When $\nu \leq 1$, we optimize $\frac{1}{10} p_{\vec{z}_S, \nu}(z_{\overline{S}})$ (that is, the function $f_{\frac{e^{-\nu^2}}{10} T}(\vec{z}_{\overline{S}})$) to $\appreps/10$ error, since with this choice the tensors satisfy
\begin{align*}
    \sum_{k=0}^d \frac{e^{-\nu^2}}{10}\fnorm{T^{(k)}}
    \leq \sum_{k=0}^d \frac{1}{10(k!)}
    \leq 1.
\end{align*}
The domain we want to optimize over is, recalling \eqref{program:final-domain},
\begin{multline*}
    \mathcal{D}_{\toleps} = \braces{
        \vec{z}_{\overline{S}} \in \C^{\abs{\overline{S}}}
        \mid \abs{\twonorm{\vec{z}_{\overline{S}}} - \nu} \leq \toleps,\,
        \twonorm{\Pi_\vsub \vec{z} - \vec{v}} \leq \toleps, \\
        \infnorm{\vec{z}_{\overline{S}}} \leq \mu + \toleps
        \text{ for } \vec{z} = (\vec{v}_S,\,\vec{z}_{\overline{S}})}.
\end{multline*}
This domain is almost of the form needed for \cref{thm:opt-main}, \cref{def:opt-domain}; all we need is to make the following adjustment:
\begin{gather*}
    \twonorm{\Pi_\vsub \vec{z} - \vec{v}} \leq \toleps
    \iff \twonorm{W^\dagger (\vec{z}_{\overline{S}} - \vec{v}_{\overline{S}})} \leq \toleps,
\end{gather*}
where $W \in \C^{\abs{\overline{S}} \times r}$ is a matrix with $\opnorm{W} \leq 1$ such that $(WW^\dagger) \vec{z}_{\overline{S}} = \Pi_\vsub (\vec{0}_S, \vec{z}_{\overline{S}})$.
This is possible since $\Pi_{\vsub}$ is rank at most $r$ over the subspace spanned by the vectors $(\vec{0}_S, \vec{z}_{\overline{S}})$.
Then, \cref{thm:opt-main} outputs a $\vec{y} \in \C^{\abs{\overline{S}}}$ such that $\frac{1}{10} p_{\vec{z}_S, \nu}(\vec{y}) \geq  \max_{\vec{z}_{\overline{S}} \in \calD^{2\toleps}} \frac{1}{10} p_{\vec{z}_S, \nu}(\vec{z}_{\overline{S}}) - \frac{\appreps}{10}$.
This is the desired bound after rescaling.

When $\nu > 1$, we optimize the tensors $e^{-\nu^2}\nu^{2k} T^{(k)}$ with respect to the variables $\vec{t} = \vec{z}_{\overline{S}} / \nu$.
That is, we write
\begin{align*}
    p_{\vec{z}_S, \nu}(\vec{z}_{\overline{S}})
    &= \parens[\Big]{ e^{-\nu^2} T^{(0)} + \dots + \angles{e^{-\nu^2} \nu^{2d} T^{(d)}, (\vec{t}^*)^{\otimes d} \otimes (\vec{t})^{\otimes d}}}
\end{align*}
Then, the sum of the norms of the tensors is
\begin{align*}
    \sum_{k=0}^d e^{-\nu^2} \nu^{2k} \fnorm{T^{(k)}}
    \leq e^{- \nu^2} \sum_{k=0}^d \frac{\nu^{2k}}{k!}
    \leq 1.
\end{align*}
So, if we apply \cref{thm:opt-main} to with error parameter $\appreps$, we are done.
Finally, we can write the domain in terms of the variables $\vec{t}$:
\begin{gather*}
    \abs{\twonorm{\vec{z}_{\overline{S}}} - \nu} \leq \toleps
    \iff \abs{\twonorm{\vec{t}} - 1} \leq \toleps/\nu \\
    \twonorm{\Pi_\vsub \vec{z} - \vec{v}} \leq \toleps
    \iff \twonorm{W^\dagger (\vec{t} - \vec{v}_{\overline{S}} / \nu)} \leq \toleps/\nu
\end{gather*}
Here, we set our tolerance parameter to be $\toleps / \nu$, where $1 \leq \nu \leq B$, giving the desired running time.
\end{proof}

\subsection{Showing correctness}

\begin{lemma}\label{lem:tan-dist-lipchtiz}
For any vectors $\vec{u}, \vec{v}, \vec{a} \in \C^m$ and parameters $B > 1, \eps > 0$, if $\norm{\vec{u}}_2, \norm{\vec{v}}_2, \norm{\vec{a}}_2 \leq B$ and $\dtan(\ket{\pi_{\vec{v}}}, \ket{\pi_{\vec{a}}}) \leq B$ and $\norm{\vec{u} - \vec{v}}_2 \leq \eps $ where $\eps \leq 1/(10mB)^6$ then
\[
\left\lvert \dtan( \ket{\pi_{\vec{a}}}, \ket{\pi_{\vec{v}}}) - \dtan( \ket{\pi_{\vec{a}}}, \ket{\pi_{\vec{u}}}) \right\rvert \leq \eps (10mB)^6
\]
\end{lemma}
\begin{proof}
Recall that   
\[
\dtan( \ket{\pi_{\vec{a}}}, \ket{\pi_{\vec{v}}})^2 = \sum_{i = 1}^m \left\lvert \frac{v_i - a_i}{1 + a_i^* v_i} \right\rvert^2 \,. 
\]
The derivative of $\frac{v_i - a_i}{1 + a_i^* v_i} $ with respect to $v_i$ is 
\[
\frac{1 + a_i^* v_i - a_i^*(v_i - a_i)}{(1 + a_i^* v_i)^2} \,.
\]
Note that whenever $|1 + a_i^* v_i| \leq 1/2$ then $|v_i - a_i| \geq 0.1$.  Thus, the magnitude of the derivative is at most $100B^2\left( 1 + \left\lvert \frac{v_i - a_i}{1 + a_i^* v_i}  \right \rvert^2 \right) $.  Thus, if we define the vectors 
\[
\vec{V} =  \left\{ \frac{v_i - a_i}{1 + a_i^* v_i} \right\}_{i \in [m]} \; , \; \vec{U} = \left\{ \frac{u_i - a_i}{1 + a_i^* u_i} \right\}_{i \in [m]} 
\]
then integrating the above and using the assumption that $\norm{\vec{u} - \vec{v}}_2 \leq \eps$ implies
\[
\norm{\vec{V} - \vec{U}}_2 \leq 200B^3m \eps \,.
\]
From this, we immediately get the desired inequality.
\end{proof}

%The feasible set of \eqref{program:subspace} is contained in the feasible set of \eqref{program:guess}, and contains the feasible set of \eqref{program:guess} where $b$ is replaced by $1.5 b$ for all $s \in [r]$.

\begin{corollary}\label{coro:lipchitz-constraint}
    For two vectors $\vec{u}, \vec{v} \in \C^m$ such that $\norm{\vec{u} - \vec{v}} < \poly(1/B, 1/m, 1/b)$ and $\norm{\vec{u}}, \norm{\vec{v}} \leq B$, along with an arbitrary $\vec{a}$ and $\nu, b \in \R$, we have that 
    \begin{gather*}
        \nu^2 - \twonorm{\vec{v}_{\overline{S}}}^2 + \twonorm{\vec{v}_{\overline{S}} - \vec{a}_{\overline{S}}}^2 \geq 1.5b^2 - \dtan(\ket{\pi_{\vec{v}_S}}, \ket{\pi_{\vec{a}_S}})^2 \\
        \implies \nu^2 - \twonorm{\vec{u}_{\overline{S}}}^2 + \twonorm{\vec{u}_{\overline{S}} - \vec{a}_{\overline{S}}}^2 \geq 1.49b^2 - \dtan(\ket{\pi_{\vec{u}_S}}, \ket{\pi_{\vec{a}_S}})^2.
    \end{gather*}
\end{corollary}
\begin{proof}
Note that when $\dtan(\ket{\pi_{\vec{v}_S}}, \ket{\pi_{\vec{a}_S}}) \geq 2(b + B)$ then the inequalities are both trivially true.  Also, if $\twonorm{\vec{v}_{\overline{S}} - \vec{a}_{\overline{S}}} \geq 2(B + b)$ then the inequalities are both trivially true. 

Otherwise, we can just apply Lemma~\ref{lem:tan-dist-lipchtiz} to bound the difference between $\dtan(\ket{\pi_{\vec{u}_S}}, \ket{\pi_{\vec{a}_S}})$ and $\dtan(\ket{\pi_{\vec{v}_S}}, \ket{\pi_{\vec{a}_S}})$ and also directly bound the differences $ \twonorm{\vec{v}_{\overline{S}}}^2 - \twonorm{\vec{u}_{\overline{S}}}^2$ and $\twonorm{\vec{v}_{\overline{S}} - \vec{a}_{\overline{S}}}^2 - \twonorm{\vec{u}_{\overline{S}} - \vec{a}_{\overline{S}}}^2$ to get the desired conclusion.    
\end{proof}

\begin{claim} \label{claim:net}
    There is a $\net_S$ which satisfies the following properties.
    First, $\net_S$ satisfies that, for every $(\vec{v}, \nu) \in \net_S$, the following conditions hold: $\vec{v} \in V_S$; $\norm{\vec{v}_S}^2 + \nu^2 \leq B^2$; $\norm{\vec{v}} \leq B$; and
    \[
        \nu^2 - \twonorm{\vec{v}_{\overline{S}}}^2 + \twonorm{\vec{v}_{\overline{S}} - \vec{a}_{\overline{S}}^{(s)}}^2 \geq 1.49b^2 - \dtan(\ket{\pi_{\vec{v}_S}}, \ket{\pi_{\vec{a}^{(s)}_S}})^2.
    \]
    Second, for any $(\vec{v}, \nu) \in \mathcal{F}_S$, there is a $(\vec{v}', \nu')$ such that $\norm{\vec{v} - \vec{v}'} \leq \toleps$ and $\abs{\nu - \nu'} \leq \toleps$.
    This net has at most $(\frac{1}{c\toleps})^{10(\abs{S} + r + 1)}$ elements and can be constructed in $(\frac{1}{c\toleps})^{10(\abs{S} + r + 1)}$ time, for some $c = \poly(1/B, 1/b, 1/m)$.
\end{claim}
\begin{proof}
Let $\calA_S$ be the set   
\[
\calA_S = \Big\{\vec{v} \in \vsub,\, \nu \in [0, B] \,\Big|\, \twonorm{\vec{v}_S}^2 + \nu^2 \leq B^2,\, \twonorm{\vec{v}} \leq B     \Big\}
\]
i.e. it is $\calF_S$ with the tangent distance constraint removed.  Now we can construct a net $\calB_S$ for the set $\calA_S$ such that 
\begin{itemize}
\item For all $(\vec{v}, \nu) \in \calF_S$, there is $(\vec{v}', \nu') \in \calB_S$ such that $\norm{\vec{v} - \vec{v}'} \leq c\toleps$ and $\abs{\nu - \nu'} \leq c\toleps$
\item All elements of $\calB_S$ are in $\calA_S$
\item The net $\calB_S$ can be enumerated in $(\frac{1}{c\toleps})^{10(\abs{S} + r + 1)}$ time
\end{itemize}
We can construct $\calB_S$  by simply enumerating over a sufficiently fine grid for $\vec{v}$ and $\nu$.  Now given $\calB_S$, we construct $\calN_S$ by simply removing all points $(\vec{v}', \nu')$ such that
\[
\nu'^2 - \twonorm{\vec{v'}_{\overline{S}}}^2 + \twonorm{\vec{v'}_{\overline{S}} - \vec{a}_{\overline{S}}^{(s)}}^2 < 1.49b^2 - \dtan(\ket{\pi_{\vec{v'}_S}}, \ket{\pi_{\vec{a}^{(s)}_S}})^2 \,.
\]
By Corollary~\ref{coro:lipchitz-constraint}, this removal never removes any $(\vec{v}', \nu')$ such that $\norm{\vec{v} - \vec{v}'} \leq \toleps$ and $\abs{\nu - \nu'} \leq \toleps$ for some $(\vec{v}, \nu) \in \calF_S$.  Thus, $\calN_S$ still covers all points in $\calF_S$, as desired.
\end{proof}
%\ewin{TODO}

\begin{lemma}[Approximating tangent distance constraints] \label{claim:guess-to-subspace}
    For $\vec{z}, \vec{a} \in \C^m$, $b > 0$, and $\abs{S} \subseteq [m]$, if $\infnorm{\vec{z}_{\overline{S}}} \leq \frac{1}{6} \min(\frac{1}{b}, b)$, then the following implications hold.
    \begin{gather*}
        \dtan(\ket{\pi_{\vec{z}}}, \ket{\pi_{\vec{a}}}) \geq 1.5 b
        \implies 
        \dtan(\ket{\pi_{\vec{z}_S}}, \ket{\pi_{\vec{a}_S}})^2 + \twonorm{\vec{z}_{\overline{S}} - \vec{a}_{\overline{S}}}^2 \geq 1.5 b^2 \\
        \dtan(\ket{\pi_{\vec{z}_S}}, \ket{\pi_{\vec{a}_S}})^2 + \twonorm{\vec{z}_{\overline{S}} - \vec{a}_{\overline{S}}}^2 \geq (1.4 - \infnorm{\vec{a}_{\overline{S}}})b^2
        \implies
        \dtan(\ket{\pi_{\vec{z}}}, \ket{\pi_{\vec{a}}}) \geq b
    \end{gather*}
\end{lemma}
\begin{proof}
    The main idea is that we can relate $\dtan(\ket{\pi_{\vec{z}}}, \ket{\pi_{\vec{a}}})^2$ to $\dtan(\ket{\pi_{\vec{z}_S}}, \ket{\pi_{\vec{a}_S}})^2 + \twonorm{\vec{z}_{\overline{S}} - \vec{a}_{\overline{S}}}^2$ up to some small constant multiplicative error.
    Let $\mu = \frac{1}{6}\min(\frac{1}{b}, b)$, so that $\infnorm{\vec{z}_{\overline{S}}} \leq \mu$.

    We consider two cases.
    First, suppose $\abs{a_i} \leq \frac{1}{2\mu}$ for all $i \not\in S$.
    For such constraints, we have the following bound on the difference between the two constraints:
    \begin{align*}
        & \abs[\Big]{\dtan(\ket{\pi_{\vec{z}}}, \ket{\pi_{\vec{a}}})^2
        - (\dtan(\ket{\pi_{\vec{z}_S}}, \ket{\pi_{\vec{a}_S}})^2 + \twonorm{\vec{z}_{\overline{S}} - \vec{a}_{\overline{S}}}^2)} \\
        &= \abs[\Big]{\dtan(\ket{\pi_{\vec{z}_{\overline{S}}}}, \ket{\pi_{\vec{a}_{\overline{S}}}})^2
        - \twonorm{\vec{z}_{\overline{S}} - \vec{a}_{\overline{S}}}^2} \\
        &\leq \dtan(\ket{\pi_{\vec{z}_{\overline{S}}}}, \ket{\pi_{\vec{a}_{\overline{S}}}})^2(\max_{i \in \overline{S}} \abs{z_i}\abs{a_i})^2 \\
        &\leq \frac{1}{4}\dtan(\ket{\pi_{\vec{z}}}, \ket{\pi_{\vec{a}}})^2
    \end{align*}
    where the first inequality follows from \cref{lem:dtan-l2}, the second inequality uses that $\abs{z_i} \leq \mu$ for all $i \not \in S$ and by assumption $\abs{a_i^{(s)}} \leq 1/(2 \mu)$ for all $i \not\in S$.
    So, we can conclude that
    \begin{gather*}
        \dtan(\ket{\pi_{\vec{z}}}, \ket{\pi_{\vec{a}}})^2 \geq (1.5 b)^2
        \implies
        \dtan(\ket{\pi_{\vec{z}_S}}, \ket{\pi_{\vec{a}_S}})^2 + \twonorm{\vec{z}_{\overline{S}} - \vec{a}_{\overline{S}}}^2 \geq \frac{3}{4}(1.5 b)^2 \geq 1.5b^2 \,,\\
        \dtan(\ket{\pi_{\vec{z}_S}}, \ket{\pi_{\vec{a}_S}})^2 + \twonorm{\vec{z}_{\overline{S}} - \vec{a}_{\overline{S}}}^2 \geq (1.4 - \infnorm{\vec{a}_{\overline{S}}})b^2
        \implies
        \dtan(\ket{\pi_{\vec{z}}}, \ket{\pi_{\vec{a}}})^2 \geq \frac{4}{5} \parens[\big]{1.4 - \infnorm{\vec{a}_{\overline{S}}}}b^2,
    \end{gather*}
    where the last line gives the desired $\geq b^2$ bound by our case assumption that $\infnorm{\vec{a}_{\overline{S}}} \leq \frac{1}{2\mu} \leq 0.1$.

    For the other case, suppose $\abs{a_i} > \frac{1}{2\mu}$ for some $i \not\in S$.
    We then argue that all inequalities are true, so the implications hold trivially.
    \begin{align*}
        \dtan(\ket{\pi_{\vec{z}}}, \ket{\pi_{\vec{a}}})
        \geq \abs[\Big]{\frac{z_i - a_i}{1 + z_i^* a_i}}
        \geq \abs[\Big]{\frac{\mu - \frac{1}{2\mu}}{1 + \frac12}}
        \geq \frac{1}{3}\abs[\Big]{\mu - \frac{1}{\mu}}
        \geq 1.5 b.
    \end{align*}
    Similarly, using the same bounds, we have that
    \begin{align*}
        \dtan(\ket{\pi_{\vec{z}_S}}, \ket{\pi_{\vec{a}_S}})^2 + \twonorm{\vec{z}_{\overline{S}} - \vec{a}_{\overline{S}}}^2
        \geq \abs{z_i - a_i}^2
        \geq \parens[\Big]{\frac{1}{2\mu}- \mu}^2
        \geq \frac14 \parens[\Big]{\frac{1}{\mu}- \mu}^2
        \geq 1.5 b^2.
    \end{align*}
    Thus, all of the stated inequalities are true.
\end{proof}

%The objective functions of \cref{program:function} and \cref{program:subspace} are $3\appreps$-close in the feasible set.

\begin{lemma} \label{claim:subspace-to-function}
    Consider a vector $\vec{z} \in \C^m$ and set $S \subseteq [m]$ such that $\infnorm{\vec{z}_{\overline{S}}} \leq \frac{\sqrt{\appreps}}{\twonorm{\vec{z}}}$.
    Further, for an $\appreps > 0$, let $d \geq 8\twonorm{\vec{z}}^2 + \log\frac{2}{\appreps}$, and let $\varrho_d$ satisfy the guarantees of the output of \cref{alg:inner_loop}: $\opnorm{\varrho_d - \Pi_{\leq d} \Pi_{\leq d}} \leq \appreps$ for a PSD matrix $\varrho_d$ with trace at most 1 and supported only on the image of $\Pi_{\leq d}$.
    Then
    \begin{align*}
        \abs{\bra{\pi_{\vec{z}}} \rho \ket{\pi_{\vec{z}}} &- p_{\vec{z}_S}(\vec{z}_{\overline{S}})} \leq 3\appreps, \\
        &\text{where } p_{\vec{z}_S}(\vec{z}_{\overline{S}}) = \frac{e^{-\twonorm{\vec{z}_{\overline{S}}}^2}}{\prod_{i \in S} (1 + \abs{z_i}^2)} \sum_{x, x' \in \{0,1\}^m} \bra{x} \varrho_{d} \ket{x'} (\vec{z}^*)^{x}(\vec{z})^{x'}.
    \end{align*}
\end{lemma}
\begin{proof}
Fix a vector $\Vec{z}$.
Recall that $\rho_d = \Pi_{\leq d} \rho \Pi_{\leq d}$ is the unknown state truncated to strings of Hamming weight at most $d$, and $\varrho_d$ is our estimate of $\rho_d$.
We first show the following:
\begin{equation*}
    \abs{\bra{\pi_{\vec{z}}} \rho \ket{\pi_{\vec{z}}} - \bra{\pi_{\vec{z}}} \rho_{d} \ket{\pi_{\vec{z}}}} \leq \appreps\,. 
\end{equation*}
To see this, we apply \cref{lem:low_weight} to $\Vec{z}$ to get the following:
\begin{align*}
    \abs{\bra{\pi_{\vec{z}}} \rho \ket{\pi_{\vec{z}}} - \bra{\pi_{\vec{z}}} \rho_d \ket{\pi_{\vec{z}}}} 
    &= \abs[\Big]{\Tr\parens[\Big]{\rho(\proj{\pi_{\vec{z}}} - \Pi_{<d} \proj{\pi_{\vec{z}}} \Pi_{<d})}} \\
    &\leq \opnorm{\proj{\pi_{\vec{z}}} - \Pi_{<d} \proj{\pi_{\vec{z}}} \Pi_{<d}} \\
    &\leq \opnorm{\proj{\pi_{\vec{z}}} - \Pi_{<d} \proj{\pi_{\vec{z}}}} + \opnorm{\Pi_{<d} \proj{\pi_{\vec{z}}} - \Pi_{<d} \proj{\pi_{\vec{z}}} \Pi_{<d}} \\
    &\leq 2\twonorm{\ket{\pi_{\vec{z}}} - \Pi_{<d} \ket{\pi_{\vec{z}}}} \\
    &= 2\twonorm{\Pi_{\geq d} \ket{\pi_{\vec{z}}}} \\
    &\leq 2e^{-d(\log(d/\twonorm{\vec{z}}^2) - 1)}
\end{align*}
Since we take $d \geq 8\twonorm{\vec{z}}^2 + \log\frac{2}{\appreps}$, this makes the final quantity smaller than $\appreps$.
Next, because $\opnorm{\varrho_d - \rho_d} \leq \appreps$, by the definition of operator norm we have
\begin{align*}
    \abs{\bra{\pi_{\vec{z}}} \rho_d \ket{\pi_{\vec{z}}} - \bra{\pi_{\vec{z}}} \varrho_d \ket{\pi_{\vec{z}}}} &\leq \appreps\,.
\end{align*}
We can relate this quantity to $p_{\vec{z}_S}(\vec{z}_{\overline{S}})$ as follows:
\begin{align*}
    p_{\vec{z}_S}(\vec{z}_{\overline{S}}) &= \frac{e^{-\twonorm{\vec{z}_{\overline{S}}}^2}}{\prod_{i \in S} (1 + \abs{z_i}^2)} \sum_{x, x' \in \{0,1\}^m} \bra{x} \varrho_d \ket{x'} (\vec{z}^*)^{x}\vec{z}^{x'} \\
    &= \frac{e^{-\twonorm{\vec{z}_{\overline{S}}}^2}\prod_{i \in [m]} (1 + \abs{z_i})^2}{\prod_{i \in S} (1 + \abs{z_i}^2)}\bra{\pi_{\vec{z}}}\varrho_d \ket{\pi_{\vec{z}}} \\
    &= \parens[\big]{e^{-\twonorm{\vec{z}_{\overline{S}}}^2}\prod_{i \not\in S} (1 + \abs{z_i})^2}\bra{\pi_{\vec{z}}}\varrho_d  \ket{\pi_{\vec{z}}}\,.
\end{align*}
By \cref{lem:dtan-fidelity-approx}, we have:
\begin{equation*}
    e^{-\sum_{i \not\in S} \abs{z_i}^4} \leq e^{-\twonorm{\vec{z}_{\overline{S}}}^2}\prod_{i \not\in S} (1 + \abs{z_i})^2 \leq 1\,,
\end{equation*}
where we also know that 
\begin{align*}
    \sum_{i \not\in S} \abs{z_i}^4 \leq \twonorm{\vec{z}_{\overline{S}}}^2 \infnorm{\vec{z}_{\overline{S}}}^2
    \leq \appreps.
\end{align*}  
Since $1 - x \leq e^{-x}$, we have the following multiplicative error bounds on $p_{\vec{z}_{S}}(\vec{z}_{\overline{S}})$:
\begin{align*}
    (1 - \appreps) \bra{\pi_{\vec{z}}}\varrho_d \ket{\pi_{\vec{z}}} \leq p_{\vec{z}_{S}}(\vec{z}_{\overline{S}}) \leq \bra{\pi_{\vec{z}}}\varrho_d \ket{\pi_{\vec{z}}}\,.
\end{align*}
As $\varrho_d$ is a sub-normalized quantum state, $\appreps \bra{\pi_{\vec{z}}}\varrho_d \ket{\pi_{\vec{z}}} \leq \appreps$, so $\abs{\bra{\pi_{\vec{z}}}\varrho_d \ket{\pi_{\vec{z}}} - p_{\vec{z}_{S}}(\vec{z}_{\overline{S}})} \leq \appreps$.
Applying the triangle inequality, we get the desired bound:
\begin{multline*}
    \abs{\bra{\pi_{\vec{z}}}\rho\ket{\pi_{\vec{z}}} - p_{\vec{z}_S}(\vec{z}_{\overline{S}})}
    \leq \abs{\bra{\pi_{\vec{z}}}\rho\ket{\pi_{\vec{z}}} - \bra{\pi_{\vec{z}}}\rho_d\ket{\pi_{\vec{z}}}} \\
    + \abs{\bra{\pi_{\vec{z}}}\rho_d\ket{\pi_{\vec{z}}} - \bra{\pi_{\vec{z}}}\varrho_d\ket{\pi_{\vec{z}}}} + \abs{\bra{\pi_{\vec{z}}}\varrho_d\ket{\pi_{\vec{z}}} - p_{\vec{z}_S}(\vec{z}_{\overline{S}})}
    \leq 3\appreps \qedhere
\end{multline*}
\end{proof}

We now prove correctness.
We split the desired guarantee into two parts: soundness and completeness.
First, we prove soundness.

\begin{claim}[Soundness] \label{claim:sound}
    The output of \cref{alg:inner_loop} satisfies the desired correctness criteria:
    the output is either $\bot$ or a $\vec{z} \in \C^m$ such that
    \begin{enumerate}[label=(\alph*)]
        \item $\bra{\pi_{\vec{z}}}\rho\ket{\pi_{\vec{z}}} \geq \eta - \eps$;
        \item For all $s \in [r]$, $\dtan(\ket{\pi_{\vec{z}}}, \ket{\pi_{\vec{a}^{(s)}}}) \geq b$.
    \end{enumerate}
\end{claim}
\begin{proof}
    From inspecting \cref{alg:inner_loop}, we can observe that the output $\vec{u}$ satisfies the following guarantees, for some internal parameters $S$, $\vec{v} \in \vsub$, and $\nu \in [0, B]$.
    First, it satisfies the guarantees from being in the domain $\mathcal{D}^{2\toleps}$:
    \begin{align} \label{eq:soundness-domain}
        \abs{\twonorm{\vec{u}_{\overline{S}}} - \nu} \leq 2\toleps \qquad
        \twonorm{\Pi_\vsub \vec{u} - \vec{v}} \leq 2\toleps \qquad
        \infnorm{\vec{u}_{\overline{S}}} \leq \mu + 2\toleps
    \end{align}
    Since we have that $\toleps \leq 0.01 B$ and $\nu \leq B$, this implies that $\abs{\twonorm{\vec{u}_{\overline{S}}}^2 - \nu^2} \leq 3B \toleps$.
    The output $\vec{u}$ also satisfies the tangent distance guarantees inherited from the $\toleps$-net, the bound $\twonorm{\vec{u}_S}^2 + \nu^2 \leq B^2$ also from the net, and the guarantee from the objective function:
    \begin{gather}
        \nu^2 - \twonorm{\vec{v}_{\overline{S}}}^2 + \twonorm{\vec{v}_{\overline{S}} - \vec{a}_{\overline{S}}^{(s)}}^2 \geq 1.49 b^2 - \dtan(\ket{\pi_{\vec{u}_S}}, \ket{\pi_{\vec{a}^{(s)}_S}})^2 \text{ for all } s \in [r] \label{eq:soundness-tangent-guarantee}\\
        p_{\vec{u}_S, \nu}(\vec{u}_{\overline{S}}) \geq \eta - \eps / 2
        \label{eq:soundness-fidelity}
    \end{gather}
    From \eqref{eq:soundness-fidelity}, we can conclude (a).
    First, we relate $p_{\vec{u}_S, \nu}(\vec{u}_{\overline{S}}) = e^{\nu^2 - \twonorm{\vec{u}_{\overline{S}}}^2} p_{\vec{u}_S}(\vec{u}_{\overline{S}})$ to $\bra{\pi_{\vec{u}}} \rho \ket{\pi_{\vec{u}}}$:
    \begin{align*}
        \abs{\bra{\pi_{\vec{u}}} \rho \ket{\pi_{\vec{u}}} - p_{\vec{u}_S, \nu}(\vec{u}_{\overline{S}})}
        &\leq \abs{\bra{\pi_{\vec{u}}} \rho \ket{\pi_{\vec{u}}} - p_{\vec{u}_S}(\vec{u}_{\overline{S}})} + \abs{p_{\vec{u}_S}(\vec{u}_{\overline{S}}) - p_{\vec{u}_S, \nu}(\vec{u}_{\overline{S}})} \\
        &\leq 3\appreps + \abs{1 - e^{\nu^2 - \twonorm{\vec{u}_{\overline{S}}}^2}}\abs{p_{\vec{u}_S}(\vec{u}_{\overline{S}})} \\
        &\leq 3\appreps + \abs{1 - e^{\nu^2 - \twonorm{\vec{u}_{\overline{S}}}^2}}(1 + 3\appreps)\\
        &\leq 3\appreps + 0.1 \eps.
    \end{align*}
    Above, we use that $\toleps \leq 0.01 \eps / B$ and \cref{claim:subspace-to-function}; we satisfy the assumptions of the lemma since $\twonorm{\vec{u}} \leq B + 2\toleps \leq 1.1 B$ and $\infnorm{\vec{u}_{\overline{S}}} \leq 1.1 \mu \leq \frac{\sqrt{\appreps}}{2B} \leq \frac{\sqrt{\appreps}}{\twonorm{\vec{u}}}$.
    \begin{align*}
        \bra{\pi_{\vec{u}}} \rho \ket{\pi_{\vec{u}}}
        \geq (\eta - \eps / 2) + (\bra{\pi_{\vec{u}}} \rho \ket{\pi_{\vec{u}}} - p_{\vec{u}_S}(\vec{u}_{\overline{S}}))
        \geq \eta - \eps/2 - 3\appreps - 0.1\eps \geq \eta - \eps.
    \end{align*}
    To get (b), we work with the tangent distance constraints \eqref{eq:soundness-tangent-guarantee}.
    Along with the domain constraints \eqref{eq:soundness-domain}, these imply the following:
    \begin{align} \label{eq:soundness-tangent}
        \dtan(\ket{\pi_{\vec{u}_S}}, \ket{\pi_{\vec{a}^{(s)}_S}})^2
        + \twonorm{\vec{u}_{\overline{S}} - \vec{a}_{\overline{S}}^{(s)}}^2
        \geq 1.49 b^2 - 20\toleps(B + \twonorm{\vec{a}_{\overline{S}}^{(s)}}) \text{ for all } s \in [r]
    \end{align}
    This follows from the below argument, where we use triangle inequality, $\toleps \leq 0.01$, and the Pythagorean theorem.
    Since $\vsub$ contains the subspace corresponding to $S$, there is a corresponding projector $\Pi$ such that $\Pi \vec{x}_{\overline{S}} = (\Pi_\vsub \vec{x})_{\overline{S}}$ (so, in particular, $\twonorm{\Pi \vec{u}_{\overline{S}} - \vec{v}_{\overline{S}}} = \twonorm{(\Pi_\vsub \vec{u} - \vec{v})_{\overline{S}}} \leq 2\toleps$ and $\Pi \vec{a}_{\overline{S}}^{(s)} = \vec{a}_{\overline{S}}^{(s)}$ for all $s \in [r]$).
    \begin{align*}
        & \abs{\nu^2 - \twonorm{\vec{v}_{\overline{S}}}^2 + \twonorm{\vec{v}_{\overline{S}} - \vec{a}_{\overline{S}}^{(s)}}^2 - \twonorm{\vec{u}_{\overline{S}} - \vec{a}_{\overline{S}}^{(s)}}^2} \\
        &\leq 3B\toleps + \abs{\twonorm{\vec{u}_{\overline{S}}}^2 - \twonorm{\vec{v}_{\overline{S}}}^2 + \twonorm{\vec{v}_{\overline{S}} - \vec{a}_{\overline{S}}^{(s)}}^2 - \twonorm{\vec{u}_{\overline{S}} - \vec{a}_{\overline{S}}^{(s)}}^2} \\
        &= 3B\toleps + \abs{(\twonorm{\Pi \vec{u}_{\overline{S}}}^2 - \twonorm{\vec{v}_{\overline{S}}}^2) + (\twonorm{\vec{v}_{\overline{S}} - \vec{a}_{\overline{S}}^{(s)}}^2 - \twonorm{\Pi(\vec{u}_{\overline{S}} - \vec{a}_{\overline{S}}^{(s)})}^2)} \\
        &\leq 3B\toleps + \twonorm{\Pi \vec{u}_{\overline{S}} - \vec{v}_{\overline{S}}}(\twonorm{\Pi \vec{u}_{\overline{S}}} + \twonorm{\vec{v}_{\overline{S}}} + \twonorm{\vec{v}_{\overline{S}} - \vec{a}_{\overline{S}}^{(s)}} + \twonorm{\Pi \vec{u}_{\overline{S}} - \vec{a}_{\overline{S}}^{(s)}}) \\
        &\leq 20\toleps(B + \twonorm{\vec{a}_{\overline{S}}^{(s)}})
    \end{align*}
    Notice that we have no bound on the size of $\vec{a}^{(S)}$, so we cannot remove this dependence.
    Next, we use that $\toleps$ is sufficiently small (specifically, smaller than $\frac{0.001}{(1 + B)\sqrt{m}} b^2$) to conclude from \eqref{eq:soundness-tangent} that
    \begin{align*}
        \dtan(\ket{\pi_{\vec{u}_S}}, \ket{\pi_{\vec{a}^{(s)}_S}})^2
        + \twonorm{\vec{u}_{\overline{S}} - \vec{a}_{\overline{S}}^{(s)}}^2
        \geq (1.4 - \infnorm{\vec{a}_{\overline{S}}})b^2 \text{ for all } s \in [r].
    \end{align*}
    Then, we can appeal to \cref{claim:guess-to-subspace}, since $\infnorm{\vec{u}_{\overline{S}}} \leq \mu + 2\toleps \leq \frac{1}{6} \min(b, \frac{1}{b})$ for every $s \in [r]$, to get that (b) is satisfied:
    \begin{equation*}
        \dtan(\ket{\pi_{\vec{u}}}, \ket{\pi_{\vec{a}^{(s)}}})
        \geq b \text{ for all } s \in [r]. \qedhere
    \end{equation*}
\end{proof}

\begin{claim}[Completeness] \label{claim:complete}
    If there is a product state $\ket{\pi}$ such that
    \begin{enumerate}[label=(\alph*')]
        \item $\bra{\pi}\rho\ket{\pi} \geq \eta$;
        \item For all $s \in [r]$, $\dtan(\ket{\pi_{\vec{z}}}, \ket{\pi_{\vec{a}^{(s)}}}) \geq 1.5b$;
        \item $\dtan(\ket{\pi}, \ket{\varphi}) \leq B$;
    \end{enumerate}
    then the output of \cref{alg:inner_loop} is guaranteed to not be $\bot$.
\end{claim}
\begin{proof}
    We work in the basis where $\ket{\varphi}$ is rotated to $\ket{0^m}$ via single-qubit unitaries.
    Let $\ket{\pi_{\vec{u}}}$ be a product state satisfying (a'), (b'), and (c').
    We will show that \cref{alg:inner_loop} will, in its search, find some $\vec{\wt{u}}$ close to $\vec{u}$ and which achieves a similarly large objective value.
    Thus, the algorithm will not output $\bot$.

    First, let $S = \{i \mid \abs{u_i} \geq \mu\}$.
    Using (c'), $\dtan(\ket{\pi_{\vec{u}}}, \ket{0^m}) = \twonorm{\vec{u}} \leq B$, so $\twonorm{u}^2 \leq \sum_{i \in S} u_i^2/\mu^2 \leq B^2/\mu^2$.
    So, we can consider $\calF_S$ for the corresponding choice of $S$.
    Our vector $\vec{u}$ has a corresponding point in the ``feasible set'', $(\Pi_\vsub \vec{u}, \twonorm{\vec{u}_{\overline{S}}}) \in \calF_S$.
    The key constraint to check is the tangent distance constraint in the definition of $\calF_S$ \eqref{eq:final-feasible}, which becomes
    \begin{gather*}
        \twonorm{\vec{u}_{\overline{S}}}^2 - \twonorm{(\Pi_\vsub \vec{u})_{\overline{S}}}^2 + \twonorm{(\Pi_\vsub \vec{u})_{\overline{S}} - \vec{a}_{\overline{S}}^{(s)}}^2 \geq 1.5b^2 - \dtan(\ket{\pi_{\vec{u}_S}}, \ket{\pi_{\vec{a}^{(s)}_S}})^2 \\
        \iff \twonorm{\vec{u}_{\overline{S}}}^2 - \twonorm{\Pi \vec{u}_{\overline{S}}}^2 + \twonorm{\Pi(\vec{u}_{\overline{S}} - \vec{a}_{\overline{S}}^{(s)})}^2 \geq 1.5b^2 - \dtan(\ket{\pi_{\vec{u}_S}}, \ket{\pi_{\vec{a}^{(s)}_S}})^2 \\
        \iff \dtan(\ket{\pi_{\vec{u}_S}}, \ket{\pi_{\vec{a}^{(s)}_S}})^2 + \twonorm{\vec{u}_{\overline{S}} - \vec{a}_{\overline{S}}^{(s)}}^2 \geq 1.5b^2,
    \end{gather*}
    where as in \cref{claim:sound} we define $\Pi$ to be the projector such that $\Pi \vec{z}_{\overline{S}} = (\Pi_\vsub \vec{z})_{\overline{S}}$.
    By \cref{claim:guess-to-subspace}, (b') implies the above condition for all $s \in [r]$.
    So, by \cref{claim:net}, we can find a point $(\vec{v}, \nu) \in \net_S$ such that
    \begin{align} \label{eq:completeness-net}
        \twonorm{\vec{v} - \Pi_\vsub \vec{u}} \leq \toleps\,, \qquad
        \abs{\twonorm{\vec{u}_{\overline{S}}} - \nu} \leq \toleps\,.
    \end{align}
    We now claim that the vector $\vec{u}_{\overline{S}} \in \mathcal{D}_{\toleps}$ as defined in \eqref{program:final-domain}.
    Thus, the vector $\vec{z} = (\vec{v}_S, \vec{u}_{\overline{S}})$ is a feasible solution to the optimization problem in \eqref{program:final}.
    To show this, observe that
    \begin{align*}
        \twonorm{\vec{z} - \vec{u}} = \twonorm{\vec{v}_S - \vec{u}_S} \leq \toleps
    \end{align*}
    Here, we use the definition of the two norm squared as the sum over entries; that off of $S$, $\vec{z}$ and $\vec{u}$ are identical; and \eqref{eq:completeness-net}.
    From this, we can check that for the first constraint of $\mathcal{D}_{\toleps}$, we have that
    \begin{align*}
        \abs{\twonorm{\vec{z}_{\overline{S}}} - \nu} \leq \toleps \,.
    \end{align*}
    The second and third constraints of $\mathcal{D}_{\toleps}$ can also easily be verified:
    \begin{gather*}
        \twonorm{\Pi_\vsub \vec{z} - \vec{v}}
        = \twonorm{\Pi_\vsub (\vec{0}_S, (\vec{u}-\vec{v})_{\overline{S}})}
        \leq \twonorm{\Pi_\vsub (\vec{u}-\vec{v})}
        \leq \toleps \\
        \infnorm{\vec{z}_{\overline{S}}} = \infnorm{\vec{u}_{\overline{S}}} \leq \mu \leq \mu + \toleps.
    \end{gather*}
    In the first inequality in the first line line, we use that the subspace $\vsub$ contains the subspace spanned by $S$, so the projector onto $\overline{S}$ commutes with $\Pi_\vsub$.
    To finish the completeness proof, we now need to show that, for the above choice of $\vec{z}$, $p_{\vec{z}_S, \nu}(\vec{z}_{\overline{S}}) \geq \eta - \eps/2 + \appreps$.
    Then, the best candidate solution $\vec{y}$ found by the algorithm must have an objective value of at least $\eta - \eps/2$, and thus the algorithm will not output $\bot$.
    Because $\vec{z}$ and $\vec{u}$ satisfy $\twonorm{\vec{z} - \vec{u}} \leq \toleps$, we have the following bound:
    \begin{align*}
        \left| \bra{\pi_{\vec{z}}} \rho \ket{\pi_{\vec{z}}} - \bra{\pi_{\vec{u}}} \rho \ket{\pi_{\vec{u}}}\right|
        &\leq \opnorm{\proj{\pi_{\vec{z}}} - \proj{\pi_{\vec{u}}}} \\
        &\leq \dtan(\ket{\pi_{\vec{z}}}, \ket{\pi_{\vec{u}}})\\
        &= \dtan(\ket{\pi_{\vec{z}_{\overline{S}}}}, \ket{\pi_{\vec{u}_{\overline{S}}}})\\
        &\leq \twonorm{\vec{z}_{\overline{S}} - \vec{u}_{\overline{S}}} (1 + \max_{i \in \overline{S}}(\abs{z_i}\abs{u_i}))\\
        &\leq \toleps (1 + \mu(\mu + \toleps))\\
        &\leq 2\toleps
    \end{align*}
    In the first line, we use the definition of the trace distance.  Then we use \cref{cor:dtan_trace_relationship}, and then the definition of tangent distance to restrict to the coordinates where $\vec{z}$ and $\vec{u}$ are not equal.  Then, we use \cref{lem:dtan-l2} and then our bound $\infnorm{u_{\overline{S}}} \leq \mu$.
    We now relate the fidelity to the objective function for $\vec{z}$ through \cref{claim:subspace-to-function}, similarly to as in \cref{claim:sound}.
    Recall that $p_{\vec{u}_S, \nu}(\vec{u}_{\overline{S}}) = e^{\nu^2 - \twonorm{\vec{u}_{\overline{S}}}^2} p_{\vec{u}_S}(\vec{u}_{\overline{S}})$, so
    \begin{align*}
        \abs{\bra{\pi_{\vec{u}}} \rho \ket{\pi_{\vec{u}}} - p_{\vec{u}_S, \nu}(\vec{u}_{\overline{S}})}
        &\leq \abs{\bra{\pi_{\vec{u}}} \rho \ket{\pi_{\vec{u}}} - p_{\vec{u}_S}(\vec{u}_{\overline{S}})} + \abs{p_{\vec{u}_S}(\vec{u}_{\overline{S}}) - p_{\vec{u}_S, \nu}(\vec{u}_{\overline{S}})} \\
        &\leq 3\appreps + \abs{1 - e^{\nu^2 - \twonorm{\vec{u}_{\overline{S}}}^2}}\abs{p_{\vec{u}_S}(\vec{u}_{\overline{S}})} \\
        &\leq 3\appreps + \abs{1 - e^{\nu^2 - \twonorm{\vec{u}_{\overline{S}}}^2}}(1 + 3\appreps)\\
        &\leq 3\appreps + 0.1 \eps.
    \end{align*}
    Above, we use \eqref{eq:completeness-net}, $\toleps \leq 0.01 \eps/B$, and \cref{claim:subspace-to-function}; we satisfy the assumptions since $\twonorm{\vec{u}} \leq B + 2\toleps \leq 1.1 B$ and $\infnorm{\vec{u}_{\overline{S}}} \leq \mu \leq \frac{\sqrt{\appreps}}{2B} \leq \frac{\sqrt{\appreps}}{\twonorm{\vec{u}}}$.
    Altogether, we have that
    \begin{align*}
        p_{\vec{z}_{S}, \nu}(\vec{z}_{\overline{S}})
        &\geq \bra{\pi_{\vec{z}}} \rho \ket{\pi_{\vec{z}}} - 3\appreps - 0.1\eps \\
        &\geq \bra{\pi_{\vec{u}}} \rho \ket{\pi_{\vec{u}}} - 3\appreps - 0.1\eps - 2\toleps \\
        &\geq \eta - \eps/2 + \appreps.
    \end{align*}
    Thus, \cref{alg:inner_loop} does not output $\bot$.
\end{proof}

\section{Polynomial optimization} \label{sec:opt}

In this section, we provide an algorithm to solve the polynomial optimization problem subject to subspace constraints obtained in \cref{sec:reduction-to-poly-opt}. We note that while optimizing worst-case polynomial systems is hard, we are working in the regime where the polynomial when viewed as a tensor has Frobenius norm bounded by $1$. This regime is reminiscent of optimizing dense CSP's, where additive error approximations suffice. In contrast, we are optimizing the polynomial over the sphere, and do not require appealing to regularity-like statements. The key technical contribution in this section is to show that our polynomial optimization problem with subspace constraints admits a small $\eps$-net. 

%\paragraph{What we need in the optimization section}
\begin{definition}[Polynomial notation]
Let $T^{(0)} \in \C, T^{(1)} \in (\C^{n})^{\otimes 2}, \dots , T^{(d)} \in (\C^{n})^{\otimes 2d}$ be tensors.
For $\vec{x} \in \C^n$, we denote
\[
    f_{T^{(0)}, \dots , T^{(d)} }(\vec{x}) = T^{(0)} + \langle T^{(1)}, \vec{x}^* \otimes \vec{x} \rangle + \dots + \langle T^{(d)} , (\vec{x}^*)^{\otimes d} \otimes \vec{x}^{\otimes d} \rangle
\]
\end{definition}

%\ewin{First, the domain.}
\begin{definition}[Domain with subspace and flatness constraints] \label{def:opt-domain}
For a matrix $A \in \C^{r \times n}$ and vector $\vec{v} \in \R^r$, parameters $\nu, \mu$, and tolerance $\gamma$,  we define the set $\mathcal{D}_{\vec{v}, A, \nu, \mu}^{\gamma} = \{\vec{x} \in \C^n : |\norm{\vec{x}}_2 - \nu| \leq \gamma, \norm{A \vec{x} - \vec{v}}_{2} \leq \gamma , \norm{\vec{x}}_{\infty} \leq \mu + \gamma \}$.
\end{definition}

The main theorem that we will prove is stated below.
\begin{theorem}[Polynomial optimization over a subspace with flat vectors]\label{thm:opt-main}
Let $T^{(k)} \in (\C^n)^{\otimes 2k}$ be tensors for all $k = 0, 1, \dots, d$, and assume $\sum_{k=0}^d \fnorm{T_k} \leq 1$. Let $A \in \C^{r \times n}$ be a matrix with $\norm{A}_{\op} \leq 1$ and let $\vec{v} \in \C^r$ be a specified vector.  Given positive $\nu, \mu, \eps$ such that $\nu \leq 1$ and $\eps \leq 1$, let $\gamma = \poly(1/n, 1/d,  \eps)$. Then, there is an algorithm that runs in time $(1/\gamma)^{\poly(d,r, 1/\eps, 1/\mu)}$ that outputs either an $x \in \mathcal{D}_{\vec{v}, A, \nu, \mu}^{2\gamma}$ (or $\bot$ if $\mathcal{D}_{\vec{v}, A, \nu, \mu}^{2\gamma}$ is empty). The output satisfies
\[
    \abs{f_{T^{(0)}, \dots , T^{(d)}}(\vec{x})} \geq \max_{\vec{y} \in \calD_{\vec{v}, A, \nu, \mu}^\gamma} \abs{f_{T^{(0)}, \dots , T^{(d)}}(\vec{y})} - \eps \,.
\]
\end{theorem}

%\allen{TODO: need to do SVD over all flattenings because tensors are no longer symmetric}

Now we describe the algorithm for solving the polynomial optimization problem.

\begin{longfbox}[breakable=false, padding=1em, margin-top=1em, margin-bottom=1em]
\begin{algorithm}[Polynomial optimization under product state cover constraints]
\label{alg:poly-opt}
\mbox{}
    \begin{description}
    \item[Input:] Tensors $T^{(0)} \in \C, \dots, T^{(d)} \in (\C^n)^{\otimes 2d}$, matrix $A \in \C^{r \times n}$, accuracy parameter $0<\eps<1$, bound $0 \leq \nu \leq 1$.
    \item[Output:] A vector $\vec{x}$ such that  
    \[
    \abs{f_{T^{(0)}, \dots , T^{(d)}}(\vec{x})} \geq \max_{\vec{y} \in \calD_{\vec{v}, A, \nu, \mu}^\gamma} \abs{f_{T^{(0)}, \dots , T^{(d)}}(\vec{y})} - \eps 
    \]
    where $\mathcal{D}_{\vec{v}, A, \nu, \mu}^{\gamma} $ is defined in Definition~\ref{def:opt-domain} and $\gamma = \poly(1/n, 1/d, \eps )$.  
    \item[Operation:]\mbox{}
    \begin{algorithmic}[1]
        \For{ $j \in [d]$ }
        \For{$i \in [2j]$}
        \State Let $M_{j,i}$ be the $n \times n^{2j-1}$ flattening of $T^{(j)}$ along the $i$th mode. 
        \State Compute the SVD of $M_{j,i}$ and let $W_{j,i}$ be the subspace corresponding to singular values that are at least $\epsilon/(d+1)^2$.  
        \EndFor
        \EndFor
        \State Let $W$ be the subspace corresponding to the combined span of $\{ W_{j,i},  W_{j,i}^*\}_{j \in [d], i\in [2j]}$  \label{step:compute-subspace}
        \State Let $W'$ be the combined span of $W$ and the rows of $A$
        \State Let $k = O(1/\mu^2)$. 
        \For{ $S \in [\binom{n}{k}]$ }
        \State Let $Z_S$ be the subspace corresponding to the $k$ coordinate vectors indexed by $S$. 
        \State Construct the net $\calN_{W', Z_S, \gamma}$ by taking a $\gamma$-net of the ball of radius $1 + \gamma$  in the subspace spanned by $W'$ and $Z_S$ and removing all elements that are not in $\mathcal{D}_{\vec{v}, A, \nu, \mu}^{2\gamma}$ \label{line:nets}
        %such that 
        %\begin{itemize}
            %\item For each $\vec{y} \in \calN_{W', Z_S, \gamma} $, $\vec{y} \in \mathcal{D}_{\vec{v}, A, \nu, \mu}^{2\gamma} $
            %\item For each $\vec{y} \in \mathcal{D}_{\vec{v}, A, \nu, \mu}^{\gamma}$, there is a $\vec{y}' \in \calN_{W', Z_i, \gamma}$ such that $\norm{\vec{y} - \vec{y}'}_2 \leq \gamma$
        %\end{itemize} 
        \EndFor 
    
        \State Construct a larger net by  concatenating all the nets above, i.e. let $\calN_{\gamma} = \cup_{S \in [\binom{n}{k}]} \calN_{W',Z_S, \gamma}$. 
         %\ainesh{a larger net, concatenating all the nets for all sparsity patterns with $k$ non-zeros}
        \State Output $$\vec{x} = \max_{\vec{y} \in  \calN_{\gamma}} \hspace{0.05in} |f_{T^{(0)}, \dots , T^{(d)}} (\vec{y}) | $$ which can be computed by iterating over each $\vec{y} \in  \calN_{\gamma}$. 
    \end{algorithmic}
    \end{description}
\end{algorithm}
\end{longfbox}

We begin by showing that the function $f_{T^{(0)}, \dots , T^{(d)}} (\vec{y}) $ essentially only depends on the projection of $\vec{y}$ onto some constant-dimensional subspace $W$, up to additive error $\eps$.
%there is a fixed subspace of dimension  $\poly(d/\eps)$ such that projecting on this subspace incurs additive error at most $\eps$.  

\begin{lemma}[Effective dimension]\label{lem:const-dim-subspace}
Let $T^{(0)} \in \C, T^{(1)} \in (\C^n)^{\otimes 2}, \dots , T^{(d)} \in (\C^n)^{\otimes 2d} $ be tensors and assume $\norm{T^{(0)}}_F,  \dots , \norm{T^{(d)}}_F \leq 1$.  For any parameter $\eps > 0$, the subspace $W$ computed in line~\ref{step:compute-subspace} of Algorithm~\ref{alg:poly-opt} has dimension at most $8(d+1)^6/\eps^2$ and satisfies that for any $\vec{y} \in \C^n$ with $\norm{\vec{y}}_2 \leq 1$,
\[
|f_{T^{(0)}, \dots , T^{(d)}}(\vec{y}) - f_{T^{(0)}, \dots , T^{(d)}}(\Pi_W \vec{y})| \leq \eps \,, 
\]
where $\Pi_W$ is the orthogonal projection matrix for the subspace $W$. 
\end{lemma}
\begin{proof}
Consider $j \in [d]$ and a fixed tensor, $T^{(j)} \in (\C^n)^{\otimes 2j}$ such that $\norm{T^{(j)}}_F\leq 1$. Let $M_{j,i} \in \C^{ n \times n^{d-1}}$ be the flattening of $T^{(j)}$ into a $n \times n^{d-1}$ matrix along the $i$th mode and let $\sigma = (\sigma_{1,i},  \ldots, \sigma_{n,i})$ be the vector of singular values of $M_{j,i}$.  Since $\norm{ M_{j,i}}^2_F \leq 1$, it follows that $\norm{\sigma}^2_2 \leq 1$ and therefore there are at most $4(d+1)^4/\eps^2$ singular values that are at least $\eps/(2(d+1)^2)$. Now $W_{j,i}$ is the subspace corresponding to the  large singular values.  We can bound 
\[
\begin{split}
&\langle T^{(j)}, (\vec{y}^*)^{\otimes j} \otimes \vec{y}^{\otimes j} \rangle  - \langle T^{(j)} , (\Pi_{W}^*\vec{y}^*)^{\otimes j} \otimes \vec{y}^{\otimes j} \rangle \\ &= \sum_{i = 1}^j \left( \langle T^{(j)}, (\vec{y}^*)^{\otimes j - i + 1} \otimes (\Pi_{W}^*\vec{y}^*)^{\otimes i-1}  \otimes \vec{y}^{\otimes j} \rangle - \langle T^{(j)}, (\vec{y}^*)^{\otimes j - i } \otimes (\Pi_{W}^*\vec{y}^*)^{\otimes i}  \otimes \vec{y}^{\otimes j} \rangle \right)   \\ &= \sum_{i = 1}^j (\vec{y}^* - \Pi_{W}^*\vec{y}^*)^\top M_{j,i} \textsf{vec}\left((\vec{y}^*)^{\otimes j - i} \otimes (\Pi_{W}^*\vec{y}^*)^{\otimes i-1}  \otimes \vec{y}^{\otimes j} \right)  \\ 
%& =\sum_{j = 1}^d \vec{y}^\top \Paren{I - \Pi_{W}} M_{j,i} \textsf{vec}\left(\vec{y}^{\otimes d - j} \bigotimes (\Pi_{W}\vec{y})^{\otimes j-1} \right)  \\
& \leq\frac{\eps}{2(d+1)} \,,
\end{split}
\]
where the last inequality follows from observing that $(\vec{y}^* - \Pi_{W}^*\vec{y}^*)$ is orthogonal to the subspace $W^*$ and since $W^*$ contains  $W_{j,i}$, this vector is also orthogonal to $W_{j,i}$.  Similarly, we have
\[
\langle T^{(j)}, (\Pi_{W}^*\vec{y}^*)^{\otimes j} \otimes \vec{y}^{\otimes j} \rangle  - \langle T^{(j)} , (\Pi_{W}^*\vec{y}^*)^{\otimes j} \otimes (\Pi_{W}\vec{y})^{\otimes j}  \rangle \leq \frac{\eps}{2(d+1)} \,.
\]
Now we can repeat the above argument for all of the tensors $T^{(1)}, \dots , T^{(d)}$ and use triangle inequality to get the desired bound.  Since $W$ is the union of the spans of $W_{j,i}, W_{j,i}^*$, its dimension is at most $8(d+1)^6/\eps^2$, as desired.
%\ainesh{sketch: for each tensor, consider an arbitrary $n \times n^{k-1}$ flattening, write SVD, and observe that there are at most $1/\eps^2$ singular values that are at least $\eps$. The operator norm of the rest of the space is $\leq \eps$ and therefore only contributes $\eps$ to the objective value. }
\end{proof}

Next, we show a structural statement that for sets of the form $\calD_{\vec{v}, A, \nu, \mu}^\gamma$, if they are nonempty, then they contain a feasible point that is a linear combination of the rows of $A$ and a sparse vector.

\begin{lemma}[Structure of the optimizer]\label{lem:check-feasibility}
Let $A \in \C^{r \times n}$ be a matrix.  Also assume we are given parameters $0 < \nu \leq 1, \mu, \gamma > 0$ and $\vec{v} \in \C^r$.  If $\mathcal{D}_{\vec{v}, A, \nu, \mu}^{\gamma}$ is nonempty, then there exists some $\vec{x} \in \mathcal{D}_{\vec{v}, A, \nu, \mu}^{\gamma}$ that can be written as a sum of two vectors $\vec{u} + \vec{s}$ where $\vec{u}$ is in the row subspace of $A$ and  $\vec{s}$ is $1/\mu^2 + 1$-sparse.   
\end{lemma}
\begin{proof}
We may assume $r + 1/\mu^2 < n$ as otherwise the statement is trivially true.  

Consider the solution $\vec{x}^{(0)} \in \mathcal{D}_{\vec{v}, A, \nu, \mu}^{\gamma}$  that is lexicographically maximal in coordinate magnitude, i.e. it  lexicographically maximizes  the sequence $|x^{(0)}_1| , \dots , |x^{(0)}_n|$.  Note that the set $\mathcal{D}_{\vec{v}, A, \nu, \mu}^{\gamma}$ is closed so this $\vec{x}^{(0)}$ is well-defined.    

Let $S \subseteq [n]$ be the set of coordinates $j$ such that $|x^{(0)}_j| = \mu + \gamma$.  Let $V$ be the subspace spanned by the rows of $A$ and $\{ \vec{e}_j \}_{j \in S}$ (where $\vec{e}_j$ are the standard basis vectors).

Let $T \subseteq [n]$ be the set of coordinates $j'$ such that $\vec{e}_{j'}$ is in $V$.  If $|T| = n$, then we are trivially done since $|S| \leq 1/\mu^2$ and the rows of $A$ and $\{ \vec{e}_j \}_{j \in S}$ would span all of $\C^n$.  Now assume that $|T| < n$ so there is some index that is not in $T$.  Let $j_0$ be the lexicographically minimum index such that $j_0 \notin T$. We attempt to construct a vector $\vec{\Delta} \in \C^n$ that is orthogonal to $V$ and $\vec{x}^{(0)}$ and has nonzero coordinate on $j_0$.  If such a vector does not exist, then this implies that $\vec{e}_{j_0}$ is in the span of $V$ and $\vec{x}^{(0)}$.  Since by assumption, $\vec{e}_{j_0}$ is not in the span $V$, this also implies that $\vec{x}^{(0)}$ is in the span of $V$ and $\vec{e}_{j_0}$ which then immediately implies the desired statement.  Now it remains to consider the case when such a vector $\vec{\Delta}$ exists.  Now consider replacing $\vec{x}^{(0)}$ with the vector
\[
\vec{x}' = \Pi_{V} \vec{x}^{(0)} +   \frac{\norm{ \Paren{ I - \Pi_{{V} } }\vec{x}^{(0)} }_2} {\sqrt{\norm{\Paren{I - \Pi_{{V}}}\vec{x}^{(0)}}_2^2  + |z|^2}} \cdot ( \Paren{I -  \Pi_{{V}} } \vec{x}^{(0)} + z \vec{\Delta})
\]
for a complex number $z$.  Note that $\norm{\vec{x}'}_2 = \norm{\vec{x}^{(0)}}_2$ since $\langle \vec{\Delta}, \Paren{ I - \Pi_{{V}}} \vec{x}^{(0)} \rangle  = \langle \Delta, \vec{x}^{(0)} \rangle = 0$ so 
\[
\norm{\Paren{I - \Pi_{{V}}} \vec{x}^{(0)} + z \vec{\Delta} }_2 = \sqrt{\norm{ \Paren{I -  \Pi_{{V} }} \vec{x}^{(0)}}_2^2  + |z|^2} \,.
\]
Also, the projection onto the subspace $V$ is unchanged so $A\vec{x}' = A\vec{x}^{(0)}$ and $\vec{x}'$ and $\vec{x}^{(0)}$ match on all coordinates in $S$.  Thus, there is some positive $\delta$ such that the above vector is in $\mathcal{D}_{\vec{v}, A, \nu, \mu}^{\gamma}$ for all complex numbers $z$ with $|z| \leq \delta$.  If $\norm{ \Paren{ I - \Pi_{{V} }} \vec{x}^{(0)}}_2 > 0$, then there would be some choice of $z$ that increases the magnitude of $x^{(0)}_{j_0}$ (without changing any of the coordinates with indices smaller than $j_0$) and this contradicts the maximality of $\vec{x}^{(0)}$.  Thus, we must actually have $\Pi_{{V} } \vec{x}^{(0)} = \vec{x}^{(0)}$ meaning that $\vec{x}^{(0)}$ is in the span of the rows of $A$ and $\{\vec{e}_j \}_{j \in S}$ which immediately gives the desired property.
\end{proof}

Combining the two lemmas above suffices to show that our optimization problem admits a small net. 

\begin{proof}[Proof of Theorem~\ref{thm:opt-main}]
Using \cref{lem:const-dim-subspace}, we can ensure that the net $\calN_{\gamma}$ described in \cref{alg:poly-opt} is of size $\Paren{n /\gamma}^{\poly(d, r, 1/\eps, 1/\mu)}$ and can be constructed efficiently via a greedy procedure. It remains to show that this net must contain a vector with $\vec{x}$ such that $|f_{T^{(0)}, \dots , T^{(d)}}(\vec{x})| \geq \max_{ \vec{y} \in \calD_{\vec{v}, A, \nu, \mu}^\gamma} |f_{T^{(0)}, \dots , T^{(d)}}(\vec{y} )| -\eps$.

%It follows from Lemma~\ref{lem:const-dim-subspace} that we can first reduce to a subspace of size $(d+1)^5/\eps^2$ and only incur additive error $\eps$ and this projection can be computed efficiently. Let $W$ be the subspace in the combined span of $W_1, \ldots , W_d$ and $V$, where $W_i$ is the subspace spanned by the large singular values of appropriate flattenings of $T_i$. 

Note that for $\vec{x}, \vec{x}' \in \C^n$ with $\norm{\vec{x}}_2 \leq 1 + 2\gamma, \norm{\vec{x} - \vec{x}'}_2 \leq 2\gamma$, we have 
\[
|f_{T^{(0)}, \dots , T^{(d)}}(\vec{x}) - f_{T^{(0)}, \dots , T^{(d)}}( \vec{x}')| \leq \eps \,.
\]

By Lemma~\ref{lem:const-dim-subspace}, it suffices to argue that when we project all points in $\calN_{\gamma}$ onto $W$, the points form a $2\gamma$-net of the projection of $\calD^{\gamma}_{\vec{v}, A, \nu , \mu}$ onto $W$.

Consider any point $\vec{y} \in \calD^{\gamma}_{\vec{v}, A, \nu , \mu}$.  Let $M$ be a matrix whose rows form an orthonormal basis for $W$ and let $B$ be the matrix obtained by stacking $A$ and $M$ and let $\vec{v}'$ be the vector obtained by stacking $\vec{v}$ and $M\vec{y}$.  Then by construction, $\vec{y}$ is an element of the set $\calD^{\gamma}_{\vec{v}', B, \nu , \mu}$ \---- in particular, $\calD^{\gamma}_{\vec{v}', B, \nu , \mu}$ is nonempty.  Lemma~\ref{lem:check-feasibility} then implies that there is some element of $\vec{y}' \in \calD^{\gamma}_{\vec{v}', B, \nu , \mu}$ that can be written as the sum of a vector in the combined span of $A$ and $W$ and a $O(1/\mu^2)$-sparse vector.  The construction of $\calN_{\gamma}$ then implies that there is some $\vec{y}'' \in \calN_{\gamma}$ such that $\norm{\vec{y}' - \vec{y}''}_2 \leq \gamma$ \---- note this is because  $\vec{y'} \in \calD^{\gamma}_{\vec{v}', B, \nu , \mu} \subseteq  \calD^{\gamma}_{\vec{v}, A, \nu , \mu}$ so the entire $\gamma$-radius ball around $\vec{y'}$ is contained in $\calD^{2\gamma}_{\vec{v}, A, \nu , \mu}$ and thus when constructing the net in line~\ref{line:nets} of Algorithm~\ref{alg:poly-opt}, we never remove any relevant points.  Then, 
\[
\norm{\Pi_W \vec{y}'' - \Pi_W \vec{y}}_2 \leq \norm{\Pi_W \vec{y}'' - \Pi_W \vec{y}'}_2 + \norm{\Pi_W \vec{y}' - \Pi_W \vec{y}}_2 \leq 2\gamma 
\]
and this shows that $\calN_{\gamma}$, when projected onto $W$, forms a $2\gamma$-net of the projection of $\calD^{\gamma}_{\vec{v}, A, \nu , \mu}$ onto $W$.  Thus $\calN_{\gamma}$ must contain an $\eps$-approximate maximizer and we are done.

%Now, we invoke Lemma~\ref{lem:check-feasibility} with $W$ as the subspace, and thus, the optimizer can be written as $u + s$ where  $u \in W$ and $s$ is $\poly(1/\mu)$-sparse. \cref{alg:poly-opt} tries all sparsity patterns with  $\poly(1/\mu)$ non-zeros, and for each one, constructs a $\gamma$-net over the resulting  $\poly(1/\mu)$-dimensional subspace and thus finds a vector that is $\gamma$-close to $u+s$, which concludes the proof.  

% So far we have not used the structure of the set $\calD_{v, V, \nu, \mu}^{\gamma}$.  

% This follows from combining Lemma~\ref{lem:const-dim-subspace} and .
\end{proof}

\section{Hardness}

\begin{definition}
The spectral norm of a tensor $T \in \C^{n \times n \times n \times n}$ is defined as follows.
\[
    \norm{T}_{\op} = \max_{\vec{x},\, \vec{y},\, \vec{u},\,\vec{v} \in \C^n} \frac{|\langle T, \vec{x} \otimes \vec{y} \otimes \vec{u} \otimes \vec{v} \rangle|}{\norm{\vec{x}}_2\norm{\vec{y}}_2\norm{\vec{u}}_2\norm{\vec{v}}_2} \,.
\]
\end{definition}
Note that we define spectral norm to be maximizing over vectors with complex entries.
\begin{theorem}\label{thm:tensor-pca-hardness}
    It is NP-hard to approximate the spectral norm of an $n \times n \times n \times n$ tensor $T$ to within additive error $\frac{\norm{T}_F}{100n^4}$.
\end{theorem}
\begin{proof}
This is essentially shown in \cite[Theorem 8.6]{fl17}, but the theorem statement did not contain quantitative bounds.
We will thus re-prove it here.

For an undirected graph $G = (V, E)$ on $n$ vertices with at least one edge, define the tensor $A_G = \sum_{(s, t) \in E} A^{(st)}$ where $A^{(st)} \in \C^{n \times n \times n \times n}$ is the tensor where the $(i,j,k,l)$th entry is $1/2$ if and only if $i, j, k, l$ is some permutation of two $s$'s and two $t$'s:
\begin{equation*}
    A_{ijkl}^{(s,t)} = \begin{cases}
        1/2 & i = s,\, j = t,\, k = s,\, l = t \\
        1/2 & i = t,\, j = s,\, k = t,\, l = s \\
        1/2 & i = s,\, j = t,\, k = t,\, l = s \\
        1/2 & i = t,\, j = s,\, k = s,\, l = t \\
        0 & \text{otherwise.}
    \end{cases}
\end{equation*}
By \cite[Theorem 8.4]{fl17}, $\norm{A_G}_{\op} = \frac{\kappa(G) - 1}{\kappa(G)}$, where $\kappa(G) \in [n]$ is the clique number of $G$.
The clique number is NP-hard to compute~\cite{karp72}, and if we have an estimate $\nu$ such that $\abs{\nu - \norm{A_G}_{\op}}$ to $\frac{1}{100n^2}$ error, then $\abs{\frac{1}{1 - \nu} - \kappa(G)} = \abs{\frac{1}{1 - \nu} - \frac{1}{1 - \norm{A_G}_{\op}}} < \frac12$, so we can determine $\kappa(G)$ by computing $\frac{1}{1 - \nu}$ and rounding to the nearest integer.
Thus, it is NP-hard to compute $\norm{A_G}_{\op}$ to $\frac{1}{100n^2}$ error.
To conclude, observe that $\fnorm{A_G} \leq n^2$.
\end{proof}

The main theorem that we will prove in this section is the following.

\begin{theorem}\label{thm:hardness}
Given an algorithm for agnostically learning product states such that for any $n$ qubit state and target error $\eps$, the algorithm has sample complexity and running time $f(n,\eps)$ for some function $f$ and succeeds with probability $0.99$, we can give a quantum algorithm for approximating the spectral norm of a $m \times m \times m \times m$ tensor $T$ to additive error $\eps\norm{T}_F$ with running time $f( \poly(m/\eps), \poly(\eps))$ that succeeds with probability $0.9$.
%Approximating the spectral norm of a $m \times m \times m \times m$ tensor $T$ to additive error $\eps\norm{T}_F$ reduces to agnostically learning product states on $n$ qubits with error $\poly(\eps)$ with $\poly(n)$ pre-processing time as long as $n = \Omega(m^{10} / \eps^{10})$. \ewin{high probability, pre-processing time is quantum.}
\end{theorem}

\begin{definition}\label{def:reduction}
Assume we are given a tensor $T \in \C^{n \times n \times n \times n}$ with $\norm{T}_F = 1$.
Then we construct a quantum state on $4n$ qubits as follows.
\[
    \ket{\psi_{T}} = \sum_{i,j,k,l} T_{ijkl}' \ket{e_ie_je_ke_l},
\]
where for $i,j,k,l \in [n]$, $\ket{e_ie_je_ke_l} = \ket{e_i} \otimes \ket{e_j} \otimes \ket{e_k} \otimes \ket{e_l}$, where we recall that $\ket{e_i}$ is the product state which is $\ket{1}$ on the $i$th qubit and $\ket{0}$ on the other $n-1$ qubits.
\end{definition}

One can verify that this is a valid quantum state, since $\twonorm{\ket{\psi_{T}}} = \fnorm{T} = 1$.
We will consider the quantum state corresponding to $U^{\otimes 4} T$, where $U \in \C^{n \times m}$ is a Haar-random matrix with orthonormal columns and $T \in \C^{m\times m\times m\times m}$ for some $m \leq n$.
Here, we are randomly embedding the tensor into a larger, $n$-dimensional space.

\begin{claim}\label{claim:bounded-entries}
For $T \in \C^{m\times m\times m\times m}$ and a Haar-random $U \in \C^{n \times m}$, as long as $m \geq 10 \log n$, then with $0.99$ probability over the randomness of $U$, $\max_{i,j,k,l \in [n]} \abs{(U^{\otimes 4}T )_{ijkl}} \leq \frac{(10m)^2 }{n^2}$.
%\ewin{what is high probability? also probably need some citation.}
\end{claim}
\begin{proof}
Note that the columns of $U$ determine a random $m$-dimensional subspace of $\C^n$.  With $0.99$ probability over the randomness of $U$, all unit vectors in the subspace spanned by the columns $U$ have entries with magnitude at most $\sqrt{10m/n}$.  If this happens, then since $T'$ is in the subspace spanned by the columns of $U \otimes U \otimes U \otimes U$ and $\norm{T'}_F = 1$, we must have $\max_{i,j,k,l \in [n]} \abs{T_{ijkl}'}  \leq \frac{(10m)^2 }{n^2}$.
\end{proof}

\begin{lemma}\label{lem:reduction-analysis}
Given a tensor $T \in \C^{n \times n \times n \times n}$ with $\norm{T}_F = 1$ and $M = \max_{i,j,k,l} n^2\abs{T_{ijkl}}$.
Let $\OPT_{T}$ be the spectral norm of $T$ and let $\OPT_{\ket{\psi_{T}}} $ be defined as:
%the maximum fidelity of $\psi_{T}$ with any product state, i.e.
\[
    \OPT_{\ket{\psi_{T}}} = \max_{\sigma = \ket{\sigma_1} \otimes \dots \otimes \ket{\sigma_n}} |\braket{\sigma |\psi_{T}}| \,.
\]
%\will{fidelity would be $\abs{\braket{\sigma |\psi_{T}}}^2$...}
Then, for sufficiently large $n$,
\[
    e^{-2} \OPT_{T} - \frac{10M}{n^{0.2}} \leq \OPT_{\ket{\psi_{T}}} \leq e^{-2}\OPT_{T} + \frac{10}{n^{0.1}} + \frac{10M}{n^{0.2}} \,.
\]
\end{lemma}
\begin{proof}
First we prove the lower bound.
Let $\vec{x},\vec{y},\vec{u},\vec{v} \in \C^n$ be unit vectors such that $\abs{\langle T, \vec{x} \otimes \vec{y} \otimes \vec{u} \otimes \vec{v} \rangle} = \OPT_T$.
Let $\vec{x}', \vec{y}', \vec{u}'$, and $\vec{v}'$ be obtained by zeroing out all entries have magnitude larger than $1/n^{0.1}$.
\begin{align*}
    \abs{\OPT_T - \langle T, \vec{x}' \otimes \vec{y}' \otimes \vec{u}' \otimes \vec{v}' \rangle}
    &= \abs{\langle T, \vec{x} \otimes \vec{y} \otimes \vec{u} \otimes \vec{v} - \vec{x}' \otimes \vec{y}'
    \otimes \vec{u}' \otimes \vec{v}' \rangle} \\
    &\leq \frac{100M}{n^{0.4}} \,.
\end{align*}
The last inequality follows from observing that the right-hand side of the inner product has Frobenius norm bounded by 2 and is non-zero in at most $4n^{3.2}$ entries.
Thus, we can apply Cauchy--Schwarz on the non-zero entries, and the Frobenius norm on $T$ restricted to such entries is at most $\frac{M}{n^2} \cdot \sqrt{4n^{3.2}}$.
Then the product state $\ket{\pi_{\vec{x}', \vec{y}', \vec{u}', \vec{v}'}}$ satisfies, for $\xi = (\prod_{a=1}^n (1 + \abs{x_i'}^2)(1 + \abs{y_i'}^2)(1 + \abs{u_i'}^2)(1 + \abs{v_i'}^2))^{-1/2}$,
% We construct the product state $\pi = \pi^{(x)} \otimes \pi^{(y)} \otimes \pi^{(u)} \otimes \pi^{(v)}$, where
% \[
%     \pi^{(x)} = (\sqrt{1 - \abs{x_1'}^2}\ket{0} + x_1' \ket{1}) \otimes \dots \otimes (\sqrt{1 - \abs{x_n'}^2}\ket{0} + x_n' \ket{1}),
% \]
% and analogously for $y$, $u$, and $v$.
\begin{align*}
    \abs{\braket{\psi_{T} | \pi }}
    = \xi \abs[\Big]{\sum_{i,j,k,l \in [n]} (T_{ijkl})^* x_i'y_j'u_k'v_l' }
    = \xi \abs{\angles{T , \vec{x}' \otimes \vec{y}' \otimes \vec{u}' \otimes \vec{v}'}}
    \geq \xi\parens[\big]{\OPT_T - \frac{10M}{n^{0.4}}}.
\end{align*}
Finally, by \cref{lem:dtan-fidelity-approx}, $\xi \geq e^{-\frac12(\norm{\vec{x}'}_2^2 + \norm{\vec{y}'}_2^2 + \norm{\vec{u}'}_2^2 + \norm{\vec{v}'}_2^2)} \geq e^{-2}$.
%\[
%    \begin{split}
%        \abs{\braket{\sigma | \psi_{T,U}}} & = \abs[\Big]{ \sum_{i,j,k,l \in [n]} \wt{x}_i\wt{y}_j\wt{u}_k\wt{v}_l T'_{ijkl} } \\
%        & \geq e^{-2 \sum_{j = 1}^n |\wt{x_j}|^2 + |\wt{x_j}|^4} \cdot \left( \frac{4}{\sqrt{24}} |\langle T'' , \wt{x'}^{\otimes 4} \rangle| \right) - \Delta \\
%        & \geq \alpha e^{-\frac{2}{n^{0.2}} \sum_{j = 1}^n |\wt{x_j}|^2 }   |\langle T'' ,  \wt{x'}^{\otimes 4} \rangle|  - \frac{(30m)^2}{\sqrt{n}} \\
%        & \geq \left(\alpha - \frac{1}{n^{0.2}}\right) \langle T ,  x^{\otimes 4} \rangle - \frac{(200m )^2}{n^{0.4}}
%    \end{split}
%\]
%which completes the proof of the lower bound (clearly we always have $\OPT_T \leq 1$).

Now we prove the upper bound.
Let $\vec{x}, \vec{y}, \vec{u}, \vec{v} \in \C^n$ be vectors attaining the optimum fidelity, $\abs{\braket{\psi_{T} | \pi_{\vec{x},\vec{y},\vec{u},\vec{v}}}} = \OPT_{\ket{\psi_{T}}}$.
Here, we interpret the $4n$-length product state as having 4 parameter vectors of length $n$.
We have
\[
    \braket{\psi_{T} | \pi_{\vec{x},\vec{y},\vec{u},\vec{v}}} = \xi \parens[\Big]{\sum_{i,j,k,l \in [n]} (T_{ijkl})^* x_i y_j u_k v_l} = \xi\angles{T, \vec{x} \otimes \vec{y} \otimes \vec{u} \otimes \vec{v}},
\]
where $\xi = (\prod_{a=1}^n (1 + \abs{x_a}^2)(1 + \abs{y_a}^2)(1 + \abs{u_a}^2)(1 + \abs{v_a}^2))^{-1/2}$.
Let $\vec{x}', \vec{y}', \vec{u}', \vec{v}'$ be obtained by zeroing out entries larger than $1/n^{0.1}$.
Then by the same argument as above,
\begin{align*}
    \abs{\angles{T, \vec{x} \otimes \vec{y} \otimes \vec{u} \otimes \vec{v}}
    - \angles{T, \vec{x}' \otimes \vec{y}' \otimes \vec{u}'
    \otimes \vec{v}'}}
    &= \abs{\angles{T, \vec{x} \otimes \vec{y} \otimes \vec{u} \otimes \vec{v} - \vec{x}' \otimes \vec{y}' \otimes \vec{u}' \otimes \vec{v}'}} \\
    &\leq \frac{10M}{n^{0.4}}\norm{\vec{x}}_2\norm{\vec{y}}_2\norm{\vec{u}}_2\norm{\vec{v}}_2.
\end{align*}
Using this, we have
\begin{align}
    \abs{\braket{\psi_{T} | \pi_{\vec{x},\vec{y},\vec{u},\vec{v}}}}
    &\leq \xi \parens[\Big]{
        \angles{T, \vec{x}' \otimes \vec{y}' \otimes \vec{u}' \otimes \vec{v}'} +
        \frac{10M}{n^{0.4}}\norm{\vec{x}}_2\norm{\vec{y}}_2\norm{\vec{u}}_2\norm{\vec{v}}_2
    } \nonumber \\
    &\leq \xi \norm{\vec{x}'}_2\norm{\vec{y}'}_2\norm{\vec{u}'}_2\norm{\vec{v}'}_2 \OPT_T + \frac{10M}{n^{0.4}} \xi \norm{\vec{x}}_2\norm{\vec{y}}_2\norm{\vec{u}}_2\norm{\vec{v}}_2. \label{eq:hardness-bd}
\end{align}
We split now into two cases.
First, suppose that $\sum_{a=1}^n \parens{\frac{\abs{x_a}^2}{1 + \abs{x_a}^2} + \frac{\abs{y_a}^2}{1 + \abs{y_a}^2} + \frac{\abs{u_a}^2}{1 + \abs{u_a}^2} + \frac{\abs{v_a}^2}{1 + \abs{v_a}^2}} \leq n^{0.1}$.
This implies that
\begin{align*}
    & \sum_{a=1}^n \parens{\abs{x_a'}^2 + \abs{y_a'}^2 + \abs{u_a'}^2 + \abs{v_a'}^2} \\
    &\leq (1 + 1/n^{0.2})\sum_{a=1}^n \parens{\frac{\abs{x_a'}^2}{1 + \abs{x_a'}^2} + \frac{\abs{y_a'}^2}{1 + \abs{y_a'}^2} + \frac{\abs{u_a'}^2}{1 + \abs{u_a'}^2} + \frac{\abs{v_a'}^2}{1 + \abs{v_a'}^2}} \\
    &\leq 2 n^{0.1}.
\end{align*}
In this case, using \cref{lem:dtan-fidelity-approx} and the norm bound and magnitude bound on $\vec{x}'$,
\begin{align*}
    \xi
    &\leq \left(\prod_{a=1}^n (1 + \abs{x_a'}^2)(1 + \abs{y_a'}^2)(1 + \abs{\wt{u}_a}^2)(1 + \abs{\wt{v}_a}^2)\right)^{-1/2} \\
    &\leq e^{-\frac12(\norm{\vec{x}'}_2^2 + \norm{\vec{y}'}_2^2 + \norm{\vec{u}'}_2^2 + \norm{\vec{v}'}_2^2 - \sum_{a=1}^n(\abs{x_a'}^4 + \abs{y_a'}^4 + \abs{u_a'}^4 + \abs{v_a'}^4))} \\
    &\leq e^{-\frac12(\norm{\vec{x}'}_2^2 + \norm{\vec{y}'}_2^2 + \norm{\vec{u}'}_2^2 + \norm{\vec{v}'}_2^2 - 2n^{-0.1})}.
\end{align*}
So, plugging this into \eqref{eq:hardness-bd},
\begin{align*}
    &\abs{\braket{\psi_{T} | \pi_{\vec{x},\vec{y},\vec{u},\vec{v}}}} \\
    &\leq e^{-\frac12(\norm{\vec{x}'}_2^2 + \norm{\vec{y}'}_2^2 + \norm{\vec{u}'}_2^2 + \norm{\vec{v}'}_2^2 - n^{-0.1})} \norm{\vec{x}'}_2\norm{\vec{y}'}_2\norm{\vec{u}'}_2\norm{\vec{v}'}_2 \OPT_T + \frac{10M}{n^{0.4}} \xi \norm{\vec{x}}_2\norm{\vec{y}}_2\norm{\vec{u}}_2\norm{\vec{v}}_2 \\
    &\leq e^{-2}e^{n^{-0.1}} \OPT_T + \frac{10M}{n^{0.4}} \xi \norm{\vec{x}}_2\norm{\vec{y}}_2\norm{\vec{u}}_2\norm{\vec{v}}_2 \\
    &\leq e^{-2}\parens[\Big]{1 + \frac{10}{n^{0.1}}} \OPT_T + \frac{10M}{n^{0.2}} \\
    &\leq e^{-2}\OPT_T + \frac{10}{n^{0.1}} + \frac{10M}{n^{0.2}},
\end{align*}
where we used that the maximum possible value of $e^{-\frac12\theta^2}\theta$ is exactly $e^{-1/2}$.
This gives the desired statement in the first case.
In the second case, suppose that $\sum_{a=1}^n \parens{\frac{\abs{x_a}^2}{1 + \abs{x_a}^2} + \frac{\abs{y_a}^2}{1 + \abs{y_a}^2} + \frac{\abs{u_a}^2}{1 + \abs{u_a}^2} + \frac{\abs{v_a}^2}{1 + \abs{v_a}^2}} > n^{0.1}$.
Then we can argue that the optimum value is very small: by Cauchy--Schwarz and \cref{lem:low_weight}, for a sufficiently large $n$,
\begin{align*}
    \abs{\braket{\psi_{T} | \pi_{\vec{x},\vec{y},\vec{u},\vec{v}}}}
    &= \abs{\bra{\psi_{T}} \Pi_{\leq 4} \ket{\pi_{\vec{x},\vec{y},\vec{u},\vec{v}}}} \\
    &\leq \norm{\ket{\psi_{T}}}_2 \norm{\Pi_{\leq 4}\ket{\pi_{\vec{x},\vec{y},\vec{u},\vec{v}}}}_2 \\
    &\leq e^{-(2n^{0.1} - 4)\log(2 - 4/n^{0.1}) + (n^{0.1} - 4)} \\
    &\leq e^{-0.1 n^{0.1}} \leq  \frac{10}{n^{0.2}} \leq \frac{10M}{n^{0.2}}\,.
\end{align*}
This completes the proof of the upper bound and we are done.
\end{proof}

Now we can complete the proof of Theorem~\ref{thm:hardness}.

\begin{proof}[Proof of Theorem~\ref{thm:hardness}]
Without loss of generality, we may assume that the original $m \times m \times m \times m$ tensor $T$ that we are given is symmetric: otherwise, we can just symmetrize it and this only decreases the Frobenius norm.
We can further normalize such that $\norm{T}_F = 1$.
Then consider $U^{\otimes 4} T$ for $U \in \C^{n \times m}$ a Haar-random isometry, and the corresponding quantum state $\ket{\psi_{U^{\otimes 4}T}}$ as defined in \cref{def:reduction}.
By \cref{lem:reduction-analysis},
\begin{align*}
    e^{-2} \OPT_{U^{\otimes 4}T} - \frac{10M}{n^{0.2}}
    \leq \OPT_{\ket{\psi_{U^{\otimes 4} T}}}
    \leq e^{-2} \OPT_{U^{\otimes 4}T} + \frac{10}{n^{0.1}} + \frac{10M}{n^{0.2}}
\end{align*}
By \cref{claim:bounded-entries}, with $0.99$ probability we can take $M = (10m)^2$, and because $U$ is an isometry, $\OPT_{U^{\otimes 4}T} = \OPT_{T}$, so
\begin{align*}
    e^{-2} \OPT_{T} - \frac{1000m^2}{n^{0.2}}
    \leq \OPT_{\ket{\psi_{U^{\otimes 4} T}}}
    \leq e^{-2} \OPT_{T} + \frac{10}{n^{0.1}} + \frac{1000m^2}{n^{0.2}}.
\end{align*}
Taking $n = \Omega((m/\eps)^{20})$ completes the proof: then, $e^{-2}(\OPT_T - \eps) \leq \OPT_{\ket{\psi_{U^{\otimes 4} T}}} \leq e^{-2} (\OPT_T + \eps)$.
So, finding the optimal product state fidelity to accuracy $0.01\eps^2$ gives the spectral norm of the tensor to $\eps$ additive error.
\end{proof}

\section*{Acknowledgments}
\addcontentsline{toc}{section}{Acknowledgments}

%\ewin{ack sitan?}
We thank Sitan Chen for helpful discussions at the early stages of this work.

AB is supported by the NSF TRIPODS program (award DMS-2022448).
JB is supported by Henry Yuen's AFOSR (award FA9550-21-1-036) and NSF CAREER (award CCF-2144219).
WK acknowledges support from the U.S.\ Department of Energy, Office of Science, National
Quantum Information Science Research Centers, Quantum Systems Accelerator.
ZL is supported by the U.S. Department of Energy, Office of Science, National Quantum Information Science Research Centers, Quantum Systems Accelerator, and by NSF Grant CCF-2311733.
AL is supported in part by an NSF GRFP and a Hertz Fellowship.
RO is supported by ARO grant W911NF2110001 and by a gift from Google Quantum AI.
ET is supported by the Miller Institute for Basic Research in Science, University of California, Berkeley. 

\printbibliography

\appendix

\section{Agnostic learning of a discrete class of product states}

%\ewin{Note: the $\mathcal{P}$ notation currently conflicts with the intro notation. not a big deal but something to fix at some point.}
%\ryan{The informal version of this in the intro now looks more dissimilar to this theorem than it did before}

\begin{theorem}[Agnostic learning of a discrete class of product states] \label{thm:discrete}
    Suppose we are given copies of an $n$-qudit state $\rho$, along with classical descriptions of sets of single-qudit pure states $\{\mathcal{A}_k\}_{k \in [n]}$ such that $\abs{\mathcal{A}_k} \leq s$ and, for some $\gamma \geq 1/e$, $\abs{\braket{\phi|\phi'}}^2 \leq \gamma$ for all distinct $\ket{\phi}, \ket{\phi'} \in \mathcal{A}_k$.
    These implicitly define a class of product states $\mathcal{P}$,
    \begin{align*}
        \mathcal{P} &= \mathcal{A}_1 \otimes \dots \otimes \mathcal{A}_n = \braces{\ket{\pi_1} \otimes \dots \otimes \ket{\pi_n} \mid \ket{\pi_k} \in \mathcal{A}_k}.
    \intertext{Let $\mathcal{P}_{\eta}$ be the states in the class with fidelity at least $\eta$ with $\pi$:}
        \mathcal{P}_\eta &= \braces{\ket{\pi} \in \mathcal{P} \mid \bra{\pi} \rho \ket{\pi} \geq \eta}.
    \end{align*}
    Then given parameters $\eta, \eps, \delta \in (0,1)$ with $\eps \leq \eta/2$, we can output a set $\mathcal{S}$ such that $\mathcal{P}_\eta \subseteq \mathcal{S} \subseteq \mathcal{P}_{\eta - \eps}$ with probability $\geq 1-\delta$ using $O((10ns)^{\log\frac{25}{\eta} / \log\frac{1}{\gamma}}\frac{1}{\eps^2}\log\frac{1}{\delta})$ copies of $\rho$.
\end{theorem}

The classical part of the algorithm is linear-time, costing $O(nC)$, where $C$ is the number of copies of $\rho$ used.
The quantum gate complexity is also $O(nC)$, where we consider a gate to be an operation on a qudit; the only circuits we will need are single layers of single-qudit gates, followed by measurement, to estimate quantities like $\bra{\pi} \rho \ket{\pi}$ for $\ket{\pi} \in \mathcal{P}$.

\begin{claim} \label{claim:discrete-bd}
    In the set-up of \cref{thm:discrete}, $\abs{\mathcal{P}_\eta} \le (10ns)^{\log\frac{2}{\eta} / \log\frac{1}{\gamma}}$.
\end{claim}
\begin{proof}
The basic idea is that $\mathcal{P}_\eta$ consists of a small number of balls, where a ball is the set of elements of $\mathcal{P}$ close to a particular $\ket{\pi} \in \mathcal{P}$.
The number of balls is small because $\rho$ must place mass in the direction of every ball, and the mass of $\rho$ is bounded.

We formalize this argument now.
For $\ell = \floor{\log\frac{2}{\eta} / \log\frac{1}{\gamma}}$, construct a net $\net \subseteq \mathcal{P}_\eta$ with the following properties.
\begin{enumerate}
    \item For any $\ket{\pi} \in \mathcal{P}_\eta$, there exists a $\ket{\varpi} \in \net$ such that $\ket{\pi}$ and $\ket{\varpi}$ differ in at most $\ell$ qudits;
    \item Any distinct $\ket{\pi}, \ket{\varpi} \in \net$ differ in at least $\ell + 1$ qudits.
\end{enumerate}
Such a net can be constructed greedily: start with an empty net, and while there is a violation of condition 1, add the corresponding $\ket{\pi}$ to the net.
Let $M$ be the matrix whose columns are the elements of the net $\net = \braces{\ket{\pi^{(i)}}}_i$.
Then following the same argument as \cref{claim:cover_size},
\begin{align*}
    \opnorm{M^\dagger M} &\leq 1 + \abs{\net} \gamma^{\ell + 1} \leq 1 + \abs{\net}(\eta/2) \\
    \opnorm{MM^\dagger} &\geq \tr(MM^\dagger \rho) = \sum_i \braket{\pi^{(i)} | \rho | \pi^{(i)}} \geq \abs{\net} \eta
\end{align*}
Together, we can conclude that $\abs{\net} \leq \frac{2}{\eta}$.

Finally, by property 2 of the net, $\mathcal{P}_\eta$ is contained in the set of product states in $\mathcal{P}$ which differ from an element of $\net$ in at most $\ell$ qudits.
Consider some $\ket{\pi} \in \mathcal{P}_\eta$.
There are $\sum_{i = 0}^{\ell} \binom{n}{i} s^i \leq \sum_{i=0}^\ell (ns)^i \leq 2(ns)^\ell$ elements of $\mathcal{P}$ which differ from $\ket{\pi}$ in at most $\ell$ qudits.
So, $\abs{\mathcal{P}_\eta} \leq \abs{\net} 2(ns)^\ell \leq \frac{4}{\eta}(ns)^\ell$, which gives the desired bound.
\end{proof}

Now, we present the proof of \cref{thm:discrete}.
The algorithm is given below: we start from qudit 1 and iteratively add a qubit, maintaining at iteration $m$ a set $\mathcal{S}_m$ of product states which have good fidelity with $\rho$.

\begin{longfbox}[breakable=false, padding=1em, margin-top=1em, margin-bottom=1em]
\begin{algorithm}[Agnostic learning for discrete product state classes]
\label{alg:discrete} \hfill
\begin{description}
    \item[Input:] Copies of an $n$-qubit state $\rho$, a class of $n$-qubit product states $\mathcal{P} = \mathcal{A}_1 \otimes \dots \otimes \mathcal{A}_n$, with parameters $\eta, \gamma, \eps, \delta$ as in \cref{thm:discrete};
    \item[Output:] A set of product states $\mathcal{S}$ such that $\mathcal{P}_\eta \subseteq \mathcal{S} \subseteq \mathcal{P}_{\eta - \eps}$ with probability $\geq 1-\delta$;
    \item[Procedure:] \mbox{}
    \begin{algorithmic}[1]
        \State Let $\mathcal{S}_0 = \varnothing$;
        \For{$m$ from $1$ to $n$}
            \State Initialize $\mathcal{S}_m = \varnothing$;
            \ForAll{$\ket{\pi} \in \mathcal{S}_{m-1} \otimes \mathcal{A}_{m}$}
                \State Estimate $\bra{\pi} \rho_{[m]} \ket{\pi}$ to $\eps/2$ error with success probability $\geq 1 - \delta / (10ns)^{\log\frac{20}{\eta}/\log\frac{1}{\gamma}}$;
                \If{the estimate is at least $\eta - \eps/2$}
                    \State Add $\ket{\pi}$ to $\mathcal{S}_m$;
                \EndIf
            \EndFor
        \EndFor
        \State \Return $\mathcal{S}_n$ \;
    \end{algorithmic}
\end{description}
\end{algorithm}
\end{longfbox}

\begin{claim}[Correctness] \label{claim:discrete-correctness}
    With probability $\geq 1-\delta$, at the completion of \cref{alg:discrete}, for every $m \in \braces{0,1,\dots,n}$,
    \begin{align*}
        \mathcal{S}_m &\supseteq \braces{\ket{\pi} \in \mathcal{A}_1 \otimes \dots \otimes \mathcal{A}_m \mid \bra{\pi} \rho_{[m]} \ket{\pi} \geq \eta} \\
        \mathcal{S}_m &\subseteq \braces{\ket{\pi} \in \mathcal{A}_1 \otimes \dots \otimes \mathcal{A}_m \mid \bra{\pi} \rho_{[m]} \ket{\pi} \geq \eta - \eps}
    \end{align*}
\end{claim}
\begin{proof}
First, we consider the algorithm under the event that the algortihm never fails.
We prove by induction on $m$.
The base case $m = 0$ is true trivially.
For the inductive step, consider some $m > 0$, and consider some $\ket{\pi} \in \mathcal{A}_1 \otimes \dots \otimes \mathcal{A}_m$ such that $\bra{\pi} \rho_{[m]} \ket{\pi} \geq \eta$.
Let $\ket{\pi'}$ be the product state $\ket{\pi}$ with the $m$th qubit traced out.
Then
\begin{align*}
    \bra{\pi'} \rho_{[m-1]} \ket{\pi'} \geq \bra{\pi} \rho_{[m]} \ket{\pi} \geq \eta,
\end{align*}
so $\ket{\pi'} \in \mathcal{S}_{m-1}$ by the inductive hypothesis, and $\ket{\pi} \in \mathcal{S}_{m-1} \otimes \mathcal{A}_m$.
Because we are assuming that the estimation procedure succeeds, this means that $\ket{\pi}$ will be added to $\mathcal{S}_m$, proving the first equation of the claim.
The second equation holds because, again assuming that the estimation procedure always succeeds, all elements added to $\mathcal{S}_m$ have fidelity at least $\eta - \eps/2 - \eps/2$ with $\rho_{[m]}$.

To conclude, we account for failure.
Supposing that the algorithm never fails, the second equation in the claim along with \cref{claim:discrete-bd} (applied on $\rho_{[m]}$, the state on the $m$-qudit subsystem) implies that $\abs{\mathcal{S}_m} \leq (10ms)^{\log\frac{2}{\eta - \eps}/\log\frac{1}{\gamma}}$.
So, on a successful run, the estimation procedure is run at most $ns \cdot (10ns)^{\log\frac{2}{\eta - \eps}/\log\frac{1}{\gamma}}$ times, and so the failure probability is chosen such that the probability a failure occurs is at most $\delta$.
\end{proof}

The above claim implies that the output of the algorithm satisfies $\mathcal{P}_\eta \subseteq \mathcal{S} \subseteq \mathcal{P}_{\eta - \eps}$, which is the desired correctness condition.
What remains is to analyze the complexity.
The dominating cost is the estimation step, where the fidelity of (a subsystem of) $\rho$ with a product state is estimated to $\eps/2$ error with failure probability $\delta / (10ns)^{\log\frac{20}{\eta}/\log\frac{1}{\gamma}}$.
This step is run at most $s \cdot \sum_{m=1}^n \abs{\mathcal{S}_m} \leq ns \cdot (10ns)^{\log\frac{2}{\eta - \eps} / \log\frac{1}{\gamma}} = O((10ns)^{\log\frac{20}{\eta}/\log\frac{1}{\gamma}})$ times, by \cref{claim:discrete-correctness} and \cref{claim:discrete-bd}.
Each estimation protocol costs $O(\frac{1}{\eps^2}\log((10ns)^{\log\frac{20}{\eta}/\log\frac{1}{\gamma}}/\delta))$ copies of $\rho$, where we use the naive algorithm of measuring in the appropriate basis and estimating the corresponding probability.

\begin{remark} \label{rmk:discrete-shadows}
    The estimation procedure could also be done using randomized Clifford measurements as in \cite{hkp20}, which reduces the sample complexity to poly-logarithmic in $n$.
    However, computing the resulting estimators takes exponential time, making the resulting algorithm computationally inefficient, except for limited settings, such as in the case of stabilizer product states.
    %\ewin{We can actually do this for stabilizer product states; is this worth stating?}
\end{remark}

\section{Agnostic improper learning of matrix product states} \label{sec:mps}

The task of product state learning motivates a more general question: What ensembles of  ``low-entanglement'' states can we perform computationally-efficient agnostic learning for?  One physically-motivated ensemble is the class of low bond-dimension matrix product states, for which the (non-agnostic) learning task was studied in \cite{cpfsgb10}.  

\begin{definition}[Matrix product state with bond dimension $r$, open boundary condition]
\label{def:mps}
     A matrix product state (MPS) is a state over $n$ total $d$-dimensional qudits that can be written as follows
     \begin{equation*}
         \ket{\psi} = \sum_{s_1, \ldots, s_{n} \in [d]^{n}} \sum_{\alpha_{1}, \ldots, \alpha_{n} \in [r]^{n}} \left(A^{(s_1)}_1 \right)_{1, \alpha_1} \left(A^{(s_2)}_2\right)_{\alpha_1, \alpha_2} \ldots \left(A^{(s_{n - 1})}_{n - 1}\right)_{\alpha_{n-2}, \alpha_{n-1}} \left(A^{(s_n)}_n\right)_{\alpha_{n-1}, \alpha_{n}} \ket{s_1, \ldots, s_n}\,,
     \end{equation*}
    where for all $i = 2, \ldots, n -1$ and all $s_i$, we have that $A_i^{(s_i)}$ are $r \times r$ complex matrices, and for $i \in \{1, n\}$ and all $s_i$.
    The dimension $r$ is known as the \emph{bond dimension} of the MPS.
    We let $\mathrm{MPS}_{n, d, r}$ denote the set of all such states.\footnote{Matrix product states with periodic boundary conditions are defined by taking the trace of the product of the $A_{i}^{(s_i)}$, or equivalently enforcing that the first index of $A_1^{(s_1)}$ is equal to $\alpha_n$, but we do not consider these in this paper.}
\end{definition}
\noindent
A MPS has a natural representation in the form of a tensor network as in Figure~\ref{fig:matrix_product_state}, where here each $A_i$ represents the concatenation of all of the $A_i^{(s_i)}$ into a $3$-tensor. 
\begin{figure}[h]
    \centering
    \tikzset{every picture/.style={line width=0.75pt}} %set default line width to 0.75pt        

\begin{tikzpicture}[x=0.75pt,y=0.75pt,yscale=-1,xscale=1]
%uncomment if require: \path (0,300); %set diagram left start at 0, and has height of 300

%Rounded Rect [id:dp029890110433692496] 
\draw  [fill={rgb, 255:red, 255; green, 166; blue, 217 }  ,fill opacity=1 ] (101,132) .. controls (101,123.72) and (107.72,117) .. (116,117) -- (143,117) .. controls (151.28,117) and (158,123.72) .. (158,132) -- (158,158) .. controls (158,166.28) and (151.28,173) .. (143,173) -- (116,173) .. controls (107.72,173) and (101,166.28) .. (101,158) -- cycle ;
%Rounded Rect [id:dp9670128092834824] 
\draw  [fill={rgb, 255:red, 255; green, 166; blue, 217 }  ,fill opacity=1 ] (205,132) .. controls (205,123.72) and (211.72,117) .. (220,117) -- (247,117) .. controls (255.28,117) and (262,123.72) .. (262,132) -- (262,158) .. controls (262,166.28) and (255.28,173) .. (247,173) -- (220,173) .. controls (211.72,173) and (205,166.28) .. (205,158) -- cycle ;
%Rounded Rect [id:dp8050675897290264] 
\draw  [fill={rgb, 255:red, 255; green, 166; blue, 217 }  ,fill opacity=1 ] (357,132) .. controls (357,123.72) and (363.72,117) .. (372,117) -- (399,117) .. controls (407.28,117) and (414,123.72) .. (414,132) -- (414,158) .. controls (414,166.28) and (407.28,173) .. (399,173) -- (372,173) .. controls (363.72,173) and (357,166.28) .. (357,158) -- cycle ;
%Rounded Rect [id:dp7625345216732511] 
\draw  [fill={rgb, 255:red, 255; green, 166; blue, 217 }  ,fill opacity=1 ] (460,132) .. controls (460,123.72) and (466.72,117) .. (475,117) -- (502,117) .. controls (510.28,117) and (517,123.72) .. (517,132) -- (517,158) .. controls (517,166.28) and (510.28,173) .. (502,173) -- (475,173) .. controls (466.72,173) and (460,166.28) .. (460,158) -- cycle ;
%Straight Lines [id:da12264052284985016] 
\draw    (158,146) -- (204,146) ;
%Straight Lines [id:da7518768221228552] 
\draw    (130,117) -- (130,74) ;
%Straight Lines [id:da7975705760670686] 
\draw    (233,117) -- (233,74) ;
%Shape: Circle [id:dp701914275426412] 
\draw  [fill={rgb, 255:red, 0; green, 0; blue, 0 }  ,fill opacity=1 ] (285,146) .. controls (285,143.24) and (287.24,141) .. (290,141) .. controls (292.76,141) and (295,143.24) .. (295,146) .. controls (295,148.76) and (292.76,151) .. (290,151) .. controls (287.24,151) and (285,148.76) .. (285,146) -- cycle ;
%Shape: Circle [id:dp723119453240296] 
\draw  [fill={rgb, 255:red, 0; green, 0; blue, 0 }  ,fill opacity=1 ] (302,146) .. controls (302,143.24) and (304.24,141) .. (307,141) .. controls (309.76,141) and (312,143.24) .. (312,146) .. controls (312,148.76) and (309.76,151) .. (307,151) .. controls (304.24,151) and (302,148.76) .. (302,146) -- cycle ;
%Shape: Circle [id:dp38122507165490593] 
\draw  [fill={rgb, 255:red, 0; green, 0; blue, 0 }  ,fill opacity=1 ] (319,146) .. controls (319,143.24) and (321.24,141) .. (324,141) .. controls (326.76,141) and (329,143.24) .. (329,146) .. controls (329,148.76) and (326.76,151) .. (324,151) .. controls (321.24,151) and (319,148.76) .. (319,146) -- cycle ;
%Straight Lines [id:da3280346395152466] 
\draw    (263,146) -- (276,146) ;
%Straight Lines [id:da07854665739778544] 
\draw    (337,146) -- (357,146) ;
%Straight Lines [id:da6822765853115446] 
\draw    (415,148) -- (461,148) ;
%Straight Lines [id:da06960791440006131] 
\draw    (385,118) -- (385,75) ;
%Straight Lines [id:da024339583029751033] 
\draw    (488,117) -- (488,74) ;

% Text Node
\draw (174,127.4) node [anchor=north west][inner sep=0.75pt]    {$r$};
% Text Node
\draw (114,88.4) node [anchor=north west][inner sep=0.75pt]    {$d$};
% Text Node
\draw (217,88.4) node [anchor=north west][inner sep=0.75pt]    {$d$};
% Text Node
\draw (431,129.4) node [anchor=north west][inner sep=0.75pt]    {$r$};
% Text Node
\draw (367,87.4) node [anchor=north west][inner sep=0.75pt]    {$d$};
% Text Node
\draw (474,88.4) node [anchor=north west][inner sep=0.75pt]    {$d$};
% Text Node
\draw (119,135.4) node [anchor=north west][inner sep=0.75pt]    {$A_{1}$};
% Text Node
\draw (222,136.4) node [anchor=north west][inner sep=0.75pt]    {$A_{2}$};
% Text Node
\draw (369,136.4) node [anchor=north west][inner sep=0.75pt]    {$A_{n-1}$};
% Text Node
\draw (477,137.4) node [anchor=north west][inner sep=0.75pt]    {$A_{n}$};

\end{tikzpicture}
    \caption{Tensor network representation of a matrix product state with bond dimension $r$.}
    \label{fig:matrix_product_state}
\end{figure}
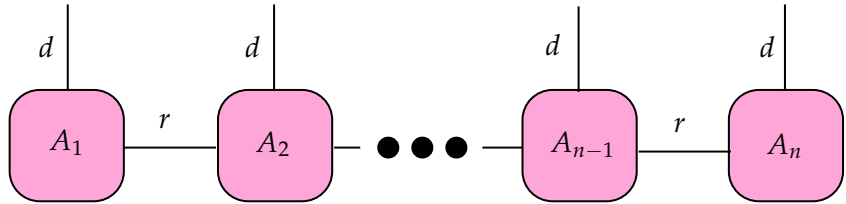

In this section we present an agnostic learning algorithm for matrix product states very closely based on the algorithm from \cite{cpfsgb10}. 
In fact, it is nearly identical to the algorithm presented there, but with different parameters.
For completeness, here we will present a more or less self-contained description of the algorithm, and a proof of its correctness.

Unlike in the product state setting, given a state that is close to a matrix product state, we will not learn a matrix product state with the same bond dimension, but one whose bond dimension is a polynomial-factor higher, hence ``agnostic improper learning''. 
At a very high level, the algorithm of~\cite{cpfsgb10} demonstrates that one can successively learn unitary rotations that act on few qubits at a time that allow you to disentangle qubits successively.
This is because when $\rho$ is a MPS (or very close to one), then it has low Schmidt rank along every cut in the tensor network, and consequently essentially all of the information about $\rho$ is contained in a $r$-dimensional subspace along every cut.
While this is obviously not true when $\rho$ is arbitrary, we show that we can find $\poly(n, d, r, 1/\eps)$-dimensional subspaces which do preserve all of the correlation structure between $\rho$ and any matrix product state.

\begin{theorem}[Agnostic improper learning of matrix product states] \label{thm:mps}
    Suppose we are given copies of an $n$-qudit state $\rho$.
    Then given parameters $\eps, \delta \in (0,1)$ and a bond dimension parameter $r$, we can output a description of a matrix product state $\ket{\hat{\phi}}$ with bond dimension $dn^2 \cdot \poly(r, 1/\eps)$ such that, with probability $\geq 1-\delta$,
    \begin{equation}
        \braket{\hat{\phi} | \rho | \hat{\phi}} \geq \max_{\ket{\phi} \in MPS_{r}} \braket{\phi | \rho | \phi} - \eps.
    \end{equation}
    This algorithm uses $N = \poly(n, d, r, 1/\eps, \log\frac{1}{\delta})$ copies of $\rho$, along with $\poly(N)$ quantum gates and classical overhead.
\end{theorem}

We will use the following notion of Schmidt rank and decomposition.

\begin{definition}
    Let $\ket{\phi}$ be an $n$-qudit state, i.e. a vector in $\left( \C^d\right)^{\otimes n}$.
    For $i = 1, \ldots, n$ we say that the \emph{Schmidt decomposition of $\ket\phi$ at position $i$} is the Schmidt decomposition of $\ket\phi$ when viewed as an element of the bipartite system $A \otimes B$ where $A = \left( \C^d\right)^{\otimes i}$ and $B = \left( \C^d\right)^{\otimes (n - i)}$.
    We correspondingly say that the Schmidt rank of $\ket\phi$ at position $i$ is the rank of this Schmidt decomposition.
\end{definition}

Note that a matrix product state with bond dimension $r$ has Schmidt rank $r$ across all bipartitions.  We will also need the following two operations:

\begin{definition}[Disentangling unitary]
\label{def:disentangle}
    Let $W$ be a subspace of $(\C^{d})^{\otimes \kappa}$ of dimension $d^{\kappa- 1}$.
    We say a unitary matrix $U: (\C^{d})^{\otimes \kappa} \to (\C^{d})^{\otimes \kappa}$ is a disentangling unitary for $W$ if for all $\ket{\phi} \in W$, we have that $U \ket{\phi} = \ket{0} \otimes \ket{\psi}$ for some state $\ket{\psi}$ on the remaining $\kappa - 1$ qudits, and for all $\ket{\phi}$ orthogonal to $W$, we have $\tr\left[(\proj{0} \otimes \id)U \proj{\phi} U^{\dagger}\right] = 0$.
\end{definition}

\begin{lemma}[\cite{iten2016quantum}]
\label{lem:disentangling_implementation}
    Given a classical description of a $d^{\kappa-1}$-dimensional subspace $W$, there is an implementation of the disentangling unitary in time $\poly(d^\kappa)$, when $d$ is a power of $2$.
\end{lemma}

We also need a result concerning the tomography of sub-normalized states in trace distance, inspired by the sub-normalized fidelity-squared tomography algorithm of \cite{fo24}.

\begin{lemma}[Sub-normalized tomography]
\label{lem:subnormalized_tomography}
    Let $\rho$ be a quantum state on $r$ registers of dimension $d$.  Let $\Pi = \proj{0^i} \otimes \id^{\otimes r - i}$ and $\mu = \tr[\Pi \rho]$, and $\delta, \epsilon > 0$.  Then there exists an algorithm that performs tomography of $\sigma = \Pi \rho \Pi$ to within trace distance error $\epsilon$ with failure probability $\delta$ using $O(\mu \cdot \frac{d^{3(r-i)}}{\epsilon^2} \log(1/\delta))$ copies of $\rho$ and in time $\poly(d^{r-i}, 1/\epsilon, \log(1/\delta))$.
\end{lemma}
\begin{proof}
    Similar to \cite{fo24}, the algorithm first takes $m = O(\mu \cdot \frac{d^{3(r-i)}}{\epsilon^2} \log(1/\delta))$ copies of $\rho$, and measures the PVM $\{\Pi, \id - \Pi\}$ on all of them.  It keeps all of the copies that had measurement outcome $\Pi$, and is thus left with $(\rho|_{\Pi})^{\otimes m'}$, where $\rho|_{\Pi} = \Pi \rho \Pi / \mu$ and $m' \sim \mathsf{Binomial}(m, \mu)$.  Then the algorithm then traces out the first $i$ registers (which are in the state $\proj{0^i}$), and performs tomography on the remaining $m'$ states with error $\epsilon / 2\mu$ and failure probability $\delta / 2$.  to produce an estimate $\hat{\rho}|_{\Pi}$.  Since we want a (runtime) efficient algorithm, we use the random Clifford tomography technique described by \cite{KUENG201788}, which uses sample-complexity $O((d')^3 / (\epsilon')^2 \log(1/\delta'))$ (where $d'$ is the dimension of the input to this subroutine, $\epsilon'$ is the error parameter given to this subroutine and $\delta'$ is the failure probability of this subroutine), and runs in polynomial time in all of those parameters.  Substituting $d' = d^{r-i}$ and $\epsilon' = \epsilon / \mu$, this step requires  $m' = O(\mu^2 (d^{r-i})^3 / \epsilon^2)$ copies of the post-selected state. Finally, the algorithm outputs $\hat{\sigma} = (m' / m)\hat{\rho}|_{\Pi}$.

    In order to show that the algorithm produces an estimate with the correct trace norm bound, we can write the following:
    \begin{align*}
        \norm{\sigma - \hat{\sigma}}_1 &= \norm{\mu \rho|_{\Pi} - (m' / m)\hat{\rho}|_{\Pi}}_1\\
        &= \mu \norm{\rho|_{\Pi} - (m' / \mu m) \hat{\rho}_{\Pi}}_1\,.
    \end{align*}
    We can apply the triangle inequality and linearity of expectation to get the following:
    \begin{align*}
        \norm{\rho|_{\Pi} - (m' / \mu m) \hat{\rho}_{\Pi}}_1 &= \norm{\rho|_{\Pi} - \hat{\rho}|_{\Pi}}_1 + \norm{(1 - (m' / (\mu m))) \hat{\rho}_{\Pi}}_1\,.
    \end{align*}
    We want to bound the probability that this value is greater than $\epsilon$.  From the guarantee of the tomography algorithm, the first term is bounded by $\epsilon / 2\mu$ except with probability $\delta / 2$.  From Hoeffding's inequality, as long as $m \geq 2\log(2/\delta) \mu / \epsilon$, the second term is also bounded by $\epsilon / (2\mu)$ with except with probability $\delta / 2$.  Plugging these back into the difference in norm of $\sigma$ and applying a union bound, we get that the trace-norm error is at most $\epsilon$ except with probability $\delta$, as desired.  Since we needed $O(\mu^2 \frac{d^{3(r-i)}}{\epsilon^2} \log(1/\delta))$ copies of the state to perform the tomography, and in expectation $\mu m$ copies survive the measurement of $\Pi$, we need to start with $O(\mu \cdot \frac{d^3}{\epsilon^2} \log(1/\delta))$ copies of $\rho$, as desired. 
\end{proof}

Note that in the remainder of this section, $r-i$ will be roughly $\log_d(\poly)$, where $\poly$ is some polynomial in all parameters.  Thus, the scaling in $d^{r-i}$ will be polynomial in the inputs. 

Our overall algorithm proceeds now as follows, quite similarly to~\cite{cpfsgb10}.
Let $\rho$ be the overall, unknown state.
Let $\tau = \tfrac{\eps^2}{9 n^2 r^4}$, and let $\kappa = \lceil \log_d (1 / \tau) \rceil + 1$.
We will produce a sequence of disentangling unitaries $U_0, U_1, \ldots, U_{n - \kappa + 1}$, where the $j$-th unitary will act only on the $(j)$-th through $(j + \kappa)$-th sites.\footnote{Here we slightly abuse notation and also let $U_i$ denote the extension of the disentangling unitary to the entire space which acts on the identity outside of the aforementioned sites.}
We will also maintain a sequence of intermediate unnormalized states $\rho'_0, \ldots, \rho'_{n-\kappa + 1}$ which we can efficiently prepare given $\rho$ and the $U_i$.
Each state $\rho'_i$ will act on the last $n - i$ qudits.
Initially set $U_0 = I$ and $\rho'_0 = \rho_0 = \rho$.

Then, for all $i \geq 1$, given $\rho'_{i - 1}$, form $\sigma_{i} = \tr_{\geq \kappa} (\rho'_{i - 1})$, and perform state tomography to obtain the classical description of a state $\widehat{\sigma}_i$ satisfying 
\begin{equation}
\label{eq:mps-tomography}
\norm{\widehat{\sigma}_i - \sigma_i}_1 \leq \tau
\end{equation}
with probability $1 - \delta / n$.
Let $W_i$ denote the subspace spanned by the singular values of $\widehat{\sigma}_i$ that exceed $\tau$.
Note that since $\widehat{\sigma}$ has trace $1$, the subspace $W$ has dimension at most $1 / \tau < d^{\kappa  - 1}$.
Extend this subspace arbitrarily to have dimension $d^{\kappa - 1}$, and let $U_i$ be a disentangling unitary for this subspace.
Finally, let $\rho_i'$ be the result of projecting the first qudit of $U_i \rho'_{i - 1} U_i^\dagger$ onto $\proj{0}$.  

After we have produced these disentangling unitaries $U_1, \ldots, U_{n - \kappa + 1}$, form $\rho'_{n - \kappa + 1}$ and $\widehat{\sigma}_{n - \kappa + 1}$ as before, and let $\ket{\psi}$ denote the top eigenvector of $\widehat{\sigma}_{n - \kappa + 1}$.
Our final estimate of $\rho$ is the state $\ket{\hat{\phi}}$ given by
\[
\ket{\hat{\phi}} = (U_1 \ldots U_{n - \kappa + 1}) \left( \ket{0^{n - \kappa}} \otimes \ket{\psi} \right) \; .
\]
The formal pseudocode for this algorithm is given in \cref{alg:mps_learning}.

\begin{lemma}
    $\ket{\hat{\phi}}$ is a matrix product state with bond dimension $d n^2 \poly(r, 1/\epsilon)$. 
\end{lemma}
\begin{proof}
    We can proceed by writing out the entries of the state $\ket{\hat{\phi}}$.  Letting $\kappa$ be the number of qudits each unitary (and $\ket{\psi}$ acts on), we can write $\ket{\psi} = \sum_{t_{n-\kappa+1}, \ldots, t_{n-1}, s_{n}} \alpha_{t_{n-\kappa+1}, \ldots, t_{n-1}, s_{n}} \ket{t_{n-\kappa+1}, \ldots, t_{n-1}, s_{n}}$. We can also write the unitary $U_{n-\kappa +1}$ as follows
    \begin{equation*}
        U_{n - \kappa +1} = \sum_{\substack{s_{n-\kappa}, \ldots s_{n-1} \\ t_{n-\kappa}, \ldots, t_{n-1}}} \left(U_{n-\kappa+1}\right)^{(s_{n-\kappa}, \ldots, s_{n-1})}_{(t_{n-\kappa}, \ldots, t_{n-1})} \ket{s_{n-\kappa}, \ldots, s_{n-1}}\!\!\bra{t_{n-\kappa}, \ldots, t_{n-1}}\,.
    \end{equation*}
    Writing the product of these, we have the following state after applying a single unitary
    \begin{equation*}
        U_{n-\kappa+1} \left(\ket{0}\ket{\psi}\right) = \sum_{s_{n-\kappa}, \ldots, s_{n}} \sum_{t_{n-\kappa - 1}, \ldots, t_{n-1}} \left(\left(U_{n-\kappa +1}\right)^{(s_{n-\kappa}, \ldots, s_{n-1})}_{(0, t_{n-\kappa+1}, \ldots, t_{n-1})}\alpha_{t_{n-\kappa + 1}, \ldots, t_{n-1}, s_n}\right)\ket{s_{n-\kappa}, \ldots, s_{n}}
    \end{equation*}
    Then, writing down the $3$-tensor with entries 
    \begin{align*}
        \left(A_1^{(s_n)}\right)_{1, (t_{n-\kappa}, \ldots, t_{n-1})} &= \alpha_{t_{n-\kappa +1}, \ldots, t_{n-1}, s_n}\\
        \left(A_{2}^{(s_{n-\kappa}, \ldots, s_{n-1})}\right)_{(t_{n-\kappa+1}, \ldots, t_{n-1})} &= \left(U_{n-\kappa +1}\right)^{(s_{n-\kappa}, \ldots, s_{n-1})}_{(0, t_{n-\kappa+1}, \ldots, t_{n-1})}\,,
        \end{align*}
        we find that this state is exactly a matrix product state, where the bond dimension is equal to the number of entries in the inner sum, or $d^{\kappa - 1} \leq \frac{9d r^4n^2}{\epsilon^2}$ (note that taking the ceiling in the definition of $\kappa$ causes us to incur an additional factor of $d$ here).  Also note that the two tensors are indexed by disjoint registers, as desired.  Applying this idea recursively to get the entries of every tensor $A_{i}$, we find that we can express every tensor for, $i \geq 2$, as follows
        \begin{equation*}
            \left(A_{i}^{(s_{n+1-i})}\right)^{(a_{n-\kappa+2-i}, \ldots, a_{n-i})}_{(b_{n-\kappa+3-i}, \ldots, b_{n+1-i})} = \left(U_{n-\kappa+3-i}\right)^{(a_{n-\kappa+2-i}, \ldots, a_{n-i}, s_{n+1-i})}_{(0, b_{n-\kappa+3-i}, \ldots, b_{n+1-i})}\,.
        \end{equation*}
        Writing out the matrix product state that results from these tensors, we will get $\ket{\hat{\phi}}$, and this is a matrix product state with open boundary condition, with bond dimension $d n^2 \poly(r, 1/\epsilon)$ since the dimensions of both the $a$'s and $b$'s is $d^{\kappa - 1}$.  Note that we can reverse the order of the $s_i$ to write this in the canonical form from \Cref{def:mps}.  This completes the proof that we have produced a matrix product state with the desired bond dimension.
\end{proof}

\begin{longfbox}[breakable=false, padding=1em, margin-top=1em, margin-bottom=1em]
\begin{algorithm}[Agnostic learning of matrix product states]
\label{alg:mps_learning}\hfill
\begin{description}
    \item[Input:] Copies of an unknown quantum state $\rho \in \reg{R}_1 \otimes \ldots \otimes \reg{R}_n$ where each $\reg{R}_i$ is dimension $d$, such that there exists matrix product state with bond dimension at most $r$ with fidelity $\eta$ with $\rho$, and error parameter $\epsilon$ and failure probability $\delta$.  
    \item[Output:] A description of a MPS
    \item[Procedure:] \mbox{}
    \begin{algorithmic}[1]
        \State Let $\tau = \frac{\epsilon^2}{9n^2r^4}$;
        \State Let $\kappa = \lceil \log_d(1/\tau)\rceil + 1$;
        \State Let $\rho'_0 = \rho$;
        \For{$i$ from $1$ to $n-\kappa$}
            \State Let $\sigma_i = \tr_{\geq \kappa}\left(\rho_{i-1}'\right)$;
            \State Let $\hat{\sigma}$ be the output of $\tomography$ with error $\tau$ and failure probability $\delta / n$ on $O\left(\frac{d^3}{\tau^5} \log(n/\delta)\right)$ copies of $\sigma_i$;
            \State Let $U_{i}$ be the disentangling unitary for the extension of the $\geq \tau$ subspace of $\hat{\sigma}_i$;
            \State Apply $U_{i}$ to all copies of $\rho'_{i-1}$ and project onto $\proj{0}$ to get copies of
            \begin{equation*}
                \rho'_i = \tr_{\reg{R}_{i-1}}\left((\proj{0} \otimes \id) U_{i} \rho'_{i-1} U_{i}^{\dagger}(\proj{0} \otimes \id)\right);
                \vspace{1em}
            \end{equation*}
        \EndFor
        \State Let $\rho^*$ be the remaining state on the final $\kappa$ qubits after applying all of the disentangling unitaries and projecting onto $\ket{0^{n-\kappa}}$:
        \begin{equation*}
            \rho^* = \Tr_{\reg{R}_{< n - \kappa}}((\proj{0^{n-\kappa}} \otimes \id) U_{n - \kappa}\ldots U_{1} \rho U_{1}^{\dagger}\ldots U_{n - \kappa}^{\dagger} (\proj{0^{n-\kappa}} \otimes \id));
            \vspace{1em}
        \end{equation*}
        \State Let $\hat{\rho}^*$ be the output of $\tomography$ on with error parameter $\tau$ and failure probability $\delta / n$ on $O(\frac{d^3}{\tau^5} \log(n/\delta))$ copies of $\rho^*$.  
        \State Let $\ket{\psi}$ be the top eigenvalue of $\hat{\rho}^*$;
        \State \Return the MPS $U_{1}^{\dagger}\ldots U_{n - \kappa}^{\dagger} (\ket{0^{n - \kappa}} \otimes \ket{\psi})$.
    \end{algorithmic}
\end{description}
\end{algorithm}
\end{longfbox}

We now proceed with the proof of correctness for this algorithm. 

\begin{lemma}
\label{lem:spectral-bound}
    Let $\rho, \sigma$ be unnormalized mixed states satisfying $\norm{\rho - \sigma}_1 \leq \eta$.
    Let $W$ be the span of all eigenvectors of $\sigma$ with eigenvalues exceeding $\eta$.
    Then, letting $\Pi_W$ denote orthogonal projection onto $W$, we have that $\norm{(I - \Pi_W)^\dagger \rho (I - \Pi_W)}_\infty \leq 2 \eta$.
\end{lemma}
\begin{proof}
Suppose not, i.e. there exists some state $\ket{\phi}$ so that $\ket{\phi}$ is orthogonal to $W$ and so that $\braket{\phi | \rho | \phi} > 2 \eta$.
But then by duality, we have
\[
\norm{\rho - \sigma}_1 \geq \tr ((\rho - \sigma) \ketbra{\phi}{\phi}) > \eta \; ,
\]
which is a contradiction.
\end{proof}

We will prove the correctness of this algorithm inductively.
For all $i = 1, \ldots, n - \kappa + 1$, define the matrices
\begin{equation}
\label{eq:induction}
E_i = U_i \cdot \ldots \cdot U_1 \; \mbox{, and } \; \rho_i = E_i^\dagger \left( \proj{0^i} \otimes \rho'_i \right) E_i \; .
\end{equation}
Note that by construction, $U_i$ only touches qudits $i$ through $i + \kappa$.
As an immediate consequence of this, $E_i$ acts nontrivially only on the first $i + \kappa$ qudits.

We will say that a call to $\tomography$ \emph{succeeds} if it outputs a $\widehat{\sigma}_i$ satisfying Equation~\eqref{eq:mps-tomography}.
By a union bound, since we do at most $n$ calls to $\tomography$, all calls succeed simultaneously with probability at least $1 - \delta$.
For the rest of the section, condition on the event that this occurs.
We first observe the following:
\begin{lemma}
\label{lem:mps-psd-ordering}
    For all $i = 1, \ldots, n - \kappa + 1$, we have that $\rho_{i} \preceq \rho_{i - 1}$.
\end{lemma}
\begin{proof}
By properties of post-selection, we have that $\ketbra{0}{0} \otimes \rho_i' \preceq U_{i}\rho_{i - 1}'U^{\dagger}_{i}$.
The claim then follows from unraveling the definitions.
\end{proof}

Our main claim is the following:
\begin{lemma}
\label{lem:mps-induction}
    Fix $i \in \{1, \ldots, n - \kappa + 1\}$.
    Let $E_i$ and $\rho_i$ be defined as above, and suppose that every step of $\tomography$ succeeds.
    Suppose that $\ket{\phi}$ has Schmidt rank at most $r$ at position $i + \kappa$.
    Then,
    \[
    \left| \braket{\phi | \rho_{i - 1} | \phi} - \braket{\phi | \rho_i | \phi} \right| \leq \frac{\eps}{2 n} \; .
    \]
\end{lemma}
\begin{proof}
    Let $\ket{\phi'} = E_{i - 1} \ket{\phi}$.
    Since $E_{i - 1}$ only acts non-trivially on the first $i + \kappa - 1$ qudits, $\ket{\phi'}$ has the same Schmidt rank as $\ket\phi$ at position $i + \kappa$, so in particular it has Schmidt rank at most $r$ at position $i + \kappa$.
    Write $\ket{\phi'} = \ket{0^{k - 1}}\ket{\psi} + \ket{\bot}$, where $\ket{\bot}$ is orthogonal to all states beginning with $\ket{0^{k - 1}}$.
    By definition, $\ket{\psi}$ has Schmidt rank at most $r$ at position $\kappa$.
    By definition, we have that
    \[
    \braket{\phi | \rho_{i - 1} | \phi} = \braket{\psi | \rho'_{i - 1} |\psi} \; ,
    \]
    and additionally, we have that
    \begin{align*}
        \braket{\phi | \rho_i | \phi} &= \bra{\phi'} U_i^\dagger \left( \ketbra{0^i}{0^i} \otimes  \rho'_i \right) U_i \ket{\phi'} \\
        &= \bra{\psi} U_i^\dagger \left( \ketbra{0}{0} \otimes \rho'_i \right) U_i \ket{\psi} \\
        &= \bra{\psi} (\Pi_W \otimes I)^\dagger \rho'_{i - 1} (\Pi_W \otimes I) \ket{\psi} \; ,
    \end{align*}
    where in the second line we use the fact that the extension of $U_i$ acts as the identity outside of qudits $i$ through $i + \kappa$, the third line follows because $\rho'_i$ is obtained by postselecting on outcome $\ket{0}$, and the last line follows since $\tomography$ succeeds, and Lemmas~\ref{lem:spectral-bound} and~\ref{lem:partial-trace-eig-bound}, the latter of which is proven below.
\end{proof}

\begin{lemma}
\label{lem:partial-trace-eig-bound}
    Let $A, B$ be two Hilbert spaces, and let $\rho$ be a density matrix over $A \otimes B$.
    Let $\ket{\phi}$ be a pure state in $A \otimes B$ with Schmidt rank at most $r$.
    Let $W \subset A$ be a subspace, with $\Pi_W$ denoting the projection onto $W$, such that $\norm{(I - \Pi_W)^\dagger \tr_B (\rho) (I - \Pi_W)}_\infty \leq \eta$.
    Then
    \[
    \left| \bra{\phi} \left(\Pi_W \otimes I \right)^\dagger \rho \left( \Pi_W \otimes I \right) \ket{\phi} - \braket{\phi | \rho | \phi} \right| \leq2  r \sqrt{\eta} \; .
    \]
\end{lemma}
\begin{proof}
    First, assume that $\ket{\phi}$ has Schmidt rank $1$, i.e. $\ket\phi = \ket{a} \ket{b}$.
    Let $\ket{c} = \ket{a - \Pi_W a}$. 
    Then
    \begin{align*}
    \left| \bra{\phi} \left( \Pi_W \otimes I \right)^\dagger \rho \left( \Pi_W \otimes I \right) \ket{\phi} - \braket{\phi | \rho | \phi} \right| &= \left| \bra{b} \bra{\Pi_W a} \rho \ket{\Pi_W a} \ket{b} - \bra{b}\bra{a} \rho \ket{a} \ket{b} \right| \\
    &\leq  \left| \bra{b} \bra{c} \rho \ket{\Pi_W a} \ket{b} \right| +  \left| \bra{b} \bra{a} \rho \ket{c} \ket{b} \right| \\
    &\leq 2 \norm{\rho \ket{c} \ket{b}} \; ,
    \end{align*}
    where the last line follows from Cauchy-Schwarz.
    To finish, we observe that
    \begin{align*}
        \norm{\rho \ket{c} \ket{b}}^2 &= \tr (\rho^2 \ketbra{c}{c} \otimes \ketbra{b}{b}) \\
        &\leq \tr (\rho \ketbra{c}{c} \otimes \ketbra{b}{b}) \\
        &\leq \braket{c | \tr_B (\rho) | c} \\
        &\leq \eta \; ,
    \end{align*}
    since $c$ is orthogonal to $W$.
    The case of general Schmidt rank then follows from linearity.
\end{proof}

\begin{proof}[Proof of Theorem~\ref{thm:mps}]
    We first concern ourselves with the sample complexity and runtime of \cref{alg:mps_learning}. From \Cref{lem:subnormalized_tomography}, taking the dimension to be $d^{\kappa} \leq d/\tau$, error parameter to be $\tau$ and failure probability $\delta / n$, we can perform tomography on the sub-normalized state $\sigma_i$ with sample complexity $O(\frac{d^3}{\tau^5}\log(n/\delta)) = O\left(\frac{d^{3}n^{10}r^{10}}{\epsilon^{10}} \log(n / \delta)\right) = \poly(n, d, r, 1/\epsilon, \log(1/\delta))$ (note that the algorithm takes samples of a unitary applied to $\rho$, the original state), and runtime that is polynomial in the same parameters.  The overall algorithm performs this $n-\kappa$ times. In addition to tomography, the algorithm also applies disentangling unitaries to the state.  From \Cref{lem:disentangling_implementation}, given a description of a subspace $W$ of dimension $d^{\kappa - 1}$, there is an implementation of the disentangling unitary in time $\poly(d^{\kappa}) = \poly(d, r, n, \epsilon)$, when $d$ is a power of $2$.  Thus, the whole algorithm run in polynomial time and uses a polynomial number of samples of $\rho$.

    We now turn our attention to correctness.
    Let $\ket{\phi}$ be a matrix product state with bond dimension $r$.
    Since such a MPS has Schmidt rank at most $r$ at every position, by iteratively applying Lemma~\ref{lem:mps-induction} with our chosen parameters, we obtain that $|\braket{\phi | \rho | \phi} - \braket{\phi | \rho_{i - \kappa + 1} | \phi}| \leq \eps / 2$.
    Together with the fact that the last call to $\tomography$ succeeds, this implies that the vector $\ket{\psi}$ satisfies $\braket{\psi | \rho'_{n -\kappa + 1} | \psi} \geq \braket{\phi | \rho | \phi} - \eps$.
    To conclude, we apply Lemma~\ref{lem:mps-psd-ordering} and the previous bound to obtain that
    \begin{align*}
        \braket{\hat{\phi} | \rho | \hat{\phi}} &= \braket{\hat{\phi} | \rho_0 | \hat{\phi}} \\
        &\geq \braket{\hat{\phi} | \rho_{n -\kappa + 1} | \hat{\phi}} \\
        &\geq \braket{\phi | \rho | \phi} - \eps.
    \end{align*}
    The result follows by taking a supremum over all matrix product states. 
\end{proof}

\end{document}